\documentclass[a4paper,11pt]{amsart}

\usepackage{amsmath,amssymb,amsthm,graphicx,stmaryrd}
\usepackage{verbatim,enumitem,xcolor}
\usepackage{hyperref}
\usepackage[enableskew]{youngtab}
\usepackage{tikz,tikz-cd,genyoungtabtikz,pgfplots}
\usepackage{caption,subcaption}
\usetikzlibrary{patterns}
\usetikzlibrary{patterns.meta}
\usetikzlibrary{arrows.meta}
\usetikzlibrary{decorations.markings}

 	\definecolor{myorange}{rgb}{1,0.5,0}
 	 	\definecolor{mygreen}{rgb}{0,1,0}
 	 	\definecolor{yellowgreen}{rgb}{0.5,1,0}
 	 	\definecolor{bluegreen}{rgb}{0,1,1}
 	 	\definecolor{bluebluered}{rgb}{0.5,0,1}
 	 	\definecolor{bluered}{rgb}{1,0,1}
 	 	\definecolor{blueredred}{rgb}{1,0,0.5}
 	 	
\pgfdeclarepattern{
name=hatch,
parameters={\hatchsize,\hatchangle,\hatchlinewidth},
bottom left={\pgfpoint{-.1pt}{-.1pt}},
top right={\pgfpoint{\hatchsize+.1pt}{\hatchsize+.1pt}},
tile size={\pgfpoint{\hatchsize}{\hatchsize}},
tile transformation={\pgftransformrotate{\hatchangle}},
code={
\pgfsetlinewidth{\hatchlinewidth}
\pgfpathmoveto{\pgfpoint{-.1pt}{-.1pt}}
\pgfpathlineto{\pgfpoint{\hatchsize+.1pt}{\hatchsize+.1pt}}
\pgfpathmoveto{\pgfpoint{-.1pt}{\hatchsize+.1pt}}
\pgfpathlineto{\pgfpoint{\hatchsize+.1pt}{-.1pt}}
\pgfusepath{stroke}
}
}

\tikzset{
hatch size/.store in=\hatchsize,
hatch angle/.store in=\hatchangle,
hatch line width/.store in=\hatchlinewidth,
hatch size=5pt,
hatch angle=0pt,
hatch line width=.5pt,
}

\renewcommand{\AA}{\mathbb{A}} 
\newcommand{\BB}{\mathbb{B}} 
\newcommand{\CC}{\mathbb{C}}

\newcommand{\NN}{\mathbb{N}}
\newcommand{\PP}{\mathbb{P}} 
\newcommand{\RR}{\mathbb{R}} 

\renewcommand{\SS}{\mathbb{S}}
\newcommand{\TT}{\mathbb{T}} 
\newcommand{\VV}{\mathbb{V}}

\newcommand{\ZZ}{\mathbb{Z}}

\newcommand{\cE}{\mathcal{E}}

\newcommand{\cH}{\mathcal{H}}

\newcommand{\cO}{\mathcal{O}}

\renewcommand{\a}{\alpha}
\newcommand{\D}{\Delta} 
\renewcommand{\d}{\delta} 
\newcommand{\G}{\Gamma}

\newcommand{\g}{\gamma}

\newcommand{\la}{\lambda}
\renewcommand{\b}{\beta} 
 
\newcommand{\Om}{\Omega}
\newcommand{\om}{\omega} 
 
\newcommand{\s}{\sigma}

\newcommand{\eps}{\varepsilon}

\newcommand{\tr}{\mathrm{tr}}

\renewcommand{\b}{\beta}
\newcommand{\oo}{\infty}

\newcommand{\sm}{\setminus}

\newcommand{\es}{\varnothing}
\newcommand{\se}{\subseteq}

\newcommand{\crit}{\mathrm{c}}

\newcommand{\GL}{\mathrm{GL}}

\newcommand{\ct}{\mathrm{ct}}

\newcommand{\one}{\hbox{\rm 1\kern-.27em I}}

\newcommand{\End}{\text{End}}

\newcommand{\Id}{\mathtt{Id}}

\newcommand{\ab}{\textsc{ab}}

\newcommand{\wb}{\textsc{wb}}
\newcommand{\mb}{\textsc{mb}}

\newcommand{\be}{\begin{equation}}
\newcommand{\ee}{\end{equation}}

\newtheoremstyle{slthm}
{}
{\baselineskip}
{\slshape}
{\parindent}
{\scshape}
{.}
{ }
{}

\theoremstyle{slthm}
\newtheorem{definition}{Definition}[section]
\newtheorem{theorem}[definition]{Theorem}
\newtheorem{proposition}[definition]{Proposition}
\newtheorem{lemma}[definition]{Lemma}

\title
{%
	Heisenberg models and Schur--Weyl duality
}
\author{J. E. Bj\"ornberg}\thanks{JEB: Chalmers University of Technology and University of
  Gothenburg, Sweden. jakob.bjornberg@gu.se}
\author{H. Rosengren}\thanks{HR: Chalmers University of Technology and University of Gothenburg,
    Sweden. hjalmar@chalmers.se}
\author{K. Ryan}\thanks{KR: University of Vienna, Austria. kieran.ryan@univie.ac.at}
\date{\today}

\begin{document}
	
\begin{abstract}
  We present a detailed analysis of certain  quantum spin
  systems with inhomogeneous (non-random) mean-field
  interactions.
  Examples include, but are not limited to, the interchange- and spin
  singlet projection interactions on complete bipartite graphs.  Using
  two instances of the 
representation theoretic framework of Schur--Weyl duality, we
  can explicitly compute the free energy and other thermodynamic
  limits in the models we consider.  This allows us  to
  describe the phase-transition,
  the ground-state phase diagram, and
  the expected structure of extremal
  states.
	\end{abstract}
	
	\maketitle
	
		\tableofcontents
	
	\section{Introduction and results}

	When Werner Heisenberg in 1928 introduced his famous model for
	ferromagnetism, he described it in terms of an \emph{exchange
		interaction} between neighbouring valence electrons (``Austausch von
	Elektronen'', 
	\cite[p. 621]{heisenberg}).  In modern notation, for the
	spin-$\tfrac12$ system he was considering, this interaction can be
	written as $T_{i,j}=2(\SS_i\cdot\SS_j)+\tfrac12$, where $T_{i,j}$ acts
	on a pure tensor $v_i\otimes v_j$ in $\CC^2\otimes\CC^2$ by
	transposing the factors, $T_{i,j}(v_i\otimes v_j)=v_j\otimes v_i$,
	and $\SS=(S^{(1)},S^{(2)},S^{(3)})$ are spin-$\tfrac12$-matrices.  Two natural
	generalisations to higher spin immediately suggest themselves:  we can
	take the interaction to be the transposition $T_{i,j}$ acting on
	$\CC^r\otimes \CC^r$, or to be 
	$\SS_i\cdot\SS_j$, where the $\SS$ are now spin-$S$-matrices and
	$r=2S+1$.  For $S>\tfrac12$, these choices are no longer
	equivalent; while both are natural generalisations, some
	authors 
	reserve the name Heisenberg model for the 
	model with interaction
	$\SS_i\cdot\SS_j$.  The model with interaction $T_{i,j}$ has been
	called the \emph{interchange model} and is one of the main topics of
	this paper.
	
	The name \emph{interchange model} can be traced back 
	to works by Harris
	\cite{harris}, Powers \cite{powers}, and T\'oth \cite{toth-93}, and is
	motivated by a probabilistic representation of the model.  Powers
	\cite{powers} was first to notice that the ferromagnetic 
	(spin-$\tfrac12$) Heisenberg model can be represented in
	terms of a random walk on permutations generated by transpositions.
	The latter random walk was  constructed on infinite lattices by Harris
	\cite{harris}.  T\'oth
	\cite{toth-93} was first to use this representation to obtain an
	important result for the Heisenberg model:  
	a bound on the free energy of the model on
	$\ZZ^3$ that was the best known for many years \cite{seiringer}. 
	The underlying random walk on permutations has come to be known as the
	interchange \emph{process} in the literature on mixing times of
	Markov chains \cite{aldous-fill}.
	The present paper does not use the probabilistic representation,
	however; indeed our methods apply also in cases where such a
	representation is not available.
	
	For the \emph{antiferromagnetic} spin-$\tfrac12$ Heisenberg model,
	Aizenman and Nachtergaele \cite{an} discovered a similar
	probabilistic representation based on the identity 
	$P_{i,j}=\tfrac12-2\SS_i\cdot\SS_j$ where $P_{i,j}$ is (twice) the projection
	onto the singlet subspace of $\CC^2\otimes\CC^2$ (eigenspace for the
	total spin operator with eigenvalue 0).  On a bipartite graph, such as
	the line $\ZZ$ considered by Aizenman and Nachtergaele, the Hamiltonian with interactions $P_{i,j}$ is
	unitarily equivalent to that with interactions $Q_{i,j}$ defined by 
	\be\label{eq:Q}
	\langle e_{\a_1}\otimes e_{\a_2}|
	Q_{i,j}
	|e_{\a_3}\otimes e_{\a_4}\rangle=
	\d_{\a_1,\a_2} \d_{\a_3,\a_4},
	\ee
	where the $e_\alpha$ are a basis for $\CC^2$.
	The interaction $Q_{i,j}$ has a natural interpretation in terms of
	random loops, and
	plays a central role in the present
	work.  The definition \eqref{eq:Q} generalises straightforwardly to
	higher spin.  
	
	If we take the underlying lattice to be the \emph{complete graph}
	$K_n$, consisting of $n$ vertices with an edge between each pair of
	distinct vertices, then the interchange model is a mean-field system
	with Hamiltonian
	\be\label{eq:intch}
	-\frac1n\sum_{1\leq i<j\leq n} T_{i,j},
	\quad\mbox{acting on } (\CC^r)^{\otimes n},\quad  r\geq 2.
	\ee
	This model was studied in the papers of Bj\"ornberg \cite{Bjo16,BFU20}, where the key
	step of the analysis was to note that the Hamiltonian \eqref{eq:intch} is a central element of
	the group algebra $\CC[S_n]$ of the symmetric group, represented on
	the tensor space $(\CC^r)^{\otimes n}$.  This means that the
	eigenspace decomposition for the Hamiltonian \eqref{eq:intch} coincides with the decomposition of 
	$(\CC^r)^{\otimes n}$ into irreducible $S_n$-modules, which is
	well-studied.  Ryan \cite{Ryan21} implemented a similar approach for the
	model with Hamiltonian
	\be\label{eq:ryan}
	-\frac1n\sum_{1\leq i<j\leq n} (a \,T_{i,j}+b \,Q_{i,j})
	\quad\mbox{acting on } (\CC^r)^{\otimes n},
	\ee
	with $a,b\in\RR$ and $r\ge2$, which can similarly be diagonalised using the irreducible
	representations of the \emph{Brauer algebra} (defined below).  
	
	The unifying principle behind this approach to determining the eigenspace decomposition of the Hamiltonian 
	is a classical algebraic theory called \emph{Schur--Weyl duality}.  
	This term is used for specific instances
	of a general result in representation theory called the \emph{double
		centraliser theorem}, which states the following 
	\cite[Theorem~4.54]{etingof}.
	Let $\VV$ be a finite-dimensional vector space, and 
	$\AA\se \End(\VV)$ a semi-simple algebra of linear mappings
	(endomorphisms)
	$\VV\to\VV$.  Then the centraliser $\BB$ of $\AA$, i.e.\ the algebra
	of endomorphisms commuting with all elements of $\AA$, is also
	semi-simple, and as a representation of $\AA\otimes\BB$ we have
	\be\label{eq:dct}
	\VV=\bigoplus_i U_i\otimes V_i,
	\ee
	where the $U_i$ (respectively $V_i$) are non-isomorphic
	irreducible representations of $\AA$ (respectively $\BB$).
	The most famous instances of this (and relevant in the
	present work) are obtained by letting $\VV=(\CC^r)^{\otimes n}$.
	If we let $\AA$ consist of all invertible endomorphisms of $\CC^r$,
	acting diagonally on $\VV$, then $\BB$ is generated by the
	permutations of the tensor factors of $\VV$:  this gives the
	Schur--Weyl duality between the general linear group $\GL_r(\CC)$ and
	the symmetric group $S_n$ (see \eqref{eq:swd} for details) which
	facilitates the analysis of the interchange model \eqref{eq:intch}.   If
	instead we take $\AA$ to consist of \emph{orthogonal} matrices, then
	$\BB$ is the Brauer-algebra used in the analysis of \eqref{eq:ryan}.  
	
	Let us note that the present work follows a line of papers analysing the
	interchange process and Heisenberg model with algebraic methods
	(including the aforementioned \cite{Bjo16}, \cite{BFU20},
	\cite{Ryan21}). Alon and Kozma \cite{A-K-13} analysed the interchange
	process on a general graph, and estimated the number of $k$-cycles at
	a given time; Berestycki and Kozma \cite{B-K-15} gave an exact formula
	for the same on the complete graph; Alon and Kozma \cite{A-K-18} gave
	an exact formula for the magnetisation of the mean-field
	spin-$\frac12$ Heisenberg model.

	In this work we carry the methods described above further, 
	to  inhomogeneous
	models on the complete graph where the coupling constants
	between different vertices take finitely many different values.
	The models for which
	our analysis goes the deepest are what we call \emph{two-block} 
	models, where coupling constants can take at most three values
        (one each for the interactions within each of the two blocks,
        and one for interactions between the two blocks). 
	Our results on these models come in several parts.
        In Theorems \ref{thm:FE-AB} and \ref{thm:FE-WB} 
we explicitly compute the
        free energy.  In Propositions \ref{prop:crittemp} to
        \ref{prop:unique}, we give results on phase transitions, and,
        for certain restrictions on the parameters, we compute the
        critical temperature.  In Theorems
        \ref{thm:TS-twoblock} and \ref{thm:mag} we compute a
        magnetisation and limits of certain 
	correlation functions.   
Using the results mentioned above, in Section \ref{ssec:gs}
we completely describe the gound-state phase diagram of the models;
and in Section \ref{sec:heuristics} we give heuristic descriptions
of the extremal Gibbs states and phase driagrams at finite temperature.
At the end of the paper, in Section \ref{sec:MB},
we give the free energy for what we call
\emph{multi-block} models, where
coupling constants can take any finite number of values, and
where we allow certain many-body interactions.

	Two highlights 
	of the new results in this paper 
	are the following. Firstly, we give a formula for the critical temperature 
	of the spin-$\tfrac{1}{2}$ quantum Heisenberg model on the complete 
	bipartite graph; see Proposition \ref{prop:r2crittemp} 
	with $a=b=0$. Secondly, a 
	curious equality of the free energy of the model on the complete 
	bipartite graph with interaction via transpositions $T_{i,j}$ 
	\eqref{eq:intch}, and the model with interaction via the (scaled) spin-singlet 
	projection $P_{i,j}$;
	see Theorem \ref{thm:FE-WB},
	also with $a=b=0$. We wonder 
	whether this equality holds for arbitrary bipartite graphs.

	\subsection{Free energy}

	For $a,b,c\in\RR$, and $1\le m\le n$, we define the
	\emph{{\ab}-interchange-model}, or
	\emph{{\ab}-model} for short, through its  Hamiltonian 
	\be\label{eq:H-AB}
	H_n^\ab=-\frac1n\Big(
	a\sum_{1\leq i<j\leq m} T_{i,j}+
	b\sum_{m+1\leq i<j\leq n} T_{i,j}+
	c\sum_{1\leq i\leq m<j\leq n} T_{i,j}
	\Big),
	\ee
	acting on $\VV=(\CC^r)^{\otimes n}$.
	For $\b>0$, introduce the partition function 
	$Z_{n}^\ab(\b)=\tr\big[e^{-\beta H^\ab_n}\big].$
	We call this a \emph{two-block} model since we may
	think of it as a 
	spin system on a graph with vertex set $\{1,2,\dotsc,n\}$
	partitioned into the two blocks $A=\{1,\dotsc,m\}$
	and $B=\{m+1,\dotsc,n\}$.  The form of the Hamiltonian \eqref{eq:H-AB}
	means that spins at two vertices within $A$ interact with coupling
	constant $a$, spins at two vertices within $B$ interact with coupling
	constant $b$, and the spin at a vertex in $A$ interacts with the spin at
	a vertex in $B$ with coupling constant $c$.
	In the homogeneous case $a=b=c$ we obtain the interchange model on
	the complete graph \eqref{eq:intch}, while if 
	$a=b=0$ and $c\neq 0$  we obtain a model on the
	complete bipartite graph $K_{m,n-m}$.
	
	We write 
	\be\label{eq:F}
	F(x_1,\dotsc,x_r;y_1,\dotsc,y_r)= \textstyle\sum_{i=1}^r f(x_i,y_i),
	\ee
	where $x_i,y_i\geq 0$ and
	\be
	f(x,y)=-x\log x-y\log y+
	\tfrac\b2\big(a x^2+by^2+2cxy\big).
	\ee 
	We have the following result about the free energy:
	
	\begin{theorem}\label{thm:FE-AB}
		Let $a,b,c\in\RR$ be fixed.
		If $n,m\to\oo$ such that $m/n\to\rho\in(0,1)$, then the free energy
		of the model \eqref{eq:H-AB} satisfies 
		\be\label{eq:fe-max-AB}
		\Phi^\ab_\b(a,b,c):=\lim_{n\to\oo} \tfrac1n\log Z_{n}^\ab(\b)
		=\max\; F(x_1,\dotsc,x_r;y_1,\dotsc,y_r)
		\ee
		where the maximum is taken over $x_1,\dotsc,x_r,y_1,\dotsc,y_r\geq 0$ 
		subject to $\sum_{i=1}^rx_i=1-\sum_{i=1}^ry_i=\rho$.
	\end{theorem}
	
	Note that if  $(x_1,\dotsc,x_r;y_1,\dotsc,y_r)$ 
	is a maximum point of $F$, and
	we order the $x$-entries so that 
	\begin{equation}\label{eq:xgeq}
	x_1\geq x_2\geq \dots\geq x_r,
	\end{equation}
	then for $c>0$ we necessarily have 
	$y_1\geq\dotsb\geq y_r$,
	while for $c<0$ we necessarily have 
	$y_1\leq\dotsb\leq y_r$.
	Indeed, the only term in $F$ which is dependent on the relative order
	of the entries  is the term
	$\sum_{i=1}^r  x_i y_i$, which is indeed maximised when the
	orders are the same and minimised if they are reversed.\\

	We next consider another two-block model but where the
	interaction ``between'' the blocks uses the operator $Q$
	defined in \eqref{eq:Q}.  We let
	\be\label{eq:H-wb}
	H_n^\wb=-\frac1n\Big(
	a\sum_{1\leq i<j\leq m} T_{i,j}+
	b\sum_{m+1\leq i<j\leq n} T_{i,j}+
	c\sum_{1\leq i\leq m<j\leq n} Q_{i,j}
	\Big).
	\ee
	Also let $Z^\wb_{n}(\b)=\tr[e^{-\b H^\wb_n}]$. Let us note
	here that for all $r\ge2$, this model is unitarily equivalent
	to the same model with each $Q_{i,j}$ replaced with $P_{i,j}$,
	the latter being ($r$ times) the projection onto the singlet state:
	\be\label{eq:P}
	\langle e_{\a_1}\otimes e_{\a_2}|
	P_{i,j}
	|e_{\a_3}\otimes e_{\a_4}\rangle=
	(-1)^{\a_1-\a_3}\d_{\a_1,-\a_2} \d_{\a_3,-\a_4}.
	\ee
	(Here we index the basis $e_\a$ for $\CC^r$ with
	$\a\in\{-S,-S+1,\dotsc,S\}$ where $S=(r-1)/2$.)
	Indeed, for the model with $a=b=0$ and $c>0$ the equivalence of
	partition functions was proved by Aizenman and Nachtergaele in
	\cite{an};  we give an algebraic proof for  general $a,b,c\in\RR$
	in Lemma \ref{lem:iso_reps_WB}.  We use the notation {\wb} for this
	model as its analysis is based on 
	the representation theory of the \emph{walled Brauer algebra}, see
	Section \ref{sec:WB}.  Interestingly, this model has the exact same free
	energy as the two-block interchange model:
	
	\begin{theorem}\label{thm:FE-WB}
		Let $a,b,c\in\RR$ be fixed.
		If $n,m\to\oo$ such that $m/n\to\rho\in(0,1)$, then the free energy
		of the model \eqref{eq:H-wb} satisfies 
		\be\label{eq:fe-max-WB}
		\Phi^\wb_\b(a,b,c):=\lim_{n\to\oo} \tfrac1n\log Z_{n}^\wb(\b)
		=\Phi^\ab_\b(a,b,c),
		\ee
		where $\Phi^\ab_\b(a,b,c)$ is given in Theorem \ref{thm:FE-AB}.
	\end{theorem}
	
	In the case $r=2$, Theorem \ref{thm:FE-WB} can be deduced from Theorem
	\ref{thm:FE-AB} in the following elementary manner.
	For $r=2$ we have \cite[Section 7.1]{ueltschi}
	\be
	T_{i,j}=2(\SS_i\cdot\SS_j)+\tfrac12,\qquad
	Q_{i,j}=2(S^{(1)}_i S^{(1)}_j-
	S^{(2)}_i S^{(2)}_j+S^{(3)}_i S^{(3)}_j)+\tfrac12.
	\ee
	Letting $W=\big(\begin{smallmatrix}0 & 1 \\ -1&
	0\end{smallmatrix}\big)$
	we have that  $W_j^{-1} T_{i,j} W_j=-Q_{i,j}+1$,
	so conjugating $H_n^\ab(a,b,-c)$ with 
	$\prod_{j=m+1}^nW_j$ gives 
	$H_n^\wb(a,b,c)-cm(n-m)/n$.  Thus
	$\Phi_\b^\wb(a,b,c)=\Phi^\ab_\b(a,b,-c)+c\rho(1-\rho)$.  
	This is  consistent with Theorem \ref{thm:FE-WB}
	since (indicating the dependence on $c$ with a subscript)
	$F_c(x_1,x_2;y_1,y_2)-F_{-c}(x_1,x_2;y_2,y_1)=c(x_1+x_2)(y_1+y_2)=c\rho(1-\rho)$,
	meaning that by Theorem \ref{thm:FE-AB} we have 
	$\Phi^\ab_\b(a,b,-c)+c\rho(1-\rho)=\Phi^\ab_\b(a,b,c)$.
	However, for general $r$ the rank of $T_{i,j}$ is $r(r+1)/2$ while the
	rank of $Q_{i,j}$ is 1, so when $r>2$, conjugating $T_{i,j}$ cannot
	give a linear 
	combination of $Q_{i,j}$ and the identity.

	\subsection{Phase transition and critical temperature}
	
	Next we discuss phase transitions as $\b$ is varied, via the maximiser of the function $F$. Essentially, when a transition is present, we expect the maximiser of $F$ to be fixed (at $\om_0$ \eqref{eq:om0}) for small $\beta$, and then at some critical $\beta_\crit$ to begin to move. This $\beta_\crit$ then corresponds to a point of phase transition in the model. For $\b=\b_\crit$ it can happen either that $\om_0$ is unique or that
	there are other maximum points. We will see that the phase-transition is also reflected in
	the behavior of observables (Theorem \ref{thm:TS-twoblock})
	and the magnetisation (Theorem \ref{thm:mag}). 
	
	In Proposition \ref{prop:crittemp}, we characterise completely
	the values of $a,b,c$ for which there exists such a phase
	transition. When it exists, finding explicit formulae for
	$\beta_\crit$ seems difficult in general; 
	we can do it in two cases, firstly in Proposition \ref{prop:r2crittemp} when $r=2$ 
	(that is, spin $\tfrac12$), and secondly in Proposition
	\ref{prop:tcrittemp} when $c\geq0$, $r\geq3$ and
	\begin{equation}\label{eq:tcond}
	(a-c)\rho=(b-c)(1-\rho)=:t.
	\end{equation}
	In the latter case, we further prove in Proposition
	\ref{prop:unique} that for $\b_\crit<\b<\b_\crit+\eps$ and
	$\eps>0$ small, there is a unique maximiser of $F$ that satisfies \eqref{eq:xgeq}.

	In what follows, we write $\vec x=(x_1,\dotsc,x_r)$,
	$\vec y=(y_1,\dotsc,y_r)$, and
	\be\label{eq:Om}
	\Om=\big\{(\vec x; \vec y):
	x_1,\dotsc,x_r,y_1,\dotsc,y_r\geq 0,\;
	\textstyle\sum_{i=1}^rx_i=1-\sum_{i=1}^ry_i=\rho
	\big\}.
	\ee
	Elements of $\Om$ will
	typically be denoted $\om=(\vec x;\vec y)$.  We write 
	\be\label{eq:om0}
	\om_0=\big(\tfrac\rho r, \tfrac\rho r,\dotsc, \tfrac\rho r; 
	\tfrac{1-\rho}r,\tfrac{1-\rho}r,\dotsc,\tfrac{1-\rho}r\big)
	\in\partial\Om,
	\ee
	and we write
	$Q(x,y)=\tfrac12(ax^2+by^2+2cxy)$
	for the quadratic form appearing in the function $f(x,y)$.

	\begin{proposition}\label{prop:crittemp}
		If $Q$ is negative semidefinite, that is, if
		\be
		a\leq 0,\quad b\leq 0,\quad \mbox{and}\quad  ab\geq c^2, 
		\ee
		then $F$ assumes its maximum value at  $\om_0$ 
		for all $\beta> 0$, and this maximum point is unique.
		Otherwise, there exists a  number $\beta_\crit> 0$ such that 
		$F$ assumes it maximum value at $\om_0$ if and only 
		if $0<\beta\leq \beta_\crit$, and this maximum is unique if
		$0<\b<\b_\crit$.  
	\end{proposition}

	Let us write $\b_\crit(r)$ to highlight the dependence on $r$. The next proposition gives $\beta_\crit(2)$ when it exists.
	For a simple interpretation of the value, see Lemma \ref{lem:perturbation}.
	
\begin{proposition}
  \label{prop:r2crittemp}
  Let $r=2$ and  assume that $Q$ is not negative semidefinite, so that 
  $\beta_\crit(2)$ exists. Then
  \begin{equation}\label{eq:r2crittemp}
    \beta_\crit(2)=\begin{cases}
      \displaystyle\frac{\rho a+(1-\rho) b-\sqrt{(\rho a-(1-\rho) b)^2+4\rho(1-\rho)c^2}}{\rho(1-\rho)(ab-c^2)}, & ab\neq c^2,\\
      \displaystyle\frac{2}{a\rho+b(1-\rho)}, & ab=c^2.
    \end{cases}
  \end{equation}
  Moreover, for
  $\beta=\beta_\crit$, $\omega_0$ is the unique maximum point.
\end{proposition}
	
	In the homogeneous spin-$\tfrac12$ {\ab}-model, i.e.\
	$r=2$ and 
	$a=b=c=1$, we recover the critical point 
	$\beta_\crit=2$
	first identified by T\'oth \cite{toth-bec} and by 
	Penrose \cite{penrose}.
	In the bipartite case $a=b=0$  we get the critical value
	$\beta_\crit=2/\sqrt{c^2\rho(1-\rho)}$;
	this has, to the best of our knowledge, not
	appeared previously in the literature.
	
	The next proposition gives $\beta_\crit(r)$, $r\ge3$ in the
	special case that $c\geq0$ and \eqref{eq:tcond} holds. 
	
	\begin{proposition}\label{prop:tcrittemp}
		Suppose that $c\geq 0$, $r\geq 3$, that
		\eqref{eq:tcond} 
		holds and
		that $Q$ is not negative semidefinite so that
		$\beta_\crit$ exists.
		Then
		\begin{equation}\label{eq:tcrittemp}
		\b_\crit=
		\beta_\crit(r)=\frac{2(r-1)\log(r-1)}{(r-2)(c+t)}.
		\end{equation}
		Moreover, if $\beta=\beta_\crit$ there are exactly two maximum
		points satisfying \eqref{eq:xgeq}, namely $\om_0$ of \eqref{eq:om0}
		and $\om_1=(\vec  x; \vec y)$ given by 
		\begin{subequations}\label{eq:specialxy}
			\begin{equation}\label{eq:specialx}
			x_1=\tfrac{(r-1)\rho}{r},\quad x_2=\dots=x_r=\tfrac {\rho}{r(r-1)}, 
			\end{equation}
			\begin{equation}\label{eq:specialy}
			y_1=\tfrac{(r-1)(1-\rho)}{r},\quad y_2=\dots=y_r=\tfrac
			{1-\rho}{r(r-1)}. 
			\end{equation}
		\end{subequations}	
	\end{proposition}
	
	For $\b>\b_\crit$ and under the conditions in 
	Proposition \ref{prop:tcrittemp} we can prove that the maximum point is unique  
	(subject to \eqref{eq:xgeq}) for $\b$ close to the critical point
	(see also Proposition	\ref{prop:maxc>0} for another special case).
	
	\begin{proposition}\label{prop:unique}
		Under the assumptions of Proposition \ref{prop:tcrittemp}, there exists $\varepsilon>0$ such that, if
		$\beta_\crit<\beta<\beta_\crit+\varepsilon$, there is a unique maximiser of
		$F$ in ${\Omega}$ with entries ordered as in \eqref{eq:xgeq}.
		Moreover as $\beta\searrow\beta_\crit$, this maximiser tends to $\omega_1$
		given in \eqref{eq:specialxy}.
	\end{proposition}

	\subsection{Correlations and magnetisation}
	
	We next move on to results about correlations which extend
	\cite[Theorem 2.3]{BFU20}.
	To state them,   introduce the function
	\be
	R(w_1,\dotsc,w_r; z_1,\dotsc,z_r)=
	\det\big[e^{w_iz_j}\big]_{i,j=1}^{r}
	\prod_{1\leq i<j\leq r} \frac{j-i}{(w_i-w_j)(z_i-z_j)}.
	\ee
	For $\#\in\{\ab,\wb\}$, we write 
	\be
	\langle\cO\rangle^\#_{\b,n}=
	\frac{\tr_\VV\big[\cO e^{-\beta H^\#_n}\big]}
	{Z_{n}^\#(\b)}
	\ee
	for the usual equilibrium state expectation of a linear operator $\cO$ 
	on $\VV$.

	\begin{theorem}\label{thm:TS-twoblock}
		Let $a,b,c\in\RR$ and $\b>0$ be such that $F$ has a unique maximum
		point $\om^\star=(\vec x^\star;\vec y^\star)$ satisfying
		\eqref{eq:xgeq}.
		Let $W$ be an $r\times r$ matrix 
		with eigenvalues $w_1,\dotsc,w_r\in\CC$.  As $n,m\to\oo$ such that
		$m/n\to\rho\in(0,1)$, we have that 
		\be\begin{split}\label{eq:TS}
			&\lim_{n\to\infty} \big\langle \exp \big\{
			\tfrac 1n \textstyle\sum_{i=1}^n W_i \big\} \big\rangle_{\beta,n}^\ab
			=R(w_1,\dotsc,w_r;z_1^\star,\dotsc,z^\star_r)\\
			& \lim_{n\to\infty} \big\langle \exp \big\{
			\tfrac 1n \big(\textstyle\sum_{i=1}^m W_i 
			-\textstyle\sum_{i=m+1}^n W^{\intercal}_i\big)
			\big\} \big\rangle_{\beta,n}^\wb
			=R(w_1,\dotsc,w_{r};z^\dagger_1,\dotsc,z^\dagger_{r}),
		\end{split}\ee
		where the superscript $^\intercal$ denotes transpose, and
		\be\label{eq:TS-z}
		z_j^\star=
		x_j^\star+y_j^\star,\qquad
		z_j^\dagger=x_j^\star-y_j^\star.
		\ee
	\end{theorem}

	As a concrete example, for 
	$W=h\,\mathrm{diag}(0,1,2,\dotsc,r-1)$
	we have  
	\be\label{eq:R}
	R(w_1,\dotsc,w_r;z_1,\dotsc,z_r)=
	\prod_{1\leq i<j\leq r} \frac{e^{h z_i}-e^{h z_j}}{h(z_i-z_j)}.
	\ee
	The phase-transition at $\b_\crit$ is reflected in the fact
	that $R\equiv 1$ when $\om^\star=(\vec{x}^\star;\vec{y}^\star)=\om_0$, while
	$R$ is non-trivial if the entries of $\vec z$ are non-constant.  The
	latter occurs e.g.\ in the $\ab$-model for $\b>\b_\crit$.
	
	For a second concrete example, let $c>0$.
	We will prove in Proposition \ref{prop:maxc>0} that any maximiser 
	$(\vec{x}^\star;\vec{y}^\star)$ of $F$ satisfying \eqref{eq:xgeq} is then of the form
	\be
	\begin{split}
		x_1^\star&\geq x_{2}^\star=\dotsb=x_r^\star,
		\qquad
		y_1^\star\geq y_{2}^\star=\dotsb=y_r^\star,
	\end{split}
	\ee
	in which case $z^\star$ \eqref{eq:TS-z} will be of the same form.  Letting
	$W$ be an arbitrary rank 1 projection, with eigenvalues
	$1,0,\dotsc,0$, and writing  $u^\star=z_1^\star-z_2^\star$,  we have
	\be\label{eq:TS-projector}
	\lim_{n\to\infty} \big\langle \exp \big\{
	\tfrac 1n \textstyle\sum_{i=1}^n W_i \big\} \big\rangle_{\beta,n}^\ab= 
	\frac{(2S)!}{(hu^\star)^{2S}}e^{\frac{h}{2S+1}(1-u^\star)}
	\sum_{j=2S}^\oo\frac{(hu^\star)^j}{j!}.
	\ee
	(The calculation of $R$ is performed in \cite[Section~6]{BFU20}.)
	
	Theorem \ref{thm:TS-twoblock} 
	also shows that the {\ab}- and {\wb}-models are not equivalent, despite
	having the same free energy (for any anti-symmetric matrix $W$, the
	observables on the left in \eqref{eq:TS} are the same, while their
	limiting expectations are different).
	The result is also relevant for understanding 
	extremal states, see Section \ref{sec:heuristics}.
	
	Finally we have the following result about the (thermodynamic)
	magnetisation.  
	Let $W$ be an $r\times r$ matrix with real 
	eigenvalues $w_1 \ge \cdots\ge w_r$,
	let $h\in\RR$, and write
	\be\label{eq:H-AB-mag}
	Z_{n}^{\ab}(\b,h) = \tr_\VV[\exp\big(-\beta H^{\ab}_{n}
	+ h\textstyle\sum_{1\le i\le n} W_i\big)],
	\ee
	\be\label{eq:H-WB-mag}
	Z_{n}^{\wb}(\b,h) = \tr_\VV[\exp\big(-\beta H^{\wb}_{n}
	+ h\big({\textstyle\sum_{1\le i\le m} W_i 
		- \sum_{m<i\le n} W_i^\intercal}\big)
	\big)].
	\ee
	In Theorem \ref{thm:fe-mag} we will obtain explicit expressions for
	the limits
	\be\label{eq:phi-mag}
	\Phi^\#(\b,h)
	:=\lim_{n\to\oo} \tfrac1n\log Z_{n}^\#(\b,h),
	\ee
	where $\#\in\{\ab,\wb\}$ (this turns out to depend on $W$ only through
	its spectrum $\vec w$).
	The \emph{magnetisation} is given by
	the left and right derivatives of
	this free energy with respect to $h$, at $h=0$. 
	
	\begin{theorem}\label{thm:mag}
		Let $\Phi$ be defined by \eqref{eq:phi-mag}, either for the
		{\ab}- or {\wb}-model. 
		Then 
		\be\label{eq:partial-fe-mag}
			\frac{\partial \Phi}{\partial h}\Big|_{h \downarrow 0}=
			\max_{(\vec x^\star;\vec y^\star)} 
			\sum_{i=1}^r z_i w_i,
			\qquad
			\frac{\partial \Phi}{\partial h}\Big|_{h \uparrow 0}=
			\min_{(\vec x^\star;\vec y^\star)} 
			\sum_{i=1}^r z_i w_{r+1-i},
		\ee
		where the maxima and minima are  over 
		all maximisers $(\vec{x}^\star;\vec{y}^\star)\in\Omega$ of
		$F(\vec{x}; \vec{y})$ such that $x_1^\star\geq\dots\geq x_r^\star$.
		The vector $\vec z$ is obtained by 
		rearranging the entries in the vector $x^\star\pm y^\star$ in decreasing order, where one should take the plus sign for the
		{\ab}-model and the minus sign for the {\wb}-model.
	\end{theorem}

	It is natural to take $W$ to have trace zero.  Then, 
	from Proposition \ref{prop:crittemp},
	for all $\b<\b_\crit$ the only maximiser is $\om_0$
	\eqref{eq:om0} and we have
	\begin{equation}\label{eq:trivmagn}
	\frac{\partial \Phi}{\partial h}\big|_{h \downarrow 0}
	=\frac{\partial \Phi}{\partial h}\big|_{h \uparrow 0}=0,
	\end{equation}
	for both {\ab}- and {\wb}-models and for both $c>0$ and $c<0$.  
	This holds also for $\b=\b_\crit$ when $r=2$. 
	
	Let us discuss the case
	$r\geq 3$ in Proposition \ref{prop:tcrittemp} at $\b=\b_\crit$. 
	Recall that $c\geq 0$ in this case.
	Calculations with the point $\om_1$ \eqref{eq:specialxy} give the
	following:
	\begin{itemize}[leftmargin=*]
		\item In the \ab-case, at $\om_1$ the values
		\be
		z_1=\tfrac{r-1}{r}, \quad z_2=\dotsb=z_r=\tfrac1{r(r-1)}
		\ee
		are already decreasing.  Still assuming that $W$ has trace zero, 
		it follows that
		\be
		\tfrac{\partial \Phi^\ab}{\partial h}\big|_{h \downarrow 0}
		=\tfrac{r-2}{r-1} w_1
		\qquad
		\tfrac{\partial \Phi^\ab}{\partial h}\big|_{h \uparrow 0}
		=\tfrac{r-2}{r-1} w_r.
		\ee
		For non-trivial $W$ we have $w_1>0>w_r$, thus the magnetisation is
		discontinuous at the point of phase-transition.
		\item In the \wb-case, at $\om_1$ the ordering of the
		values $x_i-y_{i}$ depends on $\rho$.  
		If $\rho>\tfrac12$ we get
		\be
		z_1=(2\rho-1)\tfrac{r-1}{r}, 
		\quad z_2=\dotsb=z_{r}=\tfrac{2\rho-1}{r(r-1)},
		\ee
		and from there
		\be\label{eq:wbmag}\begin{split}
			\tfrac{\partial \Phi^\wb}{\partial h}\big|_{h \downarrow 0}
			&=(2\rho-1)\tfrac{r-2}{r-1}w_1 
			\\
			\tfrac{\partial \Phi^\wb}{\partial h}\big|_{h \uparrow 0}
			&=(2\rho-1)\tfrac{r-2}{r-1}w_r.
		\end{split}\ee
		For non-trivial $W$, this gives a discontinuous magnetisation.
		In the case $\rho<\tfrac12$, the magnetisation is obtained by
		exchanging $w_1$ and $w_r$ in \eqref{eq:wbmag}.  For
		$\rho=\tfrac12$, the magnetisation is continuous at the point of
		phase-transition.
	\end{itemize}

	\subsection{Ground-state phase diagrams}
\label{ssec:gs}
By analysing the location of the maximiser of the function $F$
 (given in \eqref{eq:F})
in the limit
as $\beta\to\infty$, we can identify the ground-state phase
diagram.  We provide two 
diagrams, one of the $(a,b)$ plane for $c>0$ fixed and one for
$c<0$ fixed.  Since the diagram is invariant under the scaling
$(a,b,c)\to(\alpha a,\alpha b,\alpha c)$ with $\alpha>0$, this
will suffice to describe the whole diagram for $c\neq0$.  The
case $c=0$ is just two uncoupled models on complete graphs
with $T_{i,j}$ transposition interaction; this is covered by
the results of \cite{Bjo16}.

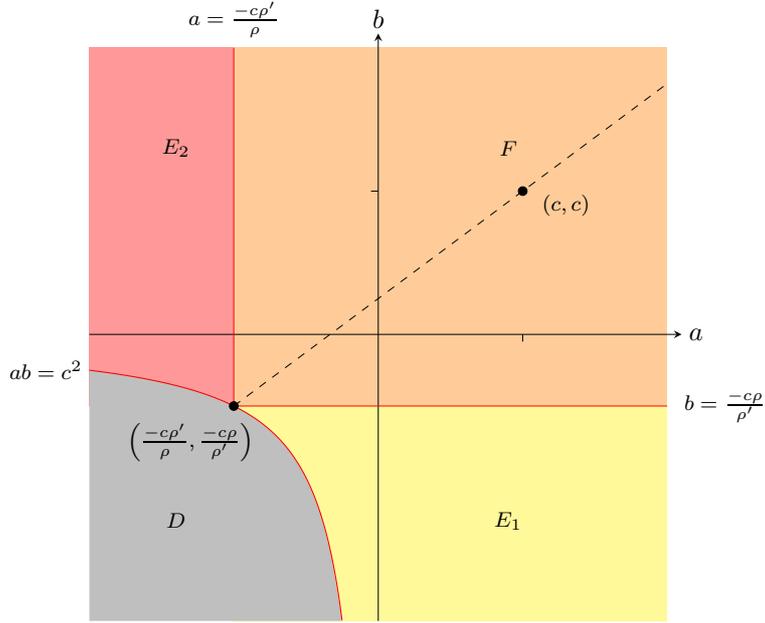
\begin{figure}[h]
  \centering
  \begin{tikzpicture}[scale=1.9]
    
    \path[fill = myorange, opacity = 0.4] (-1,-0.5) -- (-1,2) -- (2,2) -- (2,-0.5) -- (-1,-0.5);
    \path[fill = red, opacity = 0.4] (-1,-0.5) -- (-1,2) -- (-2,2) -- (-2,-0.5) -- (-1,-0.5);
    \path[fill = yellow, opacity = 0.4] (-1,-0.5) -- (-1,-2) -- (2,-2) -- (2,-0.5) -- (-1,-0.5);
    \fill [lightgray, domain=-2:-0.25, variable=\x]
  (-2, -2)
  -- plot ({\x}, {1/(2*\x)})
  -- (-0.25,-2)
  -- cycle;
    
    \draw[->, >=stealth, line width=0.3pt](-2, 0) -- (2.1,0);
    \draw[->, >=stealth, line width=0.3pt](0,-2) -- (0,2.1);
    
    \draw [smooth,samples=100,domain=-2:-0.25,variable=\x, red] plot(\x,{1/(2*\x)});
    \draw[-, >=stealth, line width=0.3pt, red](-1,-0.5) -- (-1,2);
    \draw[-, >=stealth, line width=0.3pt, red](-1,-0.5) -- (2,-0.5);
    \draw[-, dashed, >=stealth, line width=0.3pt](-1,-0.5) -- (2,1.75);
    
    \draw[font=\fontsize{10}{10}] (0, 2.2) node {$b$};
    \draw[font=\fontsize{10}{10}] (2.2, 0) node {$a$};
    
    \draw[fill=black] (1,1) circle (0.03cm);
    \draw[font=\fontsize{8}{10}] (1.3, 0.9) node {$(c,c)$};
    
    \draw[fill=black] (-1,-0.5) circle (0.03cm);
    \draw[font=\fontsize{8}{10}] (-1.3,-0.75) node {$\left(\frac{-c\rho'}{\rho},\frac{-c\rho}{\rho'}\right)$};
    
    \draw[font=\fontsize{8}{10}] (-1,2.2) node {$a=\frac{-c\rho'}{\rho}$};
    \draw[font=\fontsize{8}{10}] (2.4,-0.5) node {$b=\frac{-c\rho}{\rho'}$};
    \draw[font=\fontsize{8}{10}] (-2.3,-0.25) node {$ab=c^2$};
    
    \draw[-, >=stealth, line width=0.3pt](0,1) -- (-0.05, 1);
    \draw[-, >=stealth, line width=0.3pt](1,0) -- (1,-0.05);
    
    \draw[font=\fontsize{8}{10}] (0.9, 1.3) node {{$F$}};
    \draw[font=\fontsize{8}{10}] (-1.4, -1.3) node {$D$};
     \draw[font=\fontsize{8}{10}] (0.9, -1.3) node {$E_1$};
      \draw[font=\fontsize{8}{10}] (-1.4, 1.3) node {$E_2$};
    
    \end{tikzpicture}
    \caption{The ground state phase diagram 
for $c>0$. The dashed line indicates where we have a closed formula for the critical temperature.}
    \label{fig:PDintro-c>0}
\end{figure}

The $c>0$ diagram is portrayed in Figure \ref{fig:PDintro-c>0}.
It displays four distinct regions, separated by the curve $ab=c^2$ ($a,b<0$) and
the lines $a=-c\rho'/\rho$ and $b=-c\rho/\rho'$. The dashed line
 $(a-c)\rho=(b-c)(1-\rho)$ is where we have a precise formula for the critical
 temperature, see Proposition \ref{prop:tcrittemp}. 
The upper right region $F$ is called \emph{ferromagnetic}; the
$c$-interaction between the two blocks is ferromagnetic and the $a$-
and $b$-interactions are either ferromagnetic, or 
 not strong enough to make a difference. 
In this region, we obtain from Theorem
\ref{thm:mag} that the magnetisation 
is maximal. The lower left  region $D$ we
 call \emph{disordered}; it coincides with the range of parameters
 for which there is no
 phase transition at finite temperature, by Proposition
 \ref{prop:crittemp}.  Here the $a$- or $b$-interactions overcome the
 $c$-interactions, and the model 
behaves like two copies of
 the antiferromagnet on the complete graph, which has no phase
 transition \cite{Bjo16}. The magnetisation in this case is $0$. There are also two intermediate regions denoted $E_1$ and $E_2$. Here, at least one of the
 $a$- or $b$-interactions is antiferromagnetic, and the model begins
 to feel this effect. In these regions 
 the magnetisation interpolates between $0$ and its maximal value. As $|a|+|b|$ becomes
 large, we approach the $c=0$ limit of a ferromagnet on one subgraph
 and an antiferromagnet on the other. 
    
    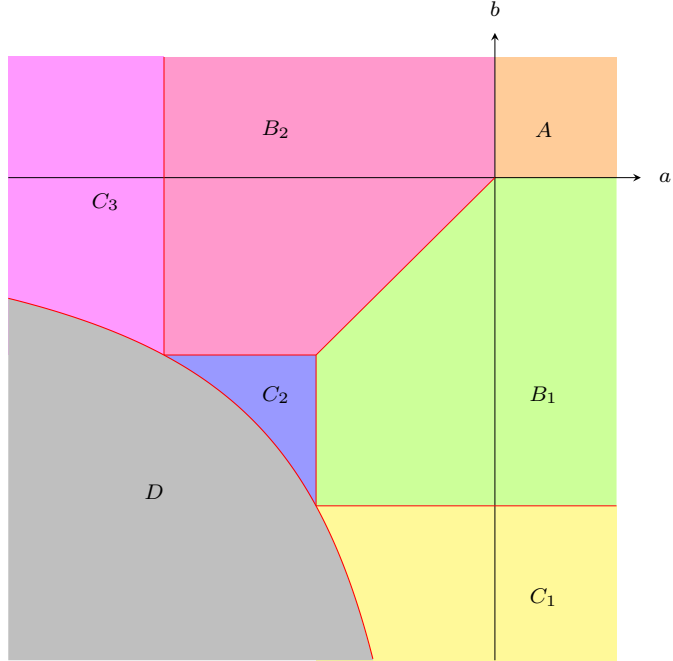
\begin{figure}[h]
    \centering
    \begin{tikzpicture}[scale=3.2]
    
      \path[fill = blueredred, opacity = 0.4] (0,0) -- (0,.5) --
    (-1.36,.5) -- (-1.36,-0.735) -- (-0.735,-0.735) -- cycle;
    \draw[font=\fontsize{8}{10}] (-0.9,0.2) node {$B_2$};
    \path[fill = bluered, opacity = 0.4](-1.36,0.5) 
    -- (-1.36,-0.735) -- (-2,-.735)--(-2,.5) -- cycle;
     \draw[font=\fontsize{8}{10}] (-1.6,-0.1) node {$C_3$};
    \path[fill = blue, opacity = 0.4]  (-1.36,-0.735) --
    (-0.735,-0.735) -- (-0.735,-1.36 ) --cycle;
\draw[font=\fontsize{8}{10}] (-0.9,-0.9) node {$C_2$};
\path[fill = yellow, opacity = 0.4] (0.5,-2) -- (.5,-1.36) --
(-.735,-1.36) -- (-0.735,-2) -- cycle; 
\draw[font=\fontsize{8}{10}] (0.2,-1.74) node {$C_1$};
    \path[fill = yellowgreen, opacity = 0.4](0,0) -- (-0.735,-.735) --
    (-0.735,-1.36) -- (.5,-1.36) -- (.5,0) -- cycle;
    \draw[font=\fontsize{8}{10}] (0.2,-0.9) node {$B_1$};
    \fill [lightgray, domain=-2:-0.5, variable=\x](-2, -2) -- plot ({\x}, {1/(1*\x)}) -- (-0.125,-2) -- cycle;
    \draw[font=\fontsize{8}{10}] (-1.4, -1.3) node {$D$};
     \path[fill = myorange, opacity = 0.4] (0.5,0) -- (0,0) --
    (0,0.5) -- (0.5,0.5) -- cycle; 
    \draw[font=\fontsize{8}{10}] (0.2,0.2) node {$A$};

    \draw[->, >=stealth, line width=0.3pt](-2, 0) -- (.6,0);
    \draw[font=\fontsize{8}{10}] (.7, 0) node {$a$};
    \draw[->, >=stealth, line width=0.3pt](0,-2) -- (0,.6);
    \draw[font=\fontsize{8}{10}] (0, .7) node {$b$};

    \draw [smooth,samples=100,domain=-2:-0.5,variable=\x, red]
    plot(\x,{1/(1*\x)});
    \draw[-, >=stealth, line width=0.3pt, red](-1.36,0.5) -- (-1.36,-.735);
    \draw[-, >=stealth, line width=0.3pt, red](-1.36,-0.735) -- (-0.735,-0.735);
    \draw[-, >=stealth, line width=0.3pt, red] (-0.735,-0.735) --  (-0.735,-1.36);
  \draw[-, >=stealth, line width=0.3pt, red]   (-0.735,-1.36) --  (0.5,-1.36);
    \draw[-, >=stealth, line width=0.3pt, red](0,0) -- (-0.735,-0.735);
\end{tikzpicture}
    \caption{The ground state phase diagram
for $c<0$, in the case $r=3$.}
    \label{fig:PDintro-c<0}
\end{figure}
    
When $c<0$ and $r=2$ the phase diagram looks identical to the case
when $c>0$, but we  
refer to the upper-right region as \emph{antiferromagnetic}. 
As $r\geq 3$, the diagram looks more complicated, with $2r-1$
intermediate regions 
between the antiferromagnetic and disordered regions. This is
illustrated in  
Figures \ref{fig:PDintro-c<0} (for $r=3$) and \ref{fig:PDproof-c<0}
(for $r=5$),
and described in detail in Proposition
\ref{prop:pdc<0}.

We can give a tentative interpretation of the diagram when $r=3$, $c<0$. Here, the $c$-interaction is $-(\SS_i\cdot \SS_j)^2$ in the $\wb$ model, so spins in one
block want to be orthogonal to those in the other, and is
$-[(\SS_i\cdot \SS_j)+(\SS_i\cdot \SS_j)^2]$ in the $\ab$ model, so spins in one
block want to be at $120^\circ$ to those in the other. The $a$ and $b$
interactions are both $(\SS_i\cdot \SS_j)+(\SS_i\cdot
\SS_j)^2$, so spins want to be aligned. 

One might interpret the diagram as follows. The region $A$ is truly ``anti''-ferromagnetic, in the sense that spins in $A$ are all aligned, and spins in $B$ are all aligned, in some direction orthogonal/at $120^\circ$ to those in $A$. We write ``anti'' in quotation marks since the angle between the spins is not $180^\circ$. There are two regions $B_1, B_2$, and three $C_1,C_2,C_3$. In the $B_1$ region, the spins in $A$ are aligned, and the spins in $B$ are disordered, but lie on the circle which is orthogonal/at $120^\circ$ to the spins in $A$; and vice-versa for $B_2$. As we decrease $b$ into the region $C_1$, the spins in $B$ become more and more disordered, until they are completely decoupled from those in $A$, which remain aligned. Similar for the $C_3$ region. It is difficult to interpret the most interesting region, $C_2$, in this way; there is some disorder in the spins in each block, but enough $c$-interaction to prevent them from completely decoupling.

	\subsection{Heuristics for extremal Gibbs states}
	\label{sec:heuristics}
	
	In \cite{BFU20}, for several models on 
$\ZZ^d$, including the interchange 
	model \eqref{eq:intch}, the authors
	give a heuristic argument which points towards the structure of the
	set $\Psi_\b$ of extremal Gibbs states at inverse temperature
	$\b$.  The description given there is expected to hold for $d$
        large enough, with $d\geq 3$ perhaps being enough.
Rather than explicitly defining the extremal Gibbs states in infinite
volume on the complete graph, the working is by analogy. Specifically,
their heuristics consist of two expected equalities: first that 
	\be\label{eq:Gibbs-states}
	\lim_{\Lambda\to\ZZ^d} \langle e^{\frac{h}{|\Lambda|}\sum_iW_i} \rangle_{\b,\Lambda} = 
	\int_{\Psi_\beta}e^{h\langle W_0\rangle_\psi}d\mu(\psi),
	\ee
	for $r\times r$ matrices $W$, where
	$\langle\cdot\rangle_\psi$ is an extremal Gibbs state, 
$\Psi_\beta$ is the set of extremal Gibbs states,
$d\mu$
        is the measure on $\Psi_\beta$
corresponding to the symmetric Gibbs state, $W_0$ is the operator $W$ at
the lattice site $0$, and the left hand side is the limit of
successively larger boxes $\Lambda\in\ZZ^d$; second that  
	\begin{equation}\label{eq:Gibbs-states-2}
	\lim_{n\to\oo}\langle e^{\frac{h}{n}\sum_iW_i} \rangle_{\b,n} =
	\lim_{\Lambda\to\ZZ^d}\langle e^{\frac{h}{|\Lambda|}\sum_iW_i} \rangle_{\b,\Lambda},
	\end{equation}
	where the left hand term is the observable on the complete graph.
	The left hand side of (\ref{eq:Gibbs-states-2}) is computed rigorously on the complete graph, and then, with the expected structure of $\Psi_\beta$ inserted, the right hand side of (\ref{eq:Gibbs-states}) is rigorously computed, and the two are shown to be the same. This working is not a proof either of the expected equalities (\ref{eq:Gibbs-states}), (\ref{eq:Gibbs-states-2}) or of the expected structure of $\Psi_\beta$, but it gives a consistency check for the three statements.

Using the results of the present paper, 
	we can provide the same
	calculations and heuristics for the 
\ab- and \wb-models.	
Both models have symmetry under $\mathrm{U}(r)$, the group of unitary 
$r\times r$  matrices, and 
for $c>0$, both models are expected to have
extremal Gibbs states labelled by $\CC\PP^{r-1}$, i.e.\ rank 1
projections in $\CC^r$. 
This means that the expected identites \eqref{eq:Gibbs-states} 
and \eqref{eq:Gibbs-states-2}  take the form
\be
\lim_{n\to\infty} \big\langle e^{
    \frac 1n \sum_{i=1}^n W_i }
    \big\rangle_{\beta,n}^\ab
=
 \int_{\CC\PP^{r-1}}  
e^{ \rho \langle W_1\rangle_\psi^\ab
+ (1-\rho)\langle W_2\rangle_\psi^\ab}
d\mu(\psi)
\ee
and
\be
\lim_{n\to\infty} \big\langle e^{
    \frac 1n (\sum_{i\in A} W_i 
-\sum_{j\in B} W_j^{\intercal})}
    \big\rangle_{\beta,n}^\wb
=
\int_{\CC\PP^{r-1}}  
e^{ \rho \langle W_1\rangle_\psi^\wb
-(1- \rho) \langle W^\intercal_2\rangle_\psi^\wb
} d\mu(\psi),
\ee
where $W_1$ and $W_2$ represent $W$ acting on arbitrary
sites in the $A$- and $B$-parts of the graph.
Using the $\mathrm{U}(r)$-invariance and the Harish-Chandra--Itzykson--Zuber
formula as in \cite{BFU20}, this leads to the predictions
\be
\lim_{n\to\infty} \big\langle e^{
    \frac 1n \sum_{i=1}^n W_i }
    \big\rangle_{\beta,n}^\ab
= R(w_1,\dotsc,w_r; x_1+y_1,\dotsc,x_r+y_r)
\ee
and
\be
\lim_{n\to\infty} \big\langle e^{
    \frac 1n (\sum_{i\in A} W_i 
-\sum_{j\in B} W_j^{\intercal})}
    \big\rangle_{\beta,n}^\wb
= R(w_1,\dotsc,w_r; x_1-y_1,\dotsc,x_r-y_r),
\ee
where $x_i=\langle P_1^{e_i}\rangle_{e_1}$
and $y_i=\langle P_2^{e_i}\rangle_{e_1}$ are the expectations of the
projections $P^{e_i}$ onto the subspace spanned by the $i$-th
coordinate vector $e_i=(0,\dotsc,0,1,0,\dotsc,0)$
under the extremal state associated with $\psi=e_1$.
By $\mathrm{U}(r)$-invariance, we expect $x_2=x_3=\dotsb=x_r$
and $y_2=y_3=\dotsb=y_r$, and it is further natural to assume that
$x_1\geq x_2$ and $y_1\geq y_2$.  Since this fits the picture given
(rigorously) by Theorem \ref{thm:TS-twoblock}
and Proposition \ref{prop:maxc>0}, we are motivated to
lend some credence to the stated heuristics.

We now turn to the case of the complete bipartite graph, given by
$a=b=0$. 
By our comments below \eqref{eq:H-wb}, the \wb-model with
$a=b=0$, $c=1$, has Hamiltonian unitarily equivalent to 
\be\label{eq:H-Singlet}
-\frac1n
\sum_{1\leq i\leq m<j\leq n}
P_{i,j},
\ee
where $P_{i,j}$ is ($r$ times) the projection onto the 
singlet state, given by \eqref{eq:P}.
For spin $S=1$ ($r=3$)
we can interpret our results and heuristics to comment on the 
bilinear-biquadratic model, 
which has Hamiltonian  
\be\label{eq:H-Bil-Biq}
-\frac1n
\sum_{1\leq i\leq m<j\leq n}\Big(
J_1(\SS_i\cdot \SS_j)+J_2(\SS_i\cdot \SS_j)^2
\Big),
\ee
where $\SS_i\cdot \SS_j=\sum_{k=1}^3S_i^{(k)}S_j^{(k)}$, 
and $J_1,J_2\in\RR$.  Indeed, using the relations 
$\SS_i \cdot \SS_j =T_{i,j} - P_{i,j}$ and 
$(\SS_i \cdot \SS_j)^2 = P_{i,j} + 1$
(see Lemma 7.1 from \cite{ueltschi}) one can rewrite
(\ref{eq:H-Bil-Biq}), up to addition of a constant, as 
\be\label{eq:H-Bil-Biq2}
-\frac1n
\sum_{1\leq i\leq m<j\leq n}\Big(
J_1T_{i,j}+(J_2-J_1)P_{i,j}
\Big).
\ee 
Setting
$J_1=J_2=\pm1$ gives the \ab\ model with $a=b=0$, $c=\pm1$, 
while setting $J_1=0$, $J_2=\pm1$
gives the \wb\ model with $a=b=0$, $c=\pm1$, in the form
\eqref{eq:H-Singlet}. The case $J_1=0$, $J_2=1$ (i.e.\ our 
\wb-model with $a=b=0$, $c=1$) is the biquadratic
Heisenberg model. These two special cases are exactly those described
by Ueltschi (\cite{ueltschi}, Section 7B) as having 
$\mathrm{SU}(3)$ invariance; 
in our language this is the $\GL(3)$-invariance that 
we exploit in this paper.
	
The phase diagram of the bilinear-biquadratic Heisenberg model on
$\mathbb{Z}^d$, $d\ge3$, is given in Ueltschi \cite{ueltschi}, and we
expect that the model on the complete bipartite graph has the same
diagram. 
(See also \cite{ueltschi-curious}, but beware that the predictions
using Gell-Mann matrices there are most likely wrong.
The corresponding one-dimensional spin chain has a different
phase-diagram, exhibiting dimerization, see \cite{ADCW,BMNU}.)
The biquadratic model ($J_1=0,J_2=1$) lies on the boundary of the nematic
phase of that diagram, but actually belongs to a N\'eel-ordered (or
antiferromagnetic) phase for bipartite graphs.  
Heuristically, we
expect the spins in the $A$-part to be anti-aligned with those in the
$B$-part. Note that for this model if we add a magnetisation term 
in the $S^{(k)}$ direction at every vertex (for any $k=1,2,3$), then, at $\b=\b_\crit$ 
and for $\rho>\tfrac12$, Theorem \ref{thm:mag} tells us that
the magnetisation is 
\be
\frac{\partial \Phi^\wb}{\partial h}\Big|_{h \downarrow 0}=
\rho-\tfrac12,
\ee
(indeed, see Lemma \ref{lem:AppendixB-mag}) which agrees with the picture of anti-aligned spins in the two blocks.

\subsection{Acknowledgements}

JEB gratefully acknowledges financial support from Vetenskapsr{\aa}det,
grants 2015-05195 and 2019-04185, from \emph{Ruth och Nils Erik
  Stenb\"acks stiftelse}, and from
Sabbatical Program at the Faculty of Science,
University of Gothenburg.
HR gratefully acknowledges support from Vetenskapsr{\aa}det,
grant 2020-04221. KR gratefully acknowledges support from the EPSRC Studentship 1936327, and from the FWF stand-alone grant P 34713.
KR would like to thank Sasha Sodin for many useful
discussions. 
JEB and KR are grateful for hospitality at the University of Warwick
and for several enlightening discussions with Daniel Ueltschi.
We all thank Martin Halln\"as for stimulating discussions at the start of
the project.

	\section{Free energy and correlations}
	\label{sec:twoblock}
	
	In this section we prove 
	Theorems \ref{thm:FE-AB}, \ref{thm:FE-WB},
	\ref{thm:TS-twoblock} and \ref{thm:mag}.
	
	\subsection{Interchange model:
		proof of Theorem \ref{thm:FE-AB}}
	\label{sec:AB}
	As noted in the introduction, our method is to identify the
	eigenspaces of the Hamiltonian \eqref{eq:H-AB}. This is
	facilitated by the classical theory of Schur--Weyl duality. 
	We start by 
	recalling a few basic definitions and facts.
	A \emph{partition} $\la\vdash n$ of $n$ is a non-increasing sequence of non-negative
	integers summing to $n$: $\la=(\la_1,\la_2,\dotsc)$ with
	$\la_1\geq\la_2\geq\dotsb\geq0$ and $\sum_{k\geq 1}\la_k=n$.  
	Its \emph{length} $\ell(\la)$ is the number of non-zero entries.
	
	For $\s\in S_n$ a permutation of
	$1,2,\dotsc,n$,  let $T_\s$ be the linear operator on
	$\VV=(\CC^r)^{\otimes n}$ which permutes the tensor factors according
	to $\s$:
	\be\label{eq:T}
	T_\s(v_1\otimes v_2\otimes\dotsb\otimes v_n)=
	v_{\s^{-1}(1)}\otimes v_{\s^{-1}(2)}\otimes 
	\dotsb\otimes v_{\s^{-1}(n)}.
	\ee
	The mapping  $\sigma\mapsto T_\sigma$  is a representation of $S_n$ and hence 
	extends to a representation of the
	group algebra $\CC[S_n]$ on $\VV$.  
	We may also regard $\VV$ as a module for the group 
	$\GL_r(\CC)$ of invertible $r\times r$ matrices by the diagonal action 
	\be
	g(v_1\otimes v_2\otimes\dotsb\otimes v_n)=
	g(v_1)\otimes g(v_2) \otimes\dotsb\otimes g(v_n).
	\ee
	Classical Schur--Weyl duality \cite[Corollary 4.59]{etingof}
	states that these actions of $S_n$ and
	of $\GL_r(\CC)$ are each others' centralisers, so that $\VV$ may be
	regarded as a representation of 
	the direct product $\GL_r(\CC)\times S_n$, and that $\VV$
	decomposes as a multiplicity-free direct sum of irreducible
	representations of $\GL_r(\CC)\times S_n$.   Specifically,
	\be\label{eq:swd}
	\VV= \bigoplus_{\la\vdash n,\, \ell(\la)\leq r} 
	U_\la \otimes V_\la.
	\ee
	Here $U_\la$ is the irreducible $\GL_r(\CC)$-representation 
	indexed by (its highest weight) $\la$, and 
	$V_\la$ is the irreducible $S_n$-representation (Specht module)
	indexed by $\la$.  We use the same notation $T$ for the representation
	of  $\GL_r(\CC)\times S_n$ on $\VV$.
	
	Recall our Hamiltonian  $H^\ab_n$ given in \eqref{eq:H-AB}.
	We now write this  as  
	$H^\ab_n=T(h^\ab_n)$ where
	\be
	h^\ab_n=-\tfrac1n[
	(a-c) \a_A+(b-c) \a_B+c \,\a_{AB}
	],
	\ee
	and where $\a_A,\a_B,\a_{AB}$ are the following elements of
	$\CC[S_n]$:
	\be
	\a_A=\sum_{1\leq i<j\leq m} (i,j),\quad
	\a_B=\sum_{m+1\leq i<j\leq n} (i,j),\quad
	\a_{AB}=\sum_{1\leq i<j\leq n} (i,j).
	\ee
	We have by linearity that 
	$e^{-\b H^\ab_n}=T(e^{-\b h^\ab_n})$.
	Now let $W$ be an $r\times r$ matrix over $\CC$.
	Then $e^W\in\GL_r(\CC)$ and we have that 
	$T(e^W)=\exp\big(\textstyle\sum_{i=1}^n W_i\big)$.
	Thus we may write 
	\be
	\exp\big(\textstyle\sum_{i=1}^n W_i\big)
	e^{-\b H^\ab_n}
	=T\big(e^W e^{-\b h^\ab_n}\big),
	\ee
	where $e^W e^{-\b h^\ab_n}\in \CC[\GL_r(\CC)\times S_n]$.
	
	Let us now consider how $e^W\times e^{-\b h^\ab_n}$ acts on the
	right-hand-side of \eqref{eq:swd}, starting with how 
	$e^{-\b h^\ab_n}$ acts on $V_\la$.  The term $\a_{AB}$ is the sum
	of all elements of a conjugacy class (the transpositions), hence it
	belongs to the center of $\CC[S_n]$.  By Schur's Lemma, it therefore
	acts as a constant multiple of the identity on $V_\la$.  The constant
	in question is well known 
	\cite[p.~52]{fulton-harris}
	to equal the \emph{content} of the
	partition $\la$, defined by
	\be\label{eq:ct-sum}
	\ct(\la)=\sum_{j\geq 1}\Big(\frac{\la_j(\la_j+1)}{2}-
	j \la_j\Big).
	\ee  
	(This equals the sum of 
	the contents of all boxes in the
	Young diagram of $\la$, where the content of a box in position
	$(x,y)$ is $y-x$.)
	We have
	\be
	\a_{AB}|_{V_\la}=\ct(\la) \Id_{V_\la}.
	\ee
	Now, to deal with the remaining two terms
	$\a_A$ and $\a_B$, note that 
	as a representation of $S_m\times S_{n-m}$,
	the module $V_\la$ splits as
	\be\label{eq:LR-def}
	V_\la= \bigoplus_{\mu\vdash m,\,\nu\vdash n-m}
	c_{\mu,\nu}^\la V_\mu\otimes V_\nu,
	\ee
	where $c_{\mu,\nu}^\la$ are non-negative integers known as 
	the \emph{Littlewood--Richardson coefficients}.  
	We give more details about these numbers later, for now we just note
	that  
	$c_{\mu,\nu}^\la\neq0$ only if $\ell(\mu),\ell(\nu)\leq \ell(\la)$.
	On each term of the sum in \eqref{eq:LR-def}, $\a_A$ acts as
	$\ct(\mu) \Id_{V_\mu}$ and $\a_B$ acts as
	$\ct(\nu) \Id_{V_\nu}$, consequently $h^\ab_n$ acts on that term as 
	\be
	-\tfrac1n
	[(a-c)\ct(\mu)+(b-c)\ct(\nu)+c\,\ct(\la)]\Id_{V_\mu\otimes V_\nu},
	\ee
	and therefore $e^{-\b h^\ab_n}$ acts as
	\be
	\exp\big(\tfrac\b n
	[(a-c)\ct(\mu)+(b-c)\ct(\nu)+c\,\ct(\la)]\big)
	\Id_{V_\mu\otimes V_\nu}.
	\ee
	
	As to the factor $e^W$, we first note that the character of the module
	$U_\la$ evaluated at $g\in \GL_r(\CC)$ with eigenvalues
	$x_1,\dotsc,x_r$ is the Schur polynomial:
	\be\label{eq:GL-char}
	\chi_{U_\la}[g]=s_\la(x_1,\dotsc,x_r)=
	\frac{\det[x_i^{\la_j+r-j}]_{i,j=1}^r}{\prod_{1\leq i<j\leq r}(x_i-x_j)}.
	\ee
	If $W$ has eigenvalues $w_1,\dotsc,w_r$, then $e^W$
	has eigenvalues $e^{w_1},\dotsc,e^{w_r}$.  
	Writing $d_\mu,d_\nu$ for the dimensions of $V_\mu,V_\nu$, we may
	summarise these findings as follows:
	
	\begin{lemma}\label{lem:A-AB}
		Suppose that $W$ has eigenvalues  $w_1,\dotsc,w_r$.  Then
		\be\label{eq:eW-bi}
		\begin{split}
			\tr_\VV[\exp\big(\textstyle{\sum_{i=1}^n} W_i\big) 
			e^{-\b H^\ab_n}]
			=
			&\sum_{\la,\mu,\nu} s_\la(e^{w_1},\dotsc,e^{w_r}) 
			c_{\mu,\nu}^\la  d_\mu d_\nu  \\
			&\quad\cdot\exp\Big(\tfrac\b n[(a-c)\ct(\mu)+(b-c)\ct(\nu)+c\cdot\ct(\la)]
			\Big),
		\end{split}
		\ee
		where the sum is over $\la\vdash n$ with $\ell(\la)\leq r$,
		$\mu\vdash m$, and $\nu\vdash n-m$.
		In particular, setting $W$ to be the zero matrix (so that $e^W=\Id$), 
		\be\label{eq:Z-bi}
		Z_{\b,n}^\ab=
		\sum_{\la,\mu,\nu} s_\la(1,\dotsc,1) 
		c_{\mu,\nu}^\la  d_\mu d_\nu 
		\exp\Big(\tfrac\b n[(a-c)\ct(\mu)+(b-c)\ct(\nu)+c\cdot\ct(\la)] \Big).
		\ee
	\end{lemma}
	
	We will use that 
	\be\label{eq:dimU}
	s_\la(1,\dotsc,1)=
	\dim(U_\la)=\prod_{1\leq i<j\leq r} 
	\frac{\la_i-i-\la_j+j}{j-i}.
	\ee
	As to $d_\mu$, a convenient formula is
	\be\label{eq:dimV}
	d_\mu=\dim(V_\mu)=
	\frac{n!}{m_1!\dotsb m_r!} 
	\prod_{1\leq i<j\leq r}(m_i-m_j)
	\ee
	where $m_i=\mu_i+r-i$, see \cite[(4.11)]{fulton-harris}.
	
	In Lemma \ref{lem:A-AB} we have written the partition function as a
	sum of terms exponentially large in $n$, with relatively few
	summands.  Such a sum is dominated by its largest term.  To prove
	Theorem \ref{thm:FE-AB} we need to understand the asymptotic behavior
	of each of the factors in \eqref{eq:Z-bi}, and since only terms with $c^\la_{\mu,\nu}\neq0$ appear in the sum, we need a condition for
	$c^\la_{\mu,\nu}\neq0$.  
	
	\begin {proof}[Proof of Theorem \ref{thm:FE-AB}]
	First, from \eqref{eq:dimU} we see that
	$\dim(U_\la)=s_\la(1,\dotsc,1)$   
	is positive whenever $\ell(\la)\leq r$, and that 
	$\dim(U_\la)=\exp(o(n))$ where the $o(n)$ is
	uniform in $\la$.
	Now consider the coefficients $c^\la_{\mu,\nu}$.
	These are known (see e.g.\ 
	\cite[Chapter 5, Proposition 3]{fulton:young}) to equal the size of a
	certain subset of semi-standard tableaux with shape 
	$\la\sm\mu$ filled with $\nu_1$ 1's, $\nu_2$ 2's, etc.
	In particular, $c^\la_{\mu,\nu}>0$ only if $\mu$ is contained in
	$\la$,  and then $\ell(\mu)\leq\ell(\la)\leq r$.  Since
	$c^\la_{\mu,\nu}=c^\la_{\nu,\mu}$ (see \cite{fulton:young} again)
	we also need $\ell(\nu)\leq r$ for $c^\la_{\mu,\nu}>0$.  
	The combinatorial description also gives the upper bound
	$c_{\mu,\nu}^\la\leq (n+1)^{r^2}=\exp(o(n))$ where the $o(n)$ is
	uniform in $\la,\mu,\nu$.

	We now turn to the remaining factors in \eqref{eq:Z-bi}.
	First, as one can see in \eqref{eq:dimV}, for fixed $r$ we have that
	$d_\mu$ is essentially a multinomial coefficient.  Thus (see
	e.g.\ \cite[pp. 14--15]{Bjo16} for details), we have
	\be\label{eq:d-asy}
	\textstyle
	\frac1n\log d_\mu=
	-\sum_{j=1}^r \tfrac{\mu_j}{n}\log\tfrac{\mu_j}{n}
	+O(\tfrac{\log n}{n}).
	\ee
	Next, from \eqref{eq:ct-sum} we have that
	\be\label{eq:ct-asy}
	\textstyle
	\ct(\la)=\tfrac{n^2}2\sum_{j=1}^r \big(\tfrac{\la_j}n\big)^2+O(n).
	\ee
	Taken altogether, these facts  mean 
	that we can write \eqref{eq:Z-bi} as 
	\be\label{eq:Z-bi-approx}
	Z_{\b,n}^\ab=
	\sum_{\la,\mu,\nu}
	\one\{c_{\mu,\nu}^\la>0\}  
	\exp\Big(n \Big\{
	\tilde F(\tfrac{\mu}{n},\tfrac{\nu}{n},\tfrac{\la}{n})+o(1)\Big\} \Big),
	\ee
	where $\la\vdash n$, $\mu\vdash m$ and $\nu\vdash n-m$, 
	all having $\leq r$ rows, and where
	\be\begin{split}\label{eq:phi}
		\tilde F(\vec x,\vec y,\vec z)=&\textstyle 
		-\sum_{j=1}^r x_j\log x_j
		-\sum_{j=1}^r y_j\log y_j\\
		&+\textstyle
		\frac\b2 \big[(a-c)\sum_{j=1}^r x_j^2
		+(b-c) \sum_{j=1}^r y_j^2
		+c \sum_{j=1}^r z_j^2\big].
	\end{split}\ee
	There is a necessary and sufficient condition for
	$c^\la_{\mu,\nu}>0$ 
	which is very useful for our purposes,
	known as \emph{Horn's conjecture}, proved by 
Knutson and Tao \cite{knutson-tao}.  It
	is best stated for our purposes in terms of eigenvalues of Hermitian matrices,
	as follows:
	$c_{\mu,\nu}^\la>0$ if and only if there are Hermitian $r\times r$
	matrices $X$ and $Y$ with eigenvalues $\mu_1,\dotsc,\mu_r$ 
	and $\nu_1,\dotsc,\nu_r$, respectively,
	such that $X+Y$ has eigenvalues $\la_1,\dotsc,\la_r$.   
	For information about this, see  e.g.\ \cite{fulton:horn}.
	We thus have
	\be\label{eq:Om-plus} 
	c_{\mu,\nu}^\la>0 \mbox{ if and only if }
	(\tfrac\mu n,\tfrac\nu n,\tfrac\la n)\in\Om_{m/n}^+
	\ee
	where $\Om_\rho^+$ is the set of triples $(\vec x,\vec y,\vec z)$ 
	such that there exist 
	positive semidefinite Hermitian matrices $X$, $Y$ with
	$\tr(X)=1-\tr(Y)=\rho$ having eigenvalues $x_1,\dotsc,x_r$ 
	and $y_1,\dotsc,y_r$, respectively, such that $Z=X+Y$ has 
	eigenvalues $z_1,\dotsc,z_r$.  
	
	From \eqref{eq:Z-bi-approx} and the fact that $\tilde F$
	is continuous in its arguments, we conclude that 
	\be\label{eq:lim1}
	\tfrac1n\log Z_{\b,n}^\ab\to \max_{(\vec x,\vec y,\vec z)\in\Om^+_\rho}
	\tilde F(\vec x,\vec y,\vec z).
	\ee
	See e.g.\ \cite[Section~3]{Bjo16} for a detailed argument in a similar
	setting.   
	Now note that if $X,Y,Z$ are as above, then 
	\be\textstyle
	\sum_{j=1}^r x_j^2=\tr(X^2),\qquad
	\sum_{j=1}^r y_j^2=\tr(Y^2),
	\ee
	and also
	\be\textstyle
	\sum_{j=1}^r z_j^2=\tr(Z^2)=\tr\big((X+Y)^2\big)
	=\tr(X^2)+\tr(Y^2)+2\,\tr(XY).
	\ee
	Thus
	\be
	(a-c)\sum_{j=1}^r x_j^2
	+(b-c) \sum_{j=1}^r y_j^2
	+c \sum_{j=1}^r z_j^2
	=
	\tr\big[ a X^2+b Y^2+ 2cXY \big].
	\ee
	So for $(\vec x,\vec y,\vec z)\in\Om_\rho$, we have that 
	\be \label{eq:phi-newer}
	\tilde F(\vec x,\vec y,\vec z)=
	\phi(X,Y):= S(X)+S(Y)+\tfrac\b2 
	\tr\big[ a X^2+b Y^2+ 2cXY  \big],
	\ee
	where $S$ is the von Neumann entropy 
	\be\label{eq:neumann}
	S(X)=-\tr(X \log X)=-\sum_{i=1}^r x_i\log x_i.
	\ee
	It follows that 
	\be\label{eq:maxXY}
	\tfrac1n\log Z_{n}^\ab(\b)\to \max_{X,Y}
	\phi(X,Y)
	\ee
	where the maximum is over positive definite Hermitian matrices $X,Y$
	with $\tr(X)=1-\tr(Y)=\rho$.
	
	The final step is to use the fact that for 
	positive semidefinite Hermitian matrices $X,Y$
	with fixed spectra $x_1,\dotsc,x_r$ and
	$y_1,\dotsc,y_r$, respectively, 
	ordered so that $x_1\geq x_2\geq \dotsb\geq x_r$ and
	$y_1\geq y_2\geq \dotsb\geq y_r$, we have the inequality
	\be\label{eq:trace-ineq}
	\sum_{j=1}^r x_j y_{r+1-j} \leq \tr[XY]\leq \sum_{j=1}^r x_j y_j,
	\ee
	see e.g.\ \cite[Prop.~9.H.1.g-h]{majorization}
	(we discuss this result in Appendix \ref{sec:trXY}).
	In particular, both the maximum and the minimum of $\tr[XY]$ are 
	attained when $X,Y$ are simultaneously diagonal.
	Since the other terms in $F(\vec x,\vec y)$ are symmetric under
	permuting the $x_i$ or the $y_i$, 
	the result follows.
\end{proof}

\subsection{Walled Brauer algebra:
	proof of Theorem \ref{thm:FE-WB}}
\label{sec:WB}

As noted above, our analysis of the model in \eqref{eq:H-wb} 
uses the walled Brauer
algebra.  We will now define this algebra, and collect some facts which allow us to approach a proof in a similar way to that of Theorem \ref{thm:FE-AB}.  An accessible introduction to the walled Brauer algebra is
given in \cite{nikitin}, and its Schur--Weyl duality is proved in
\cite{benkart-et-al}, at least for the range $r \ge n$. The extension to all $r,n$ is a straightforward extension of the work in \cite{benkart-et-al}.

Let us first define the (usual)
Brauer algebra.  Fix $n \in \NN, r \in\CC$. 
Arrange two rows each of $n$ labelled vertices, one above the
other.  We call a \textit{diagram} a graph on these $2n$ vertices, with
each vertex having degree one. Let $B_n$ be the set of such
diagrams.  The Brauer algebra $\BB_{n}(r)$ is the formal complex span of
$B_n$.  Multiplication of two diagrams is defined as follows.  Taking
two diagrams $g,h$, identify the upper vertices of $h$ with the lower
of $g$. Then form a new diagram by concatenation and removing any
closed loops, 
as in Figure \ref{fig:4:multiplicationofdiagramsexample}.
The product $gh$ is the concatenation, multiplied by 
$r^{\# \mathrm{loops}}$, where $\# \mathrm{loops}$ is the number of
loops removed. 

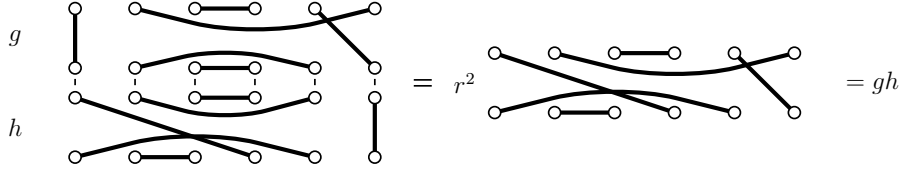
\begin{figure}[h]
	\centering
	\resizebox{0.95\textwidth}{!}{%
		\begin{tikzpicture}[scale=1]

		\draw[-, rounded corners=10pt, dashed, thick](2, 1) -- (2, 1.5);
		\draw[-, rounded corners=10pt, dashed, thick](0, 1) -- (0, 1.5);
		\draw[-, rounded corners=10pt, dashed, thick](1,1) -- (1, 1.5);
		\draw[-, rounded corners=10pt, dashed, thick](3, 1) -- (3, 1.5);
		\draw[-, rounded corners=10pt, dashed, thick](5,1) -- (5, 1.5);
		\draw[-, rounded corners=10pt, dashed, thick](4,1) -- (4, 1.5);

		\draw[font=\large] (-1, 0.5) node {$h$};
		\draw[-, rounded corners=20pt, line width=2pt, black](1, 1) -- (2.5, 0.6) -- (4, 1);
		\draw[-, rounded corners=10pt, line width=2pt, black](0, 1) -- (3,0);
		\draw[-, rounded corners=10pt, line width=2pt, black](2, 1) -- (3, 1);
		\draw[-, rounded corners=10pt, line width=2pt, black](2, 0) -- (1, 0);
		\draw[-, rounded corners=30pt, line width=2pt, black](4,0) -- (2,0.5) -- (0, 0);
		\draw[-, rounded corners=10pt, line width=2pt, black](5,0) -- (5, 1);
		
		\draw[thick, fill=white] (0, 1) circle (0.1cm);
		\draw[thick, fill=white] (1, 1) circle (0.1cm);
		\draw[thick, fill=white] (2, 1) circle (0.1cm);
		\draw[thick, fill=white] (3, 1) circle (0.1cm);
		\draw[thick, fill=white] (4, 1) circle (0.1cm);
		\draw[thick, fill=white] (5, 1) circle (0.1cm);
		\draw[thick, fill=white] (0, 0) circle (0.1cm);
		\draw[thick, fill=white] (1, 0) circle (0.1cm);
		\draw[thick, fill=white] (2, 0) circle (0.1cm);
		\draw[thick, fill=white] (3, 0) circle (0.1cm);
		\draw[thick, fill=white] (4, 0) circle (0.1cm);
		\draw[thick, fill=white] (5, 0) circle (0.1cm);

		\draw[font=\large] (-1, 2) node {$g$};
		\draw[-, rounded corners=10pt, line width=2pt, black](2, 2.5) -- (3, 2.5);
		\draw[-, rounded corners=10pt, line width=2pt, black](0, 2.5) -- (0, 1.5);
		\draw[-, rounded corners=20pt, line width=2pt, black](4,1.5) -- (2.5,1.85) -- (1, 1.5);
		\draw[-, rounded corners=10pt, line width=2pt, black](2, 1.5) -- (3, 1.5);
		\draw[-, rounded corners=10pt, line width=2pt, black](4,2.5) -- (5, 1.5);
		\draw[-, rounded corners=30pt, line width=2pt, black](5,2.5) -- (3, 2) -- (1, 2.5);
		
		\draw[thick, fill=white] (0, 1.5) circle (0.1cm);
		\draw[thick, fill=white] (1, 1.5) circle (0.1cm);
		\draw[thick, fill=white] (2, 1.5) circle (0.1cm);
		\draw[thick, fill=white] (3, 1.5) circle (0.1cm);
		\draw[thick, fill=white] (4, 1.5) circle (0.1cm);
		\draw[thick, fill=white] (5, 1.5) circle (0.1cm);
		\draw[thick, fill=white] (0, 2.5) circle (0.1cm);
		\draw[thick, fill=white] (1, 2.5) circle (0.1cm);
		\draw[thick, fill=white] (2, 2.5) circle (0.1cm);
		\draw[thick, fill=white] (3, 2.5) circle (0.1cm);
		\draw[thick, fill=white] (4, 2.5) circle (0.1cm);
		\draw[thick, fill=white] (5, 2.5) circle (0.1cm);

		\draw[-, rounded corners=10pt, thick](5.65,1.3) -- (5.95, 1.3);
		\draw[-, rounded corners=10pt, thick](5.65,1.2) -- (5.95, 1.2);

		\draw[font=\large] (13.3, 1.25) node {$=gh$};
		\draw[font=\large] (6.5, 1.25) node {$r^2$};
		\draw[-, rounded corners=10pt, line width=2pt, black](10, 1.75) -- (9, 1.75);
		\draw[-, rounded corners=10pt, line width=2pt, black](7, 1.75) -- (10, 0.75);
		\draw[-, rounded corners=30pt, line width=2pt, black](8, 1.75) -- (10, 1.25) -- (12, 1.75);
		\draw[-, rounded corners=10pt, line width=2pt, black](8, 0.75) -- (9, 0.75);
		\draw[-, rounded corners=30pt, line width=2pt, black](7, 0.75) -- (9, 1.25) -- (11, 0.75);
		\draw[-, rounded corners=10pt, line width=2pt, black](12,0.75) -- (11, 1.75);
		
		\draw[thick, fill=white] (7, 0.75) circle (0.1cm);
		\draw[thick, fill=white] (8, 0.75) circle (0.1cm);
		\draw[thick, fill=white] (9, 0.75) circle (0.1cm);
		\draw[thick, fill=white] (10, 0.75) circle (0.1cm);
		\draw[thick, fill=white] (11, 0.75) circle (0.1cm);
		\draw[thick, fill=white] (12, 0.75) circle (0.1cm);
		\draw[thick, fill=white] (7, 1.75) circle (0.1cm);
		\draw[thick, fill=white] (8, 1.75) circle (0.1cm);
		\draw[thick, fill=white] (9, 1.75) circle (0.1cm);
		\draw[thick, fill=white] (10, 1.75) circle (0.1cm);
		\draw[thick, fill=white] (11, 1.75) circle (0.1cm);
		\draw[thick, fill=white] (12, 1.75) circle (0.1cm);
		
		\end{tikzpicture}
	}
	\caption{Two diagrams $g$ and $h$ (left), and their product (right). The concatenation contains two loops, so we multiply the concatenation with middle vertices removed by $r^2$.}
	\label{fig:4:multiplicationofdiagramsexample}
	
\end{figure}

The walled Brauer algebra is a subalgebra of $\BB_{n}(r)$.  Let
$m\leq n$.  Returning to the $2n$ labelled vertices, draw
a line (a ``wall'') separating the leftmost $2m$ vertices and the
rightmost $2(n-m)$.  Let $B_{n,m}$ be the set of diagrams in $B_n$ with
the condition that any edge connecting two upper vertices or two lower
vertices \textit{must} cross the wall, and any edge connecting an
upper vertex and a lower vertex \textit{must not} cross the wall; see Figure \ref{fig:WB-diagram}. The
walled Brauer algebra $\BB_{n,m}(r)$ is the span of $B_{n,m}$, with
multiplication as in the Brauer algebra.

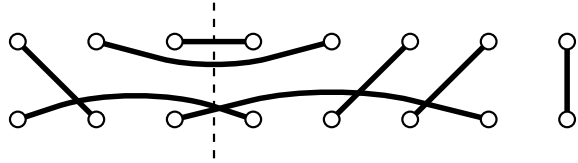
\begin{figure}[h]
	\centering
	\resizebox{0.6\textwidth}{!}{%
		\begin{tikzpicture}[scale=1]
		
		\draw[-, rounded corners=10pt, dashed, thick](2.5, -0.5) -- (2.5, 1.5);
		
		\draw[-, rounded corners=20pt, line width=2pt, black](1, 1) -- (2.5, 0.6) -- (4, 1);
		\draw[-, rounded corners=10pt, line width=2pt, black](0, 1) -- (1,0);
		\draw[-, rounded corners=10pt, line width=2pt, black](2, 1) -- (3, 1);
		\draw[-, rounded corners=30pt, line width=2pt, black](2, 0) -- (4,0.5) -- (6, 0);
		\draw[-, rounded corners=30pt, line width=2pt, black](3,0) -- (1.5,0.5) -- (0, 0);
		\draw[-, rounded corners=10pt, line width=2pt, black](4,0) -- (5, 1);
		\draw[-, rounded corners=10pt, line width=2pt, black](5,0) -- (6, 1);
		\draw[-, rounded corners=10pt, line width=2pt, black](7,0) -- (7, 1);
		
		\draw[thick, fill=white] (0, 1) circle (0.1cm);
		\draw[thick, fill=white] (1, 1) circle (0.1cm);
		\draw[thick, fill=white] (2, 1) circle (0.1cm);
		\draw[thick, fill=white] (3, 1) circle (0.1cm);
		\draw[thick, fill=white] (4, 1) circle (0.1cm);
		\draw[thick, fill=white] (5, 1) circle (0.1cm);
		\draw[thick, fill=white] (6, 1) circle (0.1cm);
		\draw[thick, fill=white] (7, 1) circle (0.1cm);
		\draw[thick, fill=white] (0, 0) circle (0.1cm);
		\draw[thick, fill=white] (1, 0) circle (0.1cm);
		\draw[thick, fill=white] (2, 0) circle (0.1cm);
		\draw[thick, fill=white] (3, 0) circle (0.1cm);
		\draw[thick, fill=white] (4, 0) circle (0.1cm);
		\draw[thick, fill=white] (5, 0) circle (0.1cm);
		\draw[thick, fill=white] (6, 0) circle (0.1cm);
		\draw[thick, fill=white] (7, 0) circle (0.1cm);

		\end{tikzpicture}
	}
	\caption{A diagram in the basis $B_{8,3}$ of the walled Brauer algebra $\BB_{8,3}(r)$. Notice that all edges connecting two upper vertices (or two lower) cross the wall, and all edges connecting an upper vertex to a lower vertex do not.}
	\label{fig:WB-diagram}
	
\end{figure}

Some useful representation-theoretic facts follow.
First, the group algebra $\CC[S_m \times S_{n-m}]$ is a
subalgebra of $\BB_{n,m}(r)$ whose basis $S_m \times S_{n-m}$
consists of  those
diagrams with no edges crossing the wall.
As above, we let $(i,j)$ denote the transposition
exchanging $i$ and $j$.  Note that in the walled Brauer algebra, we must
have $1 \le i,j \le m$ or $m+1 \le i,j \le n$. 
For $1 \le i \le m < j \le n$, let $(\overline{i,j})$ denote the 
diagram with all edges
vertical, except that the $i^{\mathrm{th}}$ and $j^{\mathrm{th}}$ 
upper vertices are
connected, and the $i^{\mathrm{th}}$ and $j^{\mathrm{th}}$ 
lower vertices are
connected; see Figure \ref{fig:transpositions}. The elements $(i,j)$ and $(\overline{i,j})$ generate the
walled Brauer algebra.

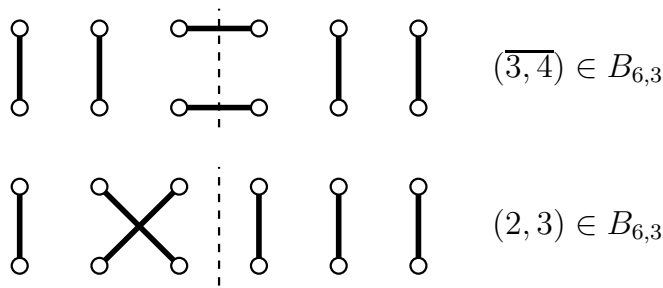
\begin{figure}[h]
	\centering
	\resizebox{0.7\textwidth}{!}{%
		\begin{tikzpicture}[scale=1]
		
		\draw[-, rounded corners=10pt, dashed, thick](2.5, -0.25) -- (2.5, 1.25);
		\draw[-, rounded corners=10pt, dashed, thick](2.5, 1.75) -- (2.5, 3.25);
		
		\draw[font=\large] (7, 0.5) node {$(2,3)\in B_{6,3}$};
		\draw[-, rounded corners=20pt, line width=2pt, black](1, 1) -- (2, 0);
		\draw[-, rounded corners=10pt, line width=2pt, black](0, 1) -- (0,0);
		\draw[-, rounded corners=10pt, line width=2pt, black](2, 1) -- (1, 0);
		\draw[-, rounded corners=10pt, line width=2pt, black](3, 0) -- (3, 1);
		\draw[-, rounded corners=30pt, line width=2pt, black](4,0)  -- (4, 1);
		\draw[-, rounded corners=10pt, line width=2pt, black](5,0) -- (5, 1);
		
		\draw[thick, fill=white] (0, 1) circle (0.1cm);
		\draw[thick, fill=white] (1, 1) circle (0.1cm);
		\draw[thick, fill=white] (2, 1) circle (0.1cm);
		\draw[thick, fill=white] (3, 1) circle (0.1cm);
		\draw[thick, fill=white] (4, 1) circle (0.1cm);
		\draw[thick, fill=white] (5, 1) circle (0.1cm);
		\draw[thick, fill=white] (0, 0) circle (0.1cm);
		\draw[thick, fill=white] (1, 0) circle (0.1cm);
		\draw[thick, fill=white] (2, 0) circle (0.1cm);
		\draw[thick, fill=white] (3, 0) circle (0.1cm);
		\draw[thick, fill=white] (4, 0) circle (0.1cm);
		\draw[thick, fill=white] (5, 0) circle (0.1cm);
		
		\draw[font=\large] (7, 2.5) node {$(\overline{3,4}) \in B_{6,3}$};
		\draw[-, rounded corners=10pt, line width=2pt, black](2, 3) -- (3, 3);
		\draw[-, rounded corners=10pt, line width=2pt, black](0, 3) -- (0, 2);
		\draw[-, rounded corners=20pt, line width=2pt, black](4,2) -- (4, 3);
		\draw[-, rounded corners=10pt, line width=2pt, black](2, 2) -- (3, 2);
		\draw[-, rounded corners=10pt, line width=2pt, black](1,3) -- (1, 2);
		\draw[-, rounded corners=30pt, line width=2pt, black](5,3) -- (5, 2);
		
		\draw[thick, fill=white] (0, 2) circle (0.1cm);
		\draw[thick, fill=white] (1, 2) circle (0.1cm);
		\draw[thick, fill=white] (2, 2) circle (0.1cm);
		\draw[thick, fill=white] (3, 2) circle (0.1cm);
		\draw[thick, fill=white] (4, 2) circle (0.1cm);
		\draw[thick, fill=white] (5, 2) circle (0.1cm);
		\draw[thick, fill=white] (0, 3) circle (0.1cm);
		\draw[thick, fill=white] (1, 3) circle (0.1cm);
		\draw[thick, fill=white] (2, 3) circle (0.1cm);
		\draw[thick, fill=white] (3, 3) circle (0.1cm);
		\draw[thick, fill=white] (4, 3) circle (0.1cm);
		\draw[thick, fill=white] (5, 3) circle (0.1cm);
		
		\end{tikzpicture}
	}
	\caption{Examples of the elements $(\overline{i,j})$ and the transpositions $(i,j)$.}
	\label{fig:transpositions}
	
\end{figure}

Next, the irreducible representations of $\BB_{n,m}(r)$ are
indexed by 
\be 
\{ (\lambda, \mu) \ | \ \lambda
\vdash m-t, \ \mu \vdash n-m-t, \ 
t=0, \dots, \min\{m,n-m\} \ \}, 
\ee
where $\lambda$ and $\mu$ are partitions (see Proposition 2.4 of \cite{cox07}).
Henceforth, we will use the notation $\hat m=\min\{m,n-m\}$
so that the standing condition on $t$ is that
$t\in\{0,1,\dotsc,\hat m\}$.
The element 
\be \label{eq:JM}
J_{n,m}=
\sum_{\substack{1 \le i<j \le  m \\ m< i<j <n }}
(i,j) - \sum_{1 \le i\le m < j \le n} (\overline{i,j})
\ee 
is central in $\BB_{n,m}(r)$, and acts as the scalar
$\ct(\lambda) + \ct(\mu) - rt$ on the irreducible representation
$(\lambda, \mu)$, where $\lambda \vdash m-t$, 
$\mu \vdash n-m-t$ and $\ct(\cdot)$ denotes the content defined in
\eqref{eq:ct-sum} (a consequence of, for example, Lemma 4.1 of \cite{cox07}).

The walled Brauer algebra, like the symmetric group
algebra, has a Schur--Weyl duality with the general linear group.
To describe this, let us first recall some facts about representations 
of the general linear group $\GL_r(\CC)$.
The irreducible \emph{rational} representations of
$\GL_r(\CC)$ are indexed by their highest weights, which are $r$-tuples 
$\nu =(\nu_1 \ge \cdots \ge \nu_r) \in \ZZ^r$.  Such a tuple can be
equivalently written as a pair $\nu=[\lambda, \mu]$
of partitions $\lambda, \mu$ 
with $\ell(\lambda) + \ell(\mu) \le r$,
by letting $\nu_i =[\lambda, \mu]_i = \lambda_i - \mu_{r-i+1}$ for 
$i=1, \dots, r$. 
Note that at most one of the terms $\lambda_i$ or $\mu_{r-i+1}$ is
non-zero for each $i$, due to the constraint 
$\ell(\lambda) + \ell(\mu) \le r$, thus $\nu$ uniquely
determines $\la$ and $\mu$. See Figure \ref{fig:youngdiagrams} for an illustration.

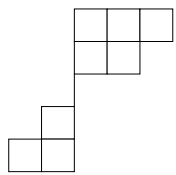
\begin{figure}[h]
	\centering
	\resizebox{0.2\textwidth}{!}{%
		\begin{tikzpicture}[scale=1]
		\draw[font=\large] (0, 0) node {\gyoung(::;;;,::;;,::'01,:;,;;,)};
		\end{tikzpicture}
	}
	\caption{The $r$-tuple $\nu=(3,2,0,-1,-2)$ illustrated in the style of a Young diagram, where negative entries are shown by boxes to the left of the main vertical line. Here $r=5$. From the figure it is straightforward to see that $\nu=[\la,\mu]$, where $\la=(3,2)$ and $\mu=(2,1)$.}
	\label{fig:youngdiagrams}
\end{figure}

We write $U_{[\lambda, \mu]}$ for the corresponding
irreducible $\GL_r(\CC)$-module.  These rational representations are
closely related to the polynomial representations $U_\la$ appearing in
\eqref{eq:swd}; the polynomial representations are the rational representations with non-negative $r$-tuple $\nu$. One can also relate the rational and polynomial representations by the 
Pieri-rule \cite{stembridge}.  Indeed,
writing $\det(\cdot)$ for the determinant
representation of $\GL_r(\CC)$, which has highest weight 
$(1,1,\dotsc,1)$ and character $x_1x_2\dotsb x_r$, we have that 
$\det^{\otimes k}\otimes U_{\nu} =U_{\nu+\underline k}$
where $\underline k=(k,k,\dotsc,k)$.   For $k=\mu_1$ we  have that 
$U_{[\lambda, \mu]+\underline{\mu_1}}$  is a polynomial
representation.  It follows from this and \eqref{eq:GL-char} 
that the character of $U_{[\lambda, \mu]}$ is
\be\label{eq:GL-char-2}
\chi_{U_{[\lambda, \mu]}}[g]=
\frac{s_{[\lambda, \mu]+\underline{\mu_1}} (x_1,\dotsc,x_r)}
{(x_1x_2\dotsb x_r)^{\mu_1}}
=\frac{\det[x_i^{[\la,\mu]_j+r-j}]_{i,j=1}^r}
{\prod_{1\leq i<j\leq r}(x_i-x_j)},
\ee
where $x_1,\dotsc,x_r$ are the eigenvalues of $g$.

We can now state the Schur-Weyl duality for the walled Brauer algebra and the general linear group. Let $\GL_r(\CC)$ act on 
$\VV = (\CC^r)^{\otimes n}=(\CC^r)^{\otimes m}\otimes(\CC^r)^{\otimes
	(n-m)}$ 
as $m$ tensor powers of
its defining representation, and $n-m$ tensor powers of the dual of its
defining representation (multiplication by the
inverse transpose):
\[
g(v_1\otimes \dotsb\otimes v_m
\otimes v_{m+1}\otimes\dotsb\otimes v_n)=
g(v_1)\otimes \dotsb\otimes g(v_m)
\otimes g^{-\intercal}(v_{m+1})\otimes\dotsb\otimes 
g^{-\intercal} (v_n).
\]
Let $\BB_{n,m}(r)$ act on $\VV$ by sending
$(i,j)$ to the transposition operator $T_{i,j}$, and
$(\overline{i,j})$ to $Q_{i,j}$ \eqref{eq:Q}. Then, as a representation of 
$\CC[GL_r(\CC)] \otimes \BB_{n,m}(r)$, 
\be\label{eq:swd-WB} 
\VV =
\bigoplus_{t=0}^{\hat m}\bigoplus_{\substack{\lambda\vdash m-t
		\\ \mu \vdash n-m-t \\ \ell(\lambda) + \ell(\mu) \le r}} U_{[\lambda, \mu]}
\otimes V_{(\lambda, \mu)}, 
\ee 
with $V_{(\lambda, \mu)}$
irreducible $\BB_{n,m}(r)$-representations as above (as noted above, this is a straightforward extension of the work in \cite{benkart-et-al}).\\

Notice now that our Hamiltonian \eqref{eq:H-wb} 
can be rewritten as 
\be
\begin{split} 
	H_n^\wb=-\frac1n  \Big( &
	(a+c) \sum_{1\leq i<j\leq m} T_{i,j} +
	(b+c) \sum_{m+1\leq i<j\leq n} T_{i,j} 
	- c J_{n,m}
	\Big),
\end{split} 
\ee
where $J_{n,m}$ is the central element given in \eqref{eq:JM}. 
Now in an identical way to how we developed equation
(\ref{eq:Z-bi}), we have 
\be\label{eq:trace-WB}
\begin{split} \mathrm{tr}_{\VV}[e^{-\beta H^{\wb}_n}] =&
	\sum_{\substack{\pi \vdash m \\ \tau \vdash n-m}} 
	\sum_{t=0}^{\hat m} \sum_{\substack{\lambda\vdash m-t \\ \mu
			\vdash n-m-t \\ \ell(\lambda) + \ell(\mu) \le r}} 
	\dim(U_{[\lambda, \mu]})
	b^{n,m,r}_{(\lambda, \mu), (\pi, \tau)} 
	d_\pi d_\tau
	\\ &\cdot \exp \big( \tfrac\beta n \big[ (c+a)\ct(\pi) +
	(c+b)\ct(\tau) - c(\ct(\lambda) + \ct(\mu) -rt) \big] \big),
\end{split} 
\ee 
where $b^{n,m,r}_{(\lambda, \mu), (\pi, \tau)}$ is
the branching coefficient from $\CC[S_m \times S_{n-m}]$ to
$\BB_{n,m}(r)$, i.e.\ the multiplicity of the 
$\CC[S_m\times S_{n-m}]$-module $V_\pi\otimes V_\tau$
in $V_{(\la,\mu)}$ when the latter is regarded as a 
$\CC[S_m \times S_{n-m}]$-module.
These branching coefficients 
play the same role as the Littlewood--Richardson
coefficient did in the {\ab}-model. 
Our next step is to determine when
$b^{n,m,r}_{(\lambda, \mu), (\pi, \tau)}$ is strictly positive.

\begin{lemma}\label{lem:WB-COEFS} 
	The branching coefficient
	$b^{n,m,r}_{(\lambda, \mu), (\pi, \tau)}$ is strictly positive if and
	only if 
	there exist $r \times r$ Hermitian matrices $X, Y, Z$ with
	respective spectra $\pi, \tau, [\lambda, \mu]$, such
	that $X-Y=Z$.
\end{lemma} 

Note that the parameter $t$ is encoded the branching coefficient, in the sense that  $b^{n,m,r}_{(\lambda, \mu), (\pi, \tau)}>0$ implies that $\la\vdash m-t=|\pi|-t$ and $\mu\vdash n-m-t=|\tau|-t$ for some $0\le t\le \hat{m}$. 
To see how $t$ appears from the Hermitian matrices, 
assume for the sake of argument that $X$ and $Y$ commute. Then, for
each $i$, 
 $[\la,\mu]_i=\pi_j-\tau_k$, for some $j,k$. Figure \ref{fig:t-parameter-Hermitian} then illustrates via an example how it follows that $\la\vdash m-t=|\pi|-t$ and $\mu\vdash n-m-t=|\tau|-t$ for some $0\le t\le \hat{m}$.

%

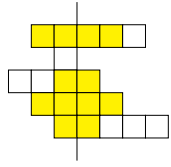
\begin{figure}[h]
	\centering
        \resizebox{0.2\textwidth}{!}{%
        \newcommand\ylw{\Yfillcolour{yellow}}
        \newcommand\wht{\Yfillcolour{white}}
          \begin{tikzpicture}[scale=1]
            \draw[font=\large] (0, 0) node {\gyoung(:::'01,:!\ylw;;;;!\wht;,::;,;;!\ylw;;,:;;;;,::;;!\wht;;;,:::'01,)};
          \end{tikzpicture}
	}
	\caption{The spectra $\pi=(3,0,1,2,4)$ and $\tau=(2,1,3,2,1)$, respectively of $X$ and $Y$ (simultaneously diagonalised), displayed in the style of Young diagrams, either side of the main vertical line. The spectrum of $Z=X-Y$ is $(1,-1,-2,0,3)$ (and so when ordered becomes $[\la,\mu]=(3,1,0,-1,-2)$). The yellow boxes are those eliminated in the subtraction. Naturally there are the same number either side of the main vertical; this is the parameter $0\le t\le \min{|\pi|,|\tau|}$. In this example, $t=6$.}
	\label{fig:t-parameter-Hermitian}
\end{figure}

The first step to prove Lemma \ref{lem:WB-COEFS} is another lemma, analogous to the well
known fact that the Littlewood--Richardson coefficients are both the
branching coefficients from $\CC[S_m \times S_{n-m}]$ 
to $\CC[S_n]$, and
the coefficients of the decomposition of the tensor product of two
irreducible polynomial representations of $\GL_r(\CC)$.

\begin{lemma} \label{lem:wbbranch}
	Let $\pi, \tau, \lambda, \mu$ denote partitions 
	with at most
	$r$ parts, with $\ell(\lambda)+ \ell(\mu) \le r$,
	and $U_\pi, U_{[\varnothing, \tau]}, U_{[\lambda, \mu]}$ denote irreducible rational representations of $\GL_r(\CC)$. Let
	\be\label{eq:tensor-of-GLr-reps} 
	U_\pi \otimes U_{[\varnothing, \tau]}
	= \bigoplus_{\substack{\lambda, \mu \\ \ell(\lambda) + \ell(\mu)\le r}}
	\hat{b}^{n,m,r}_{[\lambda, \mu], (\pi, \tau)} U_{[\lambda, \mu]}.
	\ee 
	Then $\hat{b}^{n,m,r}_{[\lambda, \mu], (\pi, \tau)} =
	b^{n,m,r}_{(\lambda, \mu), (\pi, \tau)}$.
\end{lemma}
\begin{proof} 
	This is proved using Schur--Weyl
	duality.  We restrict (\ref{eq:swd-WB}) to 
	$\CC [GL_r(\CC)] \otimes\CC[S_m \times S_{n-m}]$ 
	to see that 
	\be 
	\VV =
	\bigoplus_{t=0}^{\hat m} \bigoplus_{\substack{\lambda\vdash m-t
			\\ \mu \vdash n-m-t \\ \ell(\lambda) + \ell(\mu) \le r}}
	\bigoplus_{\substack{\pi \vdash m \\ \tau \vdash n-m \\ 
			\ell(\pi),\ell(\tau) \le r}} 
	b^{n,m,r}_{(\lambda, \mu), (\pi, \tau)} 
	U_{[\lambda,\mu]} \otimes 
	(V_\pi\otimes V_\tau).  
	\ee 
	On the other hand, the
	Schur--Weyl duality between $GL_r(\CC)\times GL_r(\CC)$ 
	and $\CC[S_m \times S_{n-m}]$ is
	\be 
	\VV = \bigoplus_{\substack{\pi \vdash m \\ \tau \vdash n-m \\
			\ell(\pi), \ell(\tau) \le r}} 
	(U_{\pi} \otimes U_{[\varnothing,\tau]}) \otimes
	(V_\pi\otimes V_\tau).  
	\ee 
	Expanding $U_{\pi} \otimes U_{[\varnothing,\tau]}$
	as in \eqref{eq:tensor-of-GLr-reps} 
	and equating coefficients from the two
	equations above gives the result.
\end{proof}

\begin{proof}[Proof of Lemma \ref{lem:WB-COEFS}] 
	We take equation
	(\ref{eq:tensor-of-GLr-reps}) and modify it using the Pieri rule: 
	\be U_\pi \otimes U_{[\varnothing,
		\tau] + \underline{\tau_1}} = 
	\bigoplus_{\substack{\lambda, \mu \\
			\ell(\lambda) + \ell(\mu) \le r}} 
	\hat{b}^{n,m,r}_{[\lambda, \mu], (\pi,
		\tau)} U_{[\lambda, \mu]+ \underline{\tau_1}}.  
	\ee 
	Now the highest
	weights appearing on both sides have no negative parts, so by 
	Lemma \ref{lem:wbbranch} and the
	Littlewood--Richardson Rule, 
	\be
	b^{n,m,r}_{(\la,\mu),(\pi,\tau)}=
	\hat{b}^{n,m,r}_{[\lambda, \mu], (\pi,\tau)} 
	= c^{[\lambda, \mu]+ \underline{\tau_1}}_{\pi, [\varnothing,
		\tau] + \underline{\tau_1}}.
	\ee  
	We know from Horn's inequalities that
	$c^{[\lambda, \mu]+ \underline{\tau_1}}_{\pi, [\varnothing, \tau] +
		\underline{\tau_1}} >0$ if and only if there exist $r \times r$
	Hermitian $\bar{X}, \bar{Y}, \bar{Z}$ with respective spectra $\pi,
	[\varnothing, \tau] + \underline{\tau_1}$ and $[\lambda, \mu]+
	\underline{\tau_1}$ such that $\bar{X} + \bar{Y} = \bar{Z}$. Now it is
	straightforward to show that such matrices exist if and only if there
	exist $r \times r$ Hermitian $X, Y, Z$ with respective spectra 
	$\pi$, $\tau$ and  $[\lambda, \mu]$ such that $X-Y=Z$. 
	Indeed, let $X= \bar{X}$, $Y = -\bar{Y} +\tau_1 \Id$, 
	and $Z = \bar{Z} - \tau_1 \Id$
	for the first implication, and similarly for the reverse implication.
\end{proof}

We can now return to equation (\ref{eq:trace-WB}). Using similar
workings as in Section \ref{sec:AB}, we let $m,n \to
\infty$ such that $m/n \to \rho \in (0,1)$, 
$\pi/n \to \vec x$, $\tau/n \to \vec y$ and 
$[\lambda,\mu]/n \to \vec z$.
Note that $\vec z$ can now have negative entries,
and that from \eqref{eq:ct-sum}
\be
\frac{\ct(\lambda) + \ct(\mu) -rt}{n^2}=
\sum_{i=1}^r \big((\tfrac{\la_i}{n})^2
+(-\tfrac{\mu_i}{n})^2\big)+o(1)
=\sum_{i=1}^r \big(\tfrac{[\la,\mu]_i}{n}\big)^2+o(1).
\ee 
We find that 
\be\label{eq:Z-WB-part2}
Z^{\wb}_n(\b)=
\sum_{\substack{\pi \vdash m \\ \tau \vdash n-m}} 
\sum_{\substack{\lambda,\mu \\ 
		(\pi/n,\tau/n,[\la,\mu]/n)\in\Om_{m/n}^-
	}} 
	\exp\Big(n\Big\{\tilde G(\tfrac\pi n,\tfrac\tau n,\tfrac{[\la,\mu]} n)+o(1)\Big\}\Big),
	\ee
	where $\Om_\rho^-$ is the set of triples of  $r$-tuples 
	$\vec x,\vec y,\vec z$
	such that $x_1,\dots, x_r\geq 0$,
	$y_1,\dotsc,y_r\geq 0$,
	$\sum_{i=1}^r x_i = \rho = 1- \sum_{i=1}^r y_i$, and there exist 
	$r\times r$ Hermitian matrices $X,Y,Z$ with respective spectra 
	$\vec x,\vec y,\vec z$
	such that $X-Y=Z$, and where
	\be\label{eq:tilde-G}
	\tilde G(\vec x,\vec y,\vec z)
	=\sum_{i=1}^r \big[\tfrac{\b}{2}((a+c)x_i^2 + (b+c)y_i^2 - cz_i^2) 
	-x_i\log x_i - y_i \log y_i\big].
	\ee 
	Notice that the sum over $t$ appearing in \eqref{eq:trace-WB} is hidden in \eqref{eq:Z-WB-part2}, as it is implicit in the definition of $\Om_\rho^-$, due to our remark after the statement of Lemma \ref{lem:WB-COEFS}. Therefore
	\be
	\begin{split} \Phi^{\wb}_\beta(a,b,c) := \lim_{n\to \infty}
		\frac{1}{n}\log Z^{\wb}_n(\b)= 
		\max_{(\vec x, \vec y,\vec z)\in\Om_\rho^-} 
		\tilde G(\vec x,\vec y,\vec z).
	\end{split} 
	\ee 
	As in \eqref{eq:phi-newer} and \eqref{eq:maxXY}, we can
	rewrite this in terms of the matrices $X$ and $Y$:
	\be\label{eq:FE-WB-matrices}
	\begin{split} \Phi^{\mathrm{WB}}_\beta(a,b,c) = \max_{X,Y}
		\big[S(X) + S(Y)+
		\tfrac{\beta}{2} \big( a\,\tr[X^2] + b\, \tr[Y^2] +2c \,\tr[XY] \big) 
		\big],
	\end{split} \ee 
	where now the maximum is only over $r \times r$ Hermitian
	matrices $X, Y$ with respective spectra 
	$\vec x,\vec y$ as above.  This is the same as \eqref{eq:maxXY}, and this completes the proof of Theorem \ref{thm:FE-WB}.
	\qed
	
	\subsection{Correlation functions:
		proof of Theorem \ref{thm:TS-twoblock}}
	
	Let us prove the result for the \ab-model first. 
	We use \eqref{eq:eW-bi} and the argument leading up to
	\eqref{eq:Z-bi-approx} to get that, as $n\to\oo$,
	\be\label{eq:total_spin_working}
	\begin{split}
		\big\langle \exp &\big\{
		\tfrac 1n {\textstyle\sum}_{i=1}^n W_i \big\} 
		\big\rangle_{\beta,n}^{\ab} 
		=\\&\frac{\sum_{\la,\mu,\nu}
			\one\{c_{\mu,\nu}^\la>0\}  
			\frac{s_\la(e^{w_1/n},\dotsc,e^{w_r/n})}{s_\la(1,\dotsc,1)}
			\exp\big(n \big\{
			\tilde F(\tfrac{\mu}{n},\tfrac{\nu}{n},\tfrac{\la}{n})+o(1)\big\} \big)}
		{\sum_{\la,\mu,\nu}
			\one\{c_{\mu,\nu}^\la>0\}  
			\exp\big(n \big\{
			\tilde F(\tfrac{\mu}{n},\tfrac{\nu}{n},\tfrac{\la}{n})+o(1)\big\} \big)},
	\end{split}
	\ee
	where  $\tilde{F}$ is as in \eqref{eq:phi}.
	Both sums on the right-hand-side are over $\la\vdash n$,
	$\mu\vdash m$ and $\nu\vdash n-m$, all having at most $r$ parts, and 
	in the numerator we have multiplied and divided by 
	$\dim(U_\la)=s_\la(1,\dotsc,1)$ in order that the $o(1)$ terms in the
	exponents are exactly equal.  Then the arguments of 
	\cite[Section~6]{BFU20} apply, meaning that 
	\be
	\lim_{n\to\oo}
	\big\langle \exp \big\{
	\tfrac 1n{\textstyle\sum}_{i=1}^n W_i \big\} \big\rangle_{\beta,n}^{\ab} 
	=\lim_{\la/n\to \vec z^\star}
	\frac{s_\la(e^{w_1/n},\dotsc,e^{w_r/n})}{s_\la(1,\dotsc,1)},
	\ee
	where $\vec z^\star=(z^\star_1,\dotsc,z^\star_r)$ 
	lists the eigenvalues of $X+Y$ where
	$X,Y$ are the Hermitian 
	matrices which maximise the right-hand-side of \eqref{eq:maxXY}.  
	But we know from
	\eqref{eq:trace-ineq} that the maximum is attained when $X,Y$ are
	simultaneously diagonal, with ordering of eigenvalues decreasing for
	both $X$ and $Y$ if $c>0$, respectively decreasing for $X$ and
	increasing for $Y$ if $c<0$.  Then clearly the eigenvalues of $Z=X+Y$
	are the sums of the eigenvalues of $X$ and of $Y$, ordered
	appropriately, giving $z^\star$ as in \eqref{eq:TS-z}.

	Turning to the \wb-model, very similarly to equation
	(\ref{eq:total_spin_working}) we have 
	\be\begin{split}
		&\big\langle \exp \big\{
		\tfrac 1n \big({\textstyle\sum_{i=1}^m W_i 
			-\sum_{i=m+1}^n} W_i^\intercal\big)\big\} 
		\big\rangle_{\beta,n}^{\wb} \\
		&\quad=
		\frac{
			\sum_{\lambda,\mu,\pi,\tau} \one\{b^{n,m,r}_{[\lambda, \mu], (\pi, \tau)}>0\}  
			\frac{\chi_{U_{[\la,\mu]}}(e^{W/n})}{\dim(U_{[\la,\mu]})}
			\exp\Big(n \Big\{
			\tilde G(\tfrac{\pi}{n},\tfrac{\tau}{n},\tfrac{[\la,\mu]}{n})+o(1)\Big\}
			\Big)
		}
		{
			\sum_{\lambda,\mu,\pi,\tau} \one\{b^{n,m,r}_{[\lambda, \mu], (\pi, \tau)}>0\}  
			\exp\Big(n \Big\{
			\tilde G(\tfrac{\pi}{n},\tfrac{\tau}{n},\tfrac{[\la,\mu]}{n})+o(1)\Big\}
			\Big),
		}
	\end{split}\ee
	where once again the $o(1)$ terms
	in the exponents are exactly equal and 
	$\tilde G$ is given in \eqref{eq:tilde-G}.
	The arguments of
	\cite[Section~6]{BFU20} apply once again, meaning
	that by \eqref{eq:GL-char-2} the limit equals  
	\be
	\lim_{[\la,\mu]/n\to z^\dagger}
	\frac{\chi_{U_{[\la,\mu]}}(e^{W/n})}{\dim(U_{[\la,\mu]})},
	\ee
	where this time, 
	$(\vec x^\star,\vec y^\star,\vec z^\dagger)$ maximises
	$\tilde G(\vec x,\vec y,\vec z)$,
	with the conditions that $x_i, y_i\geq 0$,
	$\sum_{i=1}^r x_i = \rho = 1-\sum_{i=1}^r y_i$, and that there
	exist Hermitian matrices $X,Y,Z$ with respective spectra $x,y,z$ with
	$X-Y=Z$. Following equation (\ref{eq:FE-WB-matrices}), we can rewrite
	$\tilde G$ as the function of the matrices $X$ and $Y$ being maximised
	in \eqref{eq:FE-WB-matrices}.
	If the entries of $\vec x$ are ordered decreasingly, then as before the
	trace-inequality \eqref{eq:trace-ineq} implies that for $c>0$ the entries of 
	$\vec y$ should also be ordered decreasingly, while for $c<0$ they
	should be ordered increasingly. 
	This gives the form of $\vec z^\dagger$
	stated in \eqref{eq:TS-z}.
	
	It remains only to show that  
	\be
	\lim_{[\la,\mu]/n\to z}
	\frac{\chi_{U_{[\la,\mu]}}(e^{W/n})}{\dim(U_{[\la,\mu]})}
	=
	R(w_1,\dotsc,w_{r};z_1,\dotsc,z_{r}),
	\ee
	where $R$ is given by \eqref{eq:R}.
	This is proved almost identically to Lemma 6.1 from
	\cite{BFU20}.  Indeed, using  \eqref{eq:GL-char-2} we get
	\be
	\begin{split}
		\frac{\chi_{U_{[\la,\mu]}}(e^{W/n})}{\dim(U_{[\la,\mu]})}
		&=
		\det[e^{w_i[\la,\mu]_j/n + w_i(r-j)/n}] \cdot \\ 
		&\ \ \ \cdot \prod_{1\le i<j\le r} \frac{j-i}{(e^{w_i/n}-
			e^{w_j/n})([\la,\mu]_i-[\la,\mu]_j+j-i)},
	\end{split}
	\ee
	which, noting all the products (including in the determinant)
	are finite, tends to $R(w_1,\dotsc,w_{r};z_1,\dotsc,z_{r})$ as
	$[\la,\mu]/n \to z$. 
	\qed

	\subsection{Magnetisation term:  
		proof of Theorem \ref{thm:mag}}
	
	We start by giving expressions for the free energy with a
	magnetisation term, and then afterwards we will take the 
	appropriate derivatives. 
	We will need the following notation:
	\begin{itemize}[leftmargin=*]
		\item $\Delta^+$ will denote the set of vectors 
		$\vec  z=(z_1,z_2,\dotsc,z_r)$ that can arise as spectra of $X+Y$
		where $X$ and $Y$ are positive semidefinite Hermitian
		matrices with $\tr[X]=1-\tr[Y]=\rho$, ordered so that
		$z_1\geq\dotsb\geq z_r$.  In fact, $\Delta^+$ consists of all $\vec z$
		satisfying $z_1\geq\dotsb\geq z_r\geq 0$ and $\sum_{i=1}^r z_i=1$.
		Given $\vec z\in\Delta^+$, we write 
		$\cH^+_\rho(\vec z)$ for the set of pairs
		$(X,Y)$ of such matrices with $X+Y$ having spectrum $\vec z$.
		\item $\Delta^-_\rho$ will denote the set of vectors 
		$\vec  z=(z_1,z_2,\dotsc,z_r)$ that can arise as spectra of $X-Y$
		where $X$ and $Y$ are as above, again ordered so that
		$z_1\geq\dotsb\geq z_r$.  
		Now $\Delta^-_\rho$ consists of all $\vec z$
		satisfying $\rho \geq z_1\geq\dotsb\geq z_r\geq -(1-\rho)$ 
		and $\sum_{i=1}^r z_i=2\rho-1$.
		Given $\vec z\in\Delta^-_\rho$, we write 
		$\cH^-_\rho(\vec z)$ for the set of pairs
		$(X,Y)$ of such matrices with $X-Y$ having spectrum $\vec z$.
	\end{itemize}
	Let $\Phi^\#(\b,h)=\Phi^\#_{\b,h}(a,b,c,\vec{w})$ be as in 
	\eqref{eq:phi-mag} and recall from \eqref{eq:phi-newer} that 
	\[
	\phi(X,Y)= S(X)+S(Y)+\tfrac\b2 
	\tr\big[ a X^2+b Y^2+ 2cXY  \big].
	\]

	\begin{theorem}\label{thm:fe-mag}
		Let $a,b,c\in\RR$ and $w_1 \ge \cdots \ge w_r$ be fixed. 
		If $n,m\to\oo$ such that $m/n\to\rho\in(0,1)$, then the free energy
		of the models \eqref{eq:H-AB-mag} and \eqref{eq:H-WB-mag} satisfy:
		\be\label{eq:fe-max-mag}
		\begin{split}
			\Phi^\ab({\b,h})
			&=\max_{\vec z\in\D^+} \left( 
			\max_{(X,Y)\in\cH^+_\rho(\vec z)} \phi(X,Y)+
			\left\{\begin{array}{ll}
				h\sum_{i=1}^r z_iw_i, & \mbox{if } h>0, \\
				h\sum_{i=1}^r z_iw_{r+1-i}, & \mbox{if } h<0,
			\end{array}\right. \right)\\
			\Phi^\wb({\b,h})
			&=\max_{\vec z\in\D^-_\rho} \left( 
			\max_{(X,Y)\in\cH^-_\rho(\vec z)} \phi(X,Y)+
			\left\{\begin{array}{ll}
				h\sum_{i=1}^r z_iw_i, & \mbox{if } h>0, \\
				h\sum_{i=1}^r z_iw_{r+1-i}, & \mbox{if } h<0,
			\end{array}\right. \right).
		\end{split}
		\ee
	\end{theorem}
	\begin{proof}
		Let us start with the $\ab$ case.
		Using the expression \eqref{eq:eW-bi} and arguing similarly to
		\eqref{eq:Z-bi-approx} we have 
		\be\label{eq:Z-bi-mag}
		\begin{split}
			Z_{n,h}^{\ab}=&
			\sum_{\mu,\nu,\la} s_\la(e^{hw_1},\dotsc,e^{hw_r}) \\
			&\cdot c_{\mu,\nu}^\la  d_\mu d_\nu 
			\exp\Big(\tfrac\b n[(a-c)\ct(\mu)+(b-c)\ct(\nu)+c\cdot\ct(\la)]
			\Big)\\
			=&
			\sum_{(\mu/n,\nu/n,\la/n)\in\Om^+_{m/n}}
			s_\la(e^{hw_1},\dotsc,e^{hw_r})
			\exp\Big(n \Big\{
			\tilde F(\tfrac{\mu}{n},\tfrac{\nu}{n},\tfrac{\la}{n})+o(1)\Big\} \Big),
		\end{split}
		\ee
		where $\tilde F$ is given in \eqref{eq:phi} and
		$\Om^+_\rho$ in \eqref{eq:Om-plus}.
		Recall  that  \cite[Section 2.2]{fulton:young}
		\be\label{eq:schur-tableaux-mag}
		s_\la(e^{hw_1},\dots,e^{hw_r})=
		\sum_{\TT}\prod_{i=1}^r e^{h m_i w_i}=
		\sum_{\TT} e^{\sum_{i=1}^r h m_i w_i},
		\ee
		where the sum is over all semistandard Young tableaux $\TT$ with shape
		$\la$ and entries in $\{1, \dots, r\}$, and where for each $i$, $m_i$
		is the number of times the number $i$ appears in $\TT$. The tableau
		with each box in the $i^{\mathrm{th}}$ row labelled $i$ appears in the sum, and
		in fact, for $h>0$, it maximises the sum in the exponent: 
		\be
		e^{\sum_{i=1}^r h m_i w_i} \le 
		e^{\sum_{i=1}^r h \la_i w_i},
		\ee
		for each valid $\TT$. Indeed, note that in a semistandard
		tableau, the entries of row $i$ must be at least $i$. Then,
		taking any semistandard $\TT$, shape $\lambda$, changing an
		entry $j \ge i$ in row $i$ to $i$ changes the sum in the
		exponent by $h(w_i-w_j)$, which is non-negative by our
		ordering of $\vec{w}$ as $w_1 \ge \cdots \ge w_r$. 
		Hence for $h>0$,
		\be\label{eq:schur-bound-mag}
		e^{\sum_{i=1}^r h \la_i w_i}\le
		s_\la(e^{ hw_1},\dots,e^{ hw_r})\le 
		\dim(U_\la) e^{\sum_{i=1}^r  h \la_i w_i}.
		\ee
		Recalling that $\frac{1}{n}\log \dim(U_\la) \to 0$ we get, for $h>0$,
		\be
		Z_{n,h}^{\ab}=
		\sum_{(\mu/n,\nu/n,\la/n)\in\Om^+_{m/n}}
		\exp\Big(n \Big\{
		\tilde F(\tfrac{\mu}{n},\tfrac{\nu}{n},\tfrac{\la}{n})
		+h{\textstyle\sum_{i=1}^r} \tfrac{\la_i}{n} w_i
		+o(1)\Big\} \Big).
		\ee
		In the case $h<0$, the sum in the exponent in
		\eqref{eq:schur-tableaux-mag} is maximised when $m_i =
		\la_{r+1-i}$ for each $i$; indeed, let $h' = -h$, and $w'_i =
		-w_{r+1-i}$, and apply the same reasoning as above. So, for
		$h<0$, we have 
		\be\label{eq:schur-bound-mag-2}
		e^{\sum_{i=1}^r h \la_{r+1-i} w_i}\le
		s_\la(e^{ hw_1},\dots,e^{ hw_r})\le 
		\dim(U_\la) e^{\sum_{i=1}^r  h \la_{r+1-i} w_i},
		\ee
		and consequently
		\be
		Z_{n,h}^{\ab}=
		\sum_{(\mu/n,\nu/n,\la/n)\in\Om^+_{m/n}}
		\exp\Big(n \Big\{
		\tilde F(\tfrac{\mu}{n},\tfrac{\nu}{n},\tfrac{\la}{n})
		+h{\textstyle\sum_{i=1}^r} \tfrac{\la_i}{n} w_{r+1-i}
		+o(1)\Big\} \Big).
		\ee
		The result for the {\ab}-case then follows by
		arguing as in \eqref{eq:lim1} and \cite[Lemma~3.4]{Bjo16}.
		
		For the {\wb}-case, a very similar argument as for \eqref{eq:Z-bi-mag}
		gives
		\be
		Z^\wb_n(\b,h) =\sum_{(\pi/n,\tau/n,[\la,\mu]/n)\in\Om^-_{m/n}}
		\chi_{U_{[\la,\mu]}}(e^{hw_1},\dotsc,e^{hw_r})
		\exp\Big(n \Big\{
		\tilde G(\tfrac{\mu}{n},\tfrac{\nu}{n},\tfrac{\la}{n})+o(1)\Big\} \Big),
		\ee
		where $\tilde G$ is given in \eqref{eq:tilde-G}, $\Om^-_\rho$ is
		defined just above \eqref{eq:tilde-G}, and 
		$\chi_{U_{[\la,\mu]}}$ is given in  \eqref{eq:GL-char-2}.
		In particular,   from \eqref{eq:GL-char-2}, we see that upper and lower
		bounds from \eqref{eq:schur-bound-mag} and
		\eqref{eq:schur-bound-mag-2} extend to this
		case.
		The result for the {\wb}-case then follows by
		arguing as in \eqref{eq:FE-WB-matrices} and \cite[Lemma~3.4]{Bjo16}
		again.
	\end{proof}

	\begin{proof}[Proof of Theorem \ref{thm:mag}]
		The proof closely follows that of Theorem 4.1 from
		\cite{Bjo16}. We start from the expressions \eqref{eq:fe-max-mag}
		where, for ease of notation, we
		drop the superscript.  We give details only 
		in the \ab-case with $h>0$ as the
		other cases are very similar. 
		
		Let 
		$F_{\max}=\Phi(\beta,0)=\max_{\vec z\in\D^+}\big(
		\max_{(X,Y)\in\cH^+_\rho(\vec z)} \phi(X,Y)\big)$
		and let 
		\be
		K=\Big\{\vec z \in\D^+: \max_{(X,Y)\in\cH^+_\rho(\vec z)} 
		\phi(X,Y)=F_{\max}\Big\}
		\ee
		denote the set of maximisers.
Note that $K$ is compact.
		Clearly,
		\be\label{eq:left-deriv-mag}
		\begin{split}
			\frac{\Phi(\b,h) - \Phi(\b,0)}{h}&=
			\max_{\vec z\in\D^+} \Big[
			\sum_{i=1}^rz_iw_i+
			\frac{ \max_{(X,Y)\in\cH^+_\rho(\vec z)} \phi(X,Y)
				- F_{\max}}{h} \Big]\\
			&\geq\max_{\vec z\in K}  \sum_{i=1}^rz_iw_i.
		\end{split}
		\ee
		We want to prove that the left-hand side of \eqref{eq:left-deriv-mag}
		tends to the right-hand side as $h\rightarrow 0$. For a contradiction,
		assume that there is a sequence $h_n\rightarrow 0$ such that the
		corresponding limit exists and is strictly larger than the right-hand
		side. For each  $h_n$, pick an element $\vec{z}(h_n)\in \D^+$ 
		that achieves the first maximum in
		\eqref{eq:left-deriv-mag}. 
		Since $\D^+$ is compact,  we can assume after passing to a subsequence
		if necessary that $\vec{z}(h_n)\rightarrow \vec z^\star$ 
		as  $h_n\rightarrow 0$. 
		We claim that $\vec z^\star\in K$. Otherwise, 
		$\max_{(X,Y)\in\cH^+_\rho(\vec z^\star)} \phi(X,Y)<F_{\max} $, 
		which would mean that the left-hand side of \eqref{eq:left-deriv-mag}
		tends to $-\infty$ as $h=h_n\rightarrow 0$, contradicting the lower
		bound on the right. It follows that 
		\be\begin{split}
			\frac{\Phi(\b,h_n) - \Phi(\b,0)}{h_n}&=
			\sum_{i=1}^rz_i(h_n)w_i+
			\frac{\max_{(X,Y)\in\cH^+_\rho(\vec z(h_n))} \phi(X,Y) 
				- F_{\max}}{h_n}\\
			&\leq\sum_{i=1}^rz_i(h_n)w_i\rightarrow \sum_{i=1}^r
			z_i^\star w_i\leq\max_{\vec z\in K}\sum_{i=1}^r
			z_i^\star w_i,
		\end{split}
		\ee
		as required.  
		
		In the \wb-case, we follow the same reasoning but with 
		$\D^+$ replaced by $\D^-_\rho$, with $\cH^+_\rho$ replaced by
		$\cH^-_\rho$, and the maxima in \eqref{eq:left-deriv-mag} replaced by
		minima (as well as $w_i\leftrightarrow w_{r+1-i}$).
		
		It remains to show that the $z_i$ may be expressed as in the statement
		of the Theorem.  Indeed, we know from \eqref{eq:trace-ineq} that
		$\phi(X,Y)$ is maximised when $X$ and $Y$ are simultaneously diagonal,
		with entries $x_1,\dotsc,x_r$ and $y_1,\dotsc,y_r$, respectively,
		ordered as follows:
		\begin{itemize}[leftmargin=*]
			\item if $c>0$, if $x_1\geq\dotsb\geq x_r\geq 0$
			then $y_1\geq\dotsb\geq y_r\geq 0$;
			\item if $c<0$, if $x_1\geq\dotsb\geq x_r\geq 0$
			then $0\leq y_1\leq\dotsb\leq y_r$.
		\end{itemize}
		This gives the result.
	\end{proof}

	\section{The phase-transition}
	\label{sec:max}
	
	In this section we prove Propositions \ref{prop:crittemp},
	\ref{prop:r2crittemp}, \ref{prop:tcrittemp} and \ref{prop:unique}.
	Let us start by recalling the basic quantities of interest:
	we wish to maximise the function 
	\be
	F(\om)=F(\vec x;\vec y)= \textstyle\sum_{i=1}^r f(x_i,y_i),
	\ee
	over the domain
	\be
	\Om=\big\{\om=(\vec x; \vec y):
	x_1,\dotsc,x_r,y_1,\dotsc,y_r\geq 0,\;
	\textstyle\sum_{i=1}^rx_i=1-\sum_{i=1}^ry_i=\rho
	\big\}.
	\ee
	Here
	\be
	f(x,y)=-x\log x-y\log y+
	\tfrac\b2\big(a x^2+by^2+2cxy\big),
	\ee
	and we write $Q(x,y)=\tfrac12\big(a x^2+by^2+2cxy\big)$
	for the quadratic form appearing in $f(x,y)$.
	We will write $\rho'=1-\rho$ to lighten the notation.
	
	We are particularly interested in whether the maximum
	of $F$  is attained at the point
	\be\label{eq:trivmax}
	\om_0=\big(\tfrac\rho r, \tfrac\rho r,\dotsc, \tfrac\rho r; 
	\tfrac{\rho'}r,\tfrac{\rho'}r,\dotsc,\tfrac{\rho'}r\big),
	\ee
	or at some other point in $\Om$. 
	
	\subsection{Existence of a phase transition: proof of Proposition
		\ref{prop:crittemp}}

	We are now ready to prove our result on the existence of a critical
	point. Recall that we want to prove that $\b_\crit$ exists (is
	positive and finite) if and only if $Q$ is not negative semidefinite,
	where $\b_\crit$ is the maximum of the $\b$ for which $\om_0$ is a
	maximiser of $F$. 
	We will need the following elementary identity.
	
\begin{lemma}\label{lem:quadform}
  If $Q$ is a quadratic form of two variables, then
  \be
  r\sum_{j=1}^r Q(x_j,y_j)=
  Q(x_1+\dots+x_r,y_1+\dots+y_r)+
\sum_{1\leq i<j\leq r} Q(x_i-x_j,y_i-y_j). 
  \ee
\end{lemma}

\begin{proof}
  When $Q(x,y)=xy$ we need to prove that
  \be
  r\sum_{j=1}^r x_jy_j =(x_1+\dots+x_r)(y_1+\dots+y_r)+
  \sum_{1\leq i<j\leq r}(x_i-x_j)(y_i-y_j).
  \ee
  This is easy  to see by comparing the coefficient of each monomial on
  the two sides. Specializing $x_j=y_j$ proves the result for
  $Q(x,y)=x^2$ and $Q(x,y)=y^2$, and the general case then follows by
  linearity. 
\end{proof}

\begin{proof}[Proof of Proposition \ref{prop:crittemp}]
  We will write 
  \be\label{eq:F-F0=bE+H}
  F(\om)-F(\omega_0)=
  \beta\cE(\om) +\cH(\om),
  \ee
  where 
  \be
  \cE(\vec x;\vec y)=
  \sum_{j=1}^rQ(x_j,y_j)
  -rQ\big(\tfrac{\rho}r,\tfrac{\rho'}r\big),
  \ee
and
  \begin{equation}\label{eq:Phi}
    \cH(\vec x;\vec y)=
    \sum_{j=1}^r(-x_j\log x_j-y_j\log y_j)
    +\rho\log\tfrac\rho r+\rho'\log\tfrac{\rho'}{r}.
  \end{equation}
  The term $\cE$ is in some sense an energy term, and $\cH$ an entropy term.
  Note that $F$ is maximised at $\omega_0$ if and only if 
  $\beta\cE(\om) +\cH(\om)\leq 0$
  on $\Omega$. 

  On $\Omega$, we can write
  \begin{multline*}
    \tfrac1r{\cH(\vec x;\vec y)}=
    -h\left(\frac{x_1+\dots+ x_r}r\right)+\frac{h(x_1)+\dots+h(x_r)}{r}\\
    - h\left(\frac{y_1+\dots+ y_r}r\right)+\frac{h(y_1)+\dots+h(y_r)}{r},
  \end{multline*}
  where $h(x)=-x\log x$. Since $h$ is strictly concave, 
  $\cH(\om)\leq 0$ with equality only at the point 
  $\omega_0$.  Moreover, by Lemma \ref{lem:quadform},
  \begin{equation}\label{eq:psiid}
    \cE(\vec x;\vec y)=\frac 1 r\sum_{1\leq i<j\leq r}Q(x_i-x_j,y_i-y_j).
  \end{equation}
  Thus, if $Q$ is negative semidefinite, we have $\cE(\om)\leq 0$ and
  consequently $\omega_0$ is the unique maximum point of $F$.
		
  Assume now that $Q$ is not negative semidefinite. 
  We claim that $\cE$ assumes strictly positive values in $\Om$.  To
  see this, it suffices to 
  consider the case when
  $x_2=\dots=x_r$, $y_2=\dots=y_r$. Then
  \begin{equation}\label{eq:psim2}
    \cE(\vec x;\vec y)=\frac{r-1}{r} \,Q(\xi,\eta), 
  \end{equation}
  where $\xi=x_1-x_2$ and $\eta=y_1-y_2$. 
  Here $(\xi,\eta)$ can take any value in 
  $\big[-\frac\rho{r-1},\rho\big]\times\big [-\frac{\rho'}{r-1},\rho'\big]$.
  By assumption, $Q$ assumes  positive values  in parts of this
  rectangle. Then it is clear that 
  $\cE$ takes positive values, hence that
  $\cH(\om)+\beta\cE(\om)$ assumes  positive
  values for $\beta$ large enough, and that the set of $\beta>0$ for
  which this is true is an 
  interval $\beta>\beta_\crit$. 
  To see that $\om_0$ is the unique maximiser for $\b<\b_\crit$, take
  $\om\in\Om\setminus\{\om_0\}$.  Then either $\cE(\om)>0$, in which
  case $\cH(\om)+\b\cE(\om)<\cH(\om)+\b_\crit\cE(\om)\leq0
  =\cH(\om_0)+\b\cE(\om_0)$, or $\cE(\om)\leq0$, in which case
  $\cH(\om)+\b\cE(\om)\leq \cH(\om)<0=\cH(\om_0)+\b\cE(\om_0)$.
  
  It remains to show that $\beta_\crit\neq 0$, that is, that $F$ assumes its
  maximum value at $\omega_0$ for $\beta$ close to zero. 
  We will show that this is in fact true if we maximise $F$ over the larger set
  \be
  U=\big\{	(\vec x;\vec y):
  \,0\leq x_j\leq \rho,
  \,0\leq  y_j\leq \rho',\,j=1,\dots,r
  \big\}. 
  \ee
  To do this we will show that
  the  Hessian $H(F)$ is negative definite in $U$ for $\b$ close to 0,
  meaning that $F$ is concave in $U$ for such $\b$ and
  that  $\om_0$ is a  global maximum in $U$.  
  The  Hessian $H(F)$ is a direct sum of the Hessians 
  \be\label{eq:hess} 
  H(f)=\left(\begin{matrix}f_{xx}& f_{xy}\\ 
      f_{xy} & f_{yy}\end{matrix}\right) 
  =\left(\begin{matrix}\beta a-\frac 1x& \beta c\\ 
      \beta c &\beta b-\frac 1y\end{matrix}\right),
  \ee 
  which is negative definite if and only if
  \be \label{eq:negdef} 
  \big(\beta a-\tfrac 1 x\big)\big(\beta b-\tfrac 1 y\big)>
  \beta^2c^2, 
  \qquad \tfrac 1 x>\beta a,
  \qquad \tfrac 1 y>\beta b.
  \ee 
  By monotonicity, when $x\leq \rho$ and $y\leq\rho'$ 
  the inequalities \eqref{eq:negdef} are implied by
  \be \label{eq:negdef2} 
  \big(\beta a-\tfrac 1 \rho\big)\big(\beta b-\tfrac 1 {\rho'}\big)>
  \beta^2c^2, 
  \qquad \tfrac 1 \rho>\beta a,
  \qquad \tfrac 1 {\rho'}>\beta b.
  \ee 
  But \eqref{eq:negdef2} holds for $\b=0$, hence by continuity also for
  small  positive $\b$, as required.
\end{proof}

From the proof above we note that $\beta\leq\beta_\crit$
if and only if $\cH(\om)+\b\cE(\om)\leq 0$ for all $\om\in\Om$, and
	also that we have the  expression
	\begin{equation}\label{eq:crittempformula}
	\beta_\crit=\inf_{\om \in
		\Omega^+}\Big(-\frac{\cH(\om)}{\cE(\om)}\Big), 
	\quad\mbox{where }
	\Omega^+=\big\{	\om \in\Omega:\,\cE(\om)>0\big\}. 
	\end{equation}
	
	\subsection{Formulas for $\b_\crit$:  proofs of Propositions
		\ref{prop:r2crittemp} and 
		\ref{prop:tcrittemp}}
	
	We now turn to the proofs of our formulas for $\beta_\crit$,
	Proposition \ref{prop:r2crittemp} for the case $r=2$ and Proposition
	\ref{prop:tcrittemp} for the case $r\ge3$, $c\ge0$ and $(a-c)\rho=(b-c)\rho'=:t$. 
	
	Our strategy is to obtain general lower and upper bounds on $\beta_\crit(r)$, given in Propositions \ref{Prop:crittemplow} and \ref{prop:crittemphigh} respectively, 
	which are tight in the two cases that we consider. Both bounds are given in terms of the critical temperature $\beta_\crit^{\text{h}}(r)$ of the homogeneous case $a=b=c=1$
	(the superscript h is for ``homogeneous'').
	In  \cite[Theorem 4.2]{Bjo16}, it was found that
	\begin{equation}\label{eq:crittemphom}
	\beta_\crit^{\text{h}}(r)=
	\begin{cases}2, & r=2,\\
	\displaystyle\frac{2(r-1)\log(r-1)}{r-2}, & r\geq 3.
	\end{cases} \end{equation}
	Note that this agrees with our Proposition
	\ref{prop:tcrittemp};  the corresponding form
	$Q(x,y)=\tfrac12(x+y)^2$ 
	is not negative semidefinite and  \eqref{eq:tcond} holds with $t=0$.

	To get a  better understanding of 
	Proposition \ref{prop:tcrittemp}, we note that
	\eqref{eq:tcond}  implies the explicit diagonalization
	\begin{equation}\label{eq:qtdiag}
	Q( x, y)=\frac{t\rho\rho'}{2}\left(\frac x{\rho}-\frac
	y{\rho'}\right)^2
	+\frac{c+t}{2}(x+y)^2. 
	\end{equation}
	That $Q$ is not negative semidefinite means that at least one of $t$
	and $c+t$ are positive. Since we assume that $c\geq 0$  
	this means that $c+t>0$. In particular,  the expression for
	$\beta_\crit(r)$ in  
	Proposition \ref{prop:tcrittemp} is always positive.

	Let us now obtain a lower bound for $\b_\crit$. We deduce from
	\eqref{eq:crittempformula} and 
	\cite[Theorem 4.2]{Bjo16} with
	$\rho=1$ that  
	$-\cH(\vec x;\vec 0)\geq \beta_\crit^{\text{h}}(r)\cE(\vec x;\vec 0)$.
	This inequality takes the form
	\be
	\sum_{j=1}^rx_j\log x_j-\log\tfrac 1r
	\geq \tfrac{\beta_\crit^{\text{h}}(r)}{2r}
	\sum_{1\leq i<j\leq r}(x_j-x_i)^2,
	\quad
	\mbox{where } \textstyle\sum_{j=1}^r x_j=1.
	\ee
	Replacing each $x_j$ by $x_j/\rho$ gives
	\begin{equation}\label{eq:entropyestimate}
	\sum_{j=1}^rx_j\log x_j- \rho\log\tfrac{\rho}{r}
	\geq \tfrac{\beta_\crit^{\text{h}}(r)}{2\rho r}\sum_{1\leq i<j\leq r}(x_j-x_i)^2,
	\quad
	\mbox{where } \textstyle\sum_{j=1}^r x_j=\rho.
	\end{equation}
	As was observed in \cite{Bjo16}, equality
	in \eqref{eq:entropyestimate} holds both at the point 
	$x_1=\dots=x_r=\rho/r$
	and at \eqref{eq:specialx}. 
	(They are the same point if $r=2$.)
	
	We will temporarily write $\gamma$ for the explicit expression \eqref{eq:r2crittemp} (we aim to show that $\beta_\crit(2)=\gamma$).
	We will need the following description of $\gamma$. 
	
	\begin{lemma}\label{lem:perturbation}
		Assume that $Q(x,y)=\tfrac12(ax^2+by^2+2cxy)$ is not negative 
		semidefinite and that $\beta,\,\rho,\,\rho'> 0$. Then, the form
		\begin{equation}\label{eq:perturbq}
		\beta Q(x,y)-\frac{x^2}{\rho}-\frac{y^2}{\rho'}
		\end{equation}
		is negative semidefinite if and only if $\beta\leq \gamma$, and negative definite if and only if $\beta<\gamma$.
	\end{lemma}
	
	\begin{proof}
		By assumption, the first term in \eqref{eq:perturbq} can assume
		positive values, and the second term is always non-positive. It
		follows that the range of  $\beta$ 
		for which \eqref{eq:perturbq} is negative semidefinite 
		is of the form $\beta\leq \beta_0$ and that it is negative definite if
		and only if $\b<\b_0$.  The precise conditions for 
		\eqref{eq:perturbq} to be negative semidefinite are
		\be \left(\beta a-\frac 1{2\rho}\right)\left(\beta b-\frac 1{2\rho'}\right)\geq \beta^2 c^2,\qquad \beta a\leq \frac 1{2\rho},\qquad
		\beta b\leq \frac 1{2\rho'}.
		\ee
		By continuity, 
		$$\left(\beta_0 a-\frac 1{2\rho}\right)\left(\beta_0 b-\frac 1{2\rho'}\right)= \beta_0^2 c^2. $$
		If
		$ab=c^2$, this is a linear equation with the solution 
		$\beta_0=2/(a\rho+b\rho')=\g$. Otherwise,
		it has two solutions  
		\be
		\beta_{\pm}=\frac{\rho a+(1-\rho) b\pm\sqrt{(\rho a-(1-\rho) b)^2+4\rho(1-\rho)c^2}}{\rho(1-\rho)(ab-c^2)} ,
		\ee
		which satisfy $(ab-c^2)\beta_+\beta_-=1/4\rho\rho'>0$.  
		If $ab>c^2$, both solutions are positive and $\beta_0$ equals the
		smallest solution $\beta_-=\g$. If $ab<c^2$ the solutions have opposite
		sign.  In this case $\beta_0$ is the  largest solution, which is again
		$\beta_-=\g$. 
	\end{proof}

	\begin{proposition}\label{Prop:crittemplow}
		Assume that $Q$  is not negative semidefinite, so that $\beta_\crit$ exists.
		Then,
		\begin{equation}\label{eq:crittemplow}
		\beta_\crit\geq  \tfrac12 \beta_\crit^{\text{h}}(r)\gamma.
		\end{equation}
	\end{proposition}
	
	\begin{proof}
		Using the estimate \eqref{eq:entropyestimate} in \eqref{eq:Phi} gives
		\be\label{eq:hej3}
		-\cH(\om)\geq \frac{\beta_\crit^{\text{h}}(r)}{2r}
		\sum_{1\leq i<j\leq
			r}\left(\frac{(x_i-x_j)^2}{\rho}+\frac{(y_i-y_j)^2}{\rho'}\right). 
		\ee
		It follows that
		\be\label{eq:hej4}
		\cH(\om)+\beta\cE(\om)
		\leq \frac {1} r\sum_{1\leq i<j\leq r}\tilde Q(x_j-x_i,y_j-y_i), 
		\ee
		where
		\be
		\tilde Q(x,y)=\beta Q(x,y)-\tfrac{\beta_\crit^{\text{h}}(r)}{2}
		\big(\tfrac{x^2}\rho+\tfrac{y^2}{\rho'}\big).
		\ee
		By Lemma \ref{lem:perturbation}, $\tilde Q$ is negative semidefinite if and only if
		$\beta\leq \tfrac12\beta_\crit^{\text{h}}(r)\gamma$. For $\beta$ in this range it follows that
		$\cH(\om)+\beta\cE(\om)\leq 0$ on $\Omega$.  
		This gives the desired  bound on $\beta_\crit$.
	\end{proof}
	
	Let us now move to upper bounds for $\beta_\crit$. We need to find a value of $\beta$ 
	such that	$F(\om)>F(\omega_0)$ for some points 
	$\om\in\Omega$.  We want to find upper bounds that in some case equal
	the lower bound in Proposition \ref{Prop:crittemplow}.  We can only
	expect this to work if we used the inequality  
	\eqref{eq:entropyestimate}  in cases when  it
	holds with equality. By the results of \cite{Bjo16} mentioned above, 
	it is natural to take
	$\om$ either close to $\omega_0$, or to $\om_1$ as in \eqref{eq:specialxy}. 
	This leads to the following two upper bounds.
	
	\begin{proposition}\label{prop:crittemphigh}
		Assume that $Q$  is not negative semidefinite, so that 
		$\beta_\crit$ exists.
		Then,
		\begin{equation}\label{eq:crittemphigh1} 
		\beta_\crit\leq \tfrac12 r\gamma. 
		\end{equation}
		If, in addition, $Q(\rho,\rho')>0$ and $r\geq 3$, then
		\begin{equation}\label{eq:crittemphigh2}	
		\beta_\crit\leq \frac{\beta^{\text{h}}_\crit(r)}
		{2Q(\rho,\rho')}. 
		\end{equation}
	\end{proposition}
	
	In fact,  \eqref{eq:crittemphigh2} holds also when $r=2$, but in that
	case it is weaker than \eqref{eq:crittemphigh1}.

	\begin{proof}
		We first consider the behaviour of $F$ near $\omega_0$. 
		More precisely, consider the points
		\be
		\omega_{t,u}=\omega_0+(t,-t,0,\dots,0;u,-u,0,\dots,0), 
		\ee
		which belong to $\Omega$ for $t,u$ close to $0$.
		We have the Taylor expansion
		\begin{align*} 
		F(\omega_{t,u})-F(\omega_0)&=f\big(\tfrac{\rho}{r}+t,\tfrac{\rho'}{r}+u\big)
		+f\big(\tfrac{\rho}{r}-t,\tfrac{\rho'}{r}-u\big) 
		-2f\big(\tfrac{\rho}{r},\tfrac{\rho'}{r}\big)\\
		&=
		\big(t^2f_{xx}+u^2f_{yy}+2tuf_{xy}\big)
		\big(\tfrac{\rho}{r},\tfrac{\rho'}{r}\big)
		+\mathcal O((t^2+u^2)^{3/2}).
		\end{align*}
		By \eqref{eq:hess}, the quadratic term is 
		\be
		2\beta Q(t,u)-r\big(\tfrac{t^2}{\rho}+\tfrac{u^2}{\rho'}\big). 
		\ee
		By Lemma \ref{lem:perturbation},  if $\beta>r\gamma/{2}$, 
		this form is not negative semidefinite. 
		It follows that $\omega_0$ is not a local maximum of $F$. 
		This gives the first result.
		
		Next, we consider the point $\om_1$ from
		\eqref{eq:specialxy} and assume $r\geq 3$.
		By a straightforward computation,
		\be\label{eq:hej1}
		\cH(\om_1)=-\tfrac{r-2}r\log(r-1)
		\ee
		and, by \eqref{eq:psim2},
		\be\label{eq:hej2}
		\cE(\om_1)=\tfrac{r-1}{r}\,
		Q\big(\tfrac{\rho(r-2)}{r-1},\tfrac{\rho'(r-2)}{r-1}\big)
		=\tfrac{(r-2)^2}{r(r-1)}\,Q(\rho,\rho').
		\ee
		The second upper bound now follows from \eqref{eq:crittempformula}.
	\end{proof}
	
	We can now put our upper and lower bounds together to prove Propositions \ref{prop:r2crittemp} and \ref{prop:tcrittemp}.
	
	\begin{proof}[Proof of Proposition \ref{prop:r2crittemp}]
		When $r=2$, \eqref{eq:crittemplow} and \eqref{eq:crittemphigh1} reduce to
		$\gamma\leq\beta_\crit\leq\gamma$, that is,  $\b_\crit(2)=\gamma$.
		For the statement about uniqueness of the maximiser, note that if 
		$\b=\b_\crit(2)$
		and $\om=(\vec x;\vec y)$ is a maximiser, then
		the left-hand-side of \eqref{eq:hej4}  equals  zero.  
		Then also the right-hand-side of \eqref{eq:hej4}  equals  zero, since
		$\tilde Q\leq 0$ for $\b\leq \tfrac12\b_\crit^\text{h}(2)\gamma=\b_\crit(2)$
		by the proof of Proposition \ref{Prop:crittemplow}.
		Hence
		\eqref{eq:hej3} holds with equality and therefore
		\eqref{eq:entropyestimate} holds with equality, as does the
		corresponding statement for $\vec y$.  But it follows from the proof
		of Theorem 4.2 in \cite{Bjo16} that (for $r=2$) equality in
		\eqref{eq:entropyestimate} holds only at the point $\om_0$.
	\end{proof}
	
	\begin{proof}[Proof of Proposition \ref{prop:tcrittemp}]
		Note that	 the lower bound in  \eqref{eq:crittemplow}
		and the upper bound in  \eqref{eq:crittemphigh2} are equal
		if $\gamma=Q(\rho,\rho')^{-1}$. 
		Assuming  \eqref{eq:tcond}, we can parametrise
		\be
		a=c+\tfrac t{\rho}, \qquad b=c+\tfrac t{\rho'}. 
		\ee
		It is then straight-forward to check that
		\be
		(\rho a-\rho' b)^2+4\rho\rho'c^2=c^2, \quad\mbox{and}\quad
		ab-c^2=\frac{t(c+t)}{\rho\rho'},
		\ee
		which gives
		\be
		\gamma=\frac{2t+c-\sqrt{c^2}}{t(c+t)}
		=\frac{2}{c+t}, \qquad   c\geq 0.
		\ee
		By \eqref{eq:qtdiag},
		\begin{equation}\label{eq:qrho}Q(\rho,\rho')=\frac{c+t}{2}.\end{equation}
		This shows that, under the conditions of Proposition \ref{prop:tcrittemp},
		the upper and lower bound for $\beta_\crit$ agree and hence $\beta_\crit=\beta_\crit^{\text{h}}(r)/(c+t)$.

		To see that the point $\om_1$ in \eqref{eq:specialxy} gives another
		maximiser at $\b=\b_\crit$, take 
		$\beta=\beta_\crit(r)=\b_\crit^\text{h}(r)/2Q(\rho,\rho')$ 
		to see from \eqref{eq:hej1} and \eqref{eq:hej2}
		that $\cH(\om_1)+\b\cE(\om_1)=0$ which is also the 
		maximum value of 
		$\cH(\om)+\b\cE(\om)$.  To see that $\om_1$ is the only other
		maximiser we argue as at the end of the proof of Proposition
		\ref{prop:r2crittemp}.  Namely, for 
		$\b=\b_\crit(r)=\tfrac12\b_\crit^\mathrm{h}(r)\gamma$,
		we have that \eqref{eq:entropyestimate} holds with equality, as does
		the corresponding statement for $\vec y$.  From \cite{Bjo16}, equality
		in  \eqref{eq:entropyestimate} holds only at the points $\om_0$ and
		$\om_1$ (assuming \eqref{eq:xgeq}).
	\end{proof}
	
	We can now complete the final proof of this section, that of Proposition~\ref{prop:unique}, that the maximiser is unique for $\b>\b_\crit$ close to
	$\b_\crit$ under the conditions in Proposition \ref{prop:tcrittemp}.

	\begin{proof}[Proof of Proposition \ref{prop:unique}]
		We first show that $F$ is strictly concave in neighbourhoods
		of  $\omega_0$ and $\omega_1$ in $\Om$.  More generally,
		consider $F(\vec x+\vec t;\vec y+\vec u)$, where 
		$(\vec x;\vec y)\in\Omega$ is a point with $x_2=\dots=x_r$ and $y_2=\dots=y_r$ and $(\vec t;\vec u)$ a small perturbation with
		\begin{equation}\label{eq:tusum}\sum_{j=1}^r t_j=\sum_{j=1}^r u_j=0.\end{equation}
		By \eqref{eq:hess}, the quadratic term in the Taylor expansion of $F$ is
		\begin{equation}\label{eq:quadterm}Q_1(t_1,u_1)+\sum_{j=2}^rQ_2(t_j,u_j), \end{equation}
		where
		\[ 
		Q_k(t,u)=\beta Q(t,u)-\frac {t^2}{2x_k}-\frac {u^2}{2y_k}.
		\]
		
		At the point $\omega_0$, we have
		\[
		Q_1(t,u)=Q_2(t,u)=\beta Q(t,u)-
		\Big(\frac {r t^2}{2\rho}+\frac{ru^2}{2\rho'}\Big). 
		\]
		It follows from Lemma \ref{lem:perturbation} that this is negative definite if $\beta<\beta_0=r\gamma/2$. 
		By continuity, it follows that $F$ is strictly concave  near $\omega_0$.
		Since $\omega_0$ is a stationary point it must then be a local maximum,  that is,
		$F(\vec x;\vec y)\leq F(\omega_0)$ for $(\vec x;\vec y)$ 
		near $\omega_0$ and $\beta<\beta_0$. 
		Using that
		\[
		\beta_\crit=\frac{(r-1)\log(r-1)}{r-2}\,\gamma,
		\]
		it is easy to check that $\beta_\crit<\beta_0=r\gamma/2$, so this applies in particular to $\beta$ near $\beta_\crit$.    
		
		The point $\omega_1$ cannot be handled as easily since $Q_1$
		is then not negative definite. Instead, we   use Lemma
		\ref{lem:quadform} and \eqref{eq:tusum} to write
		\[ 
		(r-1)\sum_{j=2}^rQ_2(t_j,u_j)=Q_2(t_1,u_1)+\sum_{2\leq i<j\leq
			r}Q_2(t_i-t_j,u_i-u_j).
		\]
		It follows that \eqref{eq:quadterm} equals
		\[ 
		Q_1(t_1,u_1)+\frac 1{r-1}\,Q_2(t_1,u_1)+\frac 1{r-1}\sum_{2\leq
			i<j\leq r}Q_2(t_i-t_j,u_i-u_j).
		\]
		We compute
		\[
		Q_1(t,u)+\frac 1{r-1}\,Q_2(t,u)=\frac r{r-1}\left(\beta Q(t,u)
		-\Big(\frac{r t^2}{2\rho}+\frac {r u^2}{2\rho'}\Big)\right). 
		\]
		As before,
		this is negative definite for $\beta<\beta_0$. Moreover,
		\[
		Q_2(t,u)=\beta Q(t,u)-\frac {r(r-1) t^2}{2\rho}-\frac {r(r-1)
			u^2}{2\rho'} 
		\]
		is negative definite for $\beta<(r-1)\beta_0$, which is a weaker condition. 
		We  conclude that $F$ is strictly concave for $\beta<\beta_0$ and
		$(\vec x;\vec y)$ near $\omega_1$. We note that from \eqref{eq:F-F0=bE+H},
		\[
		F(\omega_1)-F(\omega_0)=\cH(\omega_1)+\beta_\crit\cE(\omega_1)+(\beta-\beta_\crit)\cE(\omega_1),
		\]
		where the sum of the first two terms on the right hand side vanish and the last term is computed by
		\eqref{eq:hej2} and \eqref{eq:qrho}. This gives
		\[
		F(\omega_1)-F(\omega_0)=(\beta-\beta_\crit)\frac{(r-2)^2(c+t)}{2r(r-1)}, 
		\]
		which is clearly positive for $\beta>\beta_\crit$.
		
		For each $\beta>\beta_\crit$, let $\omega(\beta)$ be a maximiser of 
		$F$ in $\Omega$. Permute the coordinates so that \eqref{eq:xgeq} holds. 	  
		We claim that then $\omega(\beta)\rightarrow \omega_1$ as
		$\beta\searrow \beta_\crit$. Otherwise, there exists a sequence
		$\omega({\beta_n})$, $\beta_n\searrow \beta_\crit$, that avoids a
		neighbourhood of $\omega_1$. Since $\Omega$ is compact we may assume
		that this sequence converges. It must then converge to a maximiser of
		$F$ for $\beta=\beta_\crit$ that satisfies \eqref{eq:xgeq}. There are only
		two such points, $\omega_0$ and $\omega_1$, by Proposition \ref{prop:tcrittemp}. 
		However, we have seen that for
		$\beta_\crit<\beta<\beta_0$   
		we have  $F(\vec x;\vec y)\leq F(\omega_0)$  for
		$(\vec x;\vec y)$ near $\omega_0$ whereas $F(\omega_1)>F(\omega_0)$. Thus, a
		sequence of global maximisers cannot converge to $\omega_0$. This is a
		contradiction, and we conclude that $\omega(\beta)\rightarrow
		\omega_1$. These points must then enter a region where $F$ is strictly
		concave and hence maximisers are unique. This completes the proof. 
	\end{proof}

	\subsection{Form of the maximiser of $F$ for $c>0$}
	
In this section we will prove that, for $c>0$, any maximiser of $F$
	\eqref{eq:F} is of the form \eqref{eq:maxformleq2}. 
	This is useful for 
	the heuristic discussion of Gibbs states in Section
	\ref{sec:heuristics} and for the results on ground state phase diagrams in Section
	\ref{sec:gspd}.

	We assume thoughout this section that $\vec x$ is ordered as in
	\eqref{eq:xgeq}, that is $x_1\geq x_2\geq\dotsb\geq x_r$.
	Recall from the discussion after \eqref{eq:xgeq}
	that, for $c>0$, $F$ is maximised when the orders of $\vec{x}$ and
	$\vec{y}$ match, that is when also
	$y_1\ge\cdots\ge y_r$.
	We will adapt the arguments in \cite{Bjo16} and in the appendix of
	\cite{BFU20} to show the following.  
	
	\begin{proposition}\label{prop:maxc>0}
		For $c>0$, any maximiser $(\vec{x}^\star;\vec{y}^\star)$ of $F$ 
		in the set $\Om$ \eqref{eq:Om} 
		is of the form 
		\be\label{eq:maxformleq2}
		\begin{split}
			x_1^\star&\geq x_{2}^\star=\dotsb=x_r^\star,\\
			y_1^\star&\geq y_{2}^\star=\dotsb=y_r^\star.
		\end{split}
		\ee
		Moreover for the special case $a=b=0$, $c>0$, $\rho=1/2$, and
		$\b\neq\b_\crit$ we have that the maximiser is unique, and
		$x_i^\star=y_i^\star$ for all $i=1,\dotsc,r$.  
	\end{proposition}
	
	The proof of this proposition is divided into several steps. We first
	prove that a maximum point $(\vec{x};\vec{y})$ only has positive
	coordinates, and that 
	$x_j=x_k$ if and only if $y_j=y_k$ (this holds also for $c<0$). 
	Then we prove that, when $c>0$, the entries
	$x_i$ (and therefore $y_i$) can take at most two distinct values.  
	This reduces the number of variables we need to consider, leading to
	\eqref{eq:maxformleq2} and the uniqueness statement via  direct
	calculations.

	\begin{lemma}\label{lb}
		For any  $a,b,c\in\RR$ with $c\neq 0$, if $(\vec{x};\vec{y})$ is a
		maximum point  of  $F$ in $\Om$, then
		\begin{enumerate}
			\item all $x_j$ and $y_j$ are strictly positive,
			\item $x_j=x_k$ if and only if $y_j=y_k$. 
		\end{enumerate}
	\end{lemma}
	\begin{proof}
		In this proof we write $e_j$ for the unit vector with a 1 in the $x_j$-coordinate and remaining entries equal to $0$.
		For the first part, suppose that $\om=(\vec x;\vec y)\in\Omega$ is a maximum point
		such that $x_j=0$ for some $j$, and that $j$ is the smallest index
		with this property. Then, $\om(t)=\om+t(e_j-e_{j-1})\in\Omega$ for small
		enough $t>0$ (recall that $x_{j-1}\geq x_j$ by \eqref{eq:xgeq}). 
		By a direct computation,
		$F(\om(t))-F(\om)=-t\log t+O(t)$ as $t\to0$.
		It follows that $F(\om(t))>F(\om)$ for small $t$, 
		which contradicts $\om$ being a maximum point. 
		The same argument works for the variables $y_j$.
		
		For the second part,
		suppose that $x_j=x_k$ and $y_j\neq y_k$. If necessary, redefine $j$ and $k$ so that
		$\{l:\,x_l=x_k\}=\{j,j+1,\dotsc,k\}$. We still have 
		$y_j\neq y_k$.  Then $\om(t):=(\vec{x};\vec{y})+t(e_j-e_k)\in\Omega$  
		for small enough $t>0$.   (Here we use the first part of the lemma  in
		the case $k=r$.) 
		We have that $\tfrac\partial{\partial t} F(\om(t))|_{t=0}=c(y_j-y_k)>0$.
		This contradicts $\om$ being a maximum point.
		The same argument proves the reverse implication.
	\end{proof}
	
	Lemma \ref{lb} shows that at a maximum point  
	there is a composition $r=k_1+\dots+k_m$ so that 
	\begin{subequations}\label{xys}
		\begin{align} (x^\star_1,\dots,x^\star_r)&=(\underbrace{\xi_1,\dots,\xi_1}_{k_1},\dots,\underbrace{\xi_m,\dots,\xi_m}_{k_m}), \\
		(y^\star_1,\dots,y^\star_r)&=(\underbrace{\eta_1,\dots,\eta_1}_{k_1},\dots,\underbrace{\eta_m,\dots,\eta_m}_{k_m}), 
		\end{align}
	\end{subequations}
	where $\xi_j\neq \xi_k$ and $\eta_j\neq \eta_k$ for $j\neq k$.
	This leads to the problem of maximizing
	\be\label{eq:barF} \bar{F}(\xi;\eta)=k_1f(\xi_1, \eta_1)+\dots+k_m f( \xi_m, \eta_m)\ee 
	over the set $ \Omega^{(m)}$ defined by
	\begin{subequations}\label{eq:Omegam}
		\be \xi_1> \xi_2>\dots> \xi_m> 0,\qquad k_1\xi_1+\dots+k_m\xi_m=\rho,\ee 
		\be  \eta_1> \eta_2>\dots> \eta_m> 0,\qquad k_1\eta_1+\dots+k_m\eta_m=1-\rho.\ee 
	\end{subequations}
	
	For $m\geq2$,
	the set $ \Omega^{(m)}$ is open, so we may find local extreme points by
	using Lagrange multipliers.  At any such point we have 
	\be 
	\nabla\bar{F}(\xi;\eta)=\lambda\nabla (k_1\xi_1+\dots+k_m\xi_m)
	+\mu \nabla (k_1\eta_1+\dots+k_m\eta_m),  
	\ee 
	for some $\la,\mu\in\RR$.   Equivalently
	\begin{equation}\label{lmt}
	\tfrac{\partial f}{\partial \xi}(\xi_i,\eta_i)=\lambda,\quad 
	\tfrac{\partial f}{\partial \eta}(\xi_i,\eta_i)=\mu,\qquad 
	1\leq i\leq m. 
	\end{equation}
	The system \eqref{lmt} can in turn be rewritten in the form
	\begin{equation}\label{lmf}
	\eta_i=\phi_\lambda(\xi_i),\qquad \xi_i=\psi_\mu(\eta_i),
	\qquad 1\leq i\leq m, 
	\end{equation}
	where
	\be 
	\phi_\lambda(x)=\frac{\lambda+1+\log(x)-ax}{c},
	\qquad \psi_\mu(y)=\frac{\mu+1+\log(y)-by}{c}. 
	\ee 
	If we let $P_{\lambda,\mu}$ denote the intersection of the graphs
	$y=\phi_\lambda (x)$ and $x=\psi_\mu(y)$, we can summarise these
	findings as follows:
	the maximum of 
	$F$ in $\Omega$ is attained either at the point 
	$\om_0$ \eqref{eq:om0}, or at a point of the form
	\eqref{xys}, where  $2\leq m\leq r$, $(\xi,\eta)\in \Omega^{(m)}$ and
	$(\xi_i,\eta_i)\in P_{\lambda,\mu}$ for $1\leq i\leq m$.
	Note that
	$\phi_\lambda''(x)=- 1/{cx^2}$, $\psi_\mu''(y)=-1/{cy^2}$,
	so for $c>0$ the graphs are  convex. We can now prove that for $c>0$,
	a maximiser of $F$ can have at most two distinct entries $x_i$ (and
	therefore the same for $y_i$).
	Henceforth we suppress the indices $\la,\mu$ from $\phi,\psi$.
	
	\begin{proposition}\label{prop:mleq2}
		If $c>0$ then the $m$ of \eqref{xys} satisfies $m\leq 2$.
	\end{proposition}

	\begin{proof} 
		Suppose first that $b<0$. Then, $\psi$ is increasing and concave, so
		$\psi^{-1}$ is increasing and convex. The graph of $\psi^{-1}$ can
		intersect the graph of the concave function $\phi$ in at most two
		points. If $a<0$ the same argument works with $\phi$ and $\psi$
		interchanged.  
		
		This leaves the case when $a>0$ and $b>0$. 
		In the region  
		\be
		\mathcal{R}=\{(x,y):0<x<1/a,\,0<y<1/b\},
		\ee
		$\phi$ is increasing and concave whereas the local inverse $\psi^{-1}$ is increasing and convex. Thus, there are at most two crossing points in $\mathcal{R}$. If there are zero or two crossing points in $\mathcal{R}$, then an elementary convexity argument shows that there are no crossing points outside $\mathcal{R}$. 
		
		In all the cases considered so far there are at most two crossing points, which implies $m\leq 2$. In the remaining case, when there is
		exactly one crossing point in $\mathcal{R}$, there can be several crossing points outside $\mathcal{R}$.
		They can be ordered as a sequence $(x_j,y_j)$ with $x_j$ decreasing and $y_j$ increasing. We are only interested in subsequences of crossing points with $x_j$ and $y_j$ decreasing.  The maximum length of such a subsequence is $2$, where we may pick the unique crossing point in $\mathcal{R}$ and an arbitrary crossing point outside $\mathcal{R}$. This proves that $m\leq 2$ also in this case.
	\end{proof}
	
	We are now ready to prove Proposition \ref{prop:maxc>0}.

	\begin{proof}[Proof of 
		Proposition \ref{prop:maxc>0}]
		We absorb $\b$ in $a,b,c$, effectively setting $\beta=1$.
		It will be convenient to use $\xi=x_1-x_r$ and $\eta=y_1-y_r$ as parameters.
		By Proposition \ref{prop:mleq2} (using $k$ in place of $m$)
		we can write $\vec x$ and $\vec y$ as
		\be
		\begin{split}
			x_1&=\dotsb=x_k=\frac{\rho+(r-k)\xi}{r},\qquad x_{k+1}=\dotsb=x_r=\frac{\rho-k\xi}{r},\\
			y_1&=\dotsb=y_k=\frac{\rho'+(r-k)\eta}{r},\qquad y_{k+1}=\dotsb=y_r=\frac{\rho'-k\eta}{r},
		\end{split}
		\ee
		where $\rho'=1-\rho$.
		The function \eqref{eq:F} can then be written
		$$F(\xi,\eta,k)=kf\left(\tfrac{\rho+(r-k)\xi}{r},\tfrac{\rho'+(r-k)\eta}{r}\right)+(r-k)f\left(\tfrac{\rho-k\xi}{r},\tfrac{\rho'-k\eta}{r}\right). $$
		We need to show that the maximum of $F$ over $\xi\in[\rho,k]$, $\eta\in[\rho',k]$ and $k\in\{0,1,\dots,r\}$ is achieved at $k=1$. 
		Note that $k=0$, which corresponds to the point $\omega_0$ \eqref{eq:om0}, 
		is included in that case as  $k=1$, $\xi=\eta=0$. The idea is now to consider  $k$ as continuous. We will show the stronger statement that the maximum of 
		$F$ on the domain
		\begin{equation}\label{eq:Fdomain}0\leq\xi\leq\frac \rho k,\qquad 0\leq\eta\leq\frac{\rho'}k,\qquad 1\leq k\leq r\end{equation}
		is achieved at $k=1$.

		We first show that $F$ does not have any stationary points in  the interior.
		By a straightforward computation,
		\begin{align*}
		\frac{\partial F}{\partial \xi}&=\tfrac{k(r-k)}{r}\left(a\xi+c\eta-\log\tfrac{\rho+(r-k)\xi}{\rho-k\xi}\right),\\
		\frac{\partial F}{\partial \eta}&=\tfrac{k(r-k)}{r}\left(c\xi+b\eta-\log\tfrac{\rho'+(r-k)\eta}{\rho'-k\eta}\right),\\
		\frac{\partial F}{\partial k}&=\xi+\eta+\tfrac{r-2k}{r}\,Q(\xi,\eta)\\
		&\quad		-\tfrac{\rho+(r-2k)\xi}{r}\log\tfrac{\rho+(r-k)\xi}{\rho-k\xi}
		-\tfrac{\rho'+(r-2k)\eta}{r}\log\tfrac{\rho'+(r-k)\eta}{\rho'-k\eta}.
		\end{align*}		
		By the first two equations, at any stationary point we have
		\begin{equation}\label{eq:logxieta} \log\tfrac{\rho+(r-k)\xi}{\rho-k\xi}=a\xi+c\eta,\qquad
		\log\tfrac{\rho'+(r-k)\eta}{\rho'-k\eta}=c\xi+b\eta.
		\end{equation}
		Inserting this in the third equation and using
		$$Q(\xi,\eta)=\frac12\left(\xi(a\xi+c\eta)+\eta(c\xi+b\eta)\right)$$
		gives
		$$\frac{\partial F}{\partial k}=\xi+\eta-\tfrac{2\rho+(r-2k)\xi}{2r}(a\xi+c\eta)
		-\tfrac{2\rho'+(r-2k)\eta}{2r}(c\xi+b\eta). $$
		We now observe that \eqref{eq:logxieta} implies
		$$\coth\frac{a\xi+c\eta}2=\frac{2\rho+(r-2k)\xi}{\xi r},\qquad 
		\coth\frac{c\xi+b\eta}2=\frac{2\rho'+(r-2k)\eta}{\eta r},
		$$
		which in turn gives
		\begin{equation}\label{dfdk}\frac{\partial F}{\partial k}=\xi\left(1-\tfrac{a\xi+c\eta}2\coth\tfrac{a\xi+c\eta}2\right) +\eta\left(1-\tfrac{c\xi+b\eta}2\coth\tfrac{c\xi+b\eta}2\right).\end{equation}
		Note that $1-(x/2)\coth(x/2)\leq 0$ for all $x$, with equality only if $x=0$.
		So a stationary point must satisfy $\xi(a\xi+c\eta)=\eta(c\xi+b\eta)=0$.
		However, if $a\xi+c\eta=0$ then \eqref{eq:logxieta} gives $\xi=0$ and similarly if $c\xi+b\eta=0$ then  $\eta=0$. Thus, $F$ has no stationary points in the
		interior of \eqref{eq:Fdomain}.
		
		It remains to study $F$ on the boundary of \eqref{eq:Fdomain}. 
		At the boundary component $\xi=0$, all $x$-variables are equal. By Lemma~\ref{lb},  at any such maximum point also the $y$-variables are equal, so it must be the point $\omega_0$. Similarly, any maximum point with $\eta=0$ is $\omega_0$. If $\xi=\rho/k$ then $x_r=0$, but we know from
		Lemma \ref{lb} that $F$ is not maximised at such a
		point. 	Similarly, we exclude the case $\eta=\rho'/k$. The case
		$k=r$ again corresponds to  $\om_0$.
		The only remaining boundary component is $k=1$. This shows that 
		any maximiser of $F$ has the form \eqref{eq:maxformleq2}.
		
		To finish the proof of  Proposition
		\ref{prop:maxc>0}, it remains to show that in the case $a=b=0$,
		$c>0$, $\rho=\frac12$, 
		and $\beta\neq\beta_\crit$, the maximiser is unique and satisfies
		$x_i=y_i$ for all $i=1,\dotsc,r$. Without loss of generality  we can
		let $c=1$. 
		Using the fact that the maximiser must be of the form
		\eqref{eq:maxformleq2}, and setting $x_1=x$, $y_1=y$, we can write 
		\be
		\begin{split}
			F(\vec{x};\vec{y})=F_0(x,y):=&
			\b \Big(xy+
			\tfrac{(\frac12-x)(\frac12-y)}{r-1}\Big)
			-x\log x-y\log y\\
			&-
			\big(\tfrac12-x\big)\log\tfrac{\frac12-x}{r-1}-
			\big(\tfrac12-y\big)\log\tfrac{\frac12-y}{r-1}.
		\end{split}
		\ee
		We are maximising $F_0$ in the box $[\frac{1}{2r},\frac12]^2$.  
		Calculations yield that when $x>y$, $\frac{\partial F_0}{\partial
			x}<\frac{\partial F_0}{\partial y}$, and vice-versa, so that the
		maximum points of $F_0$ must satisfy $x=y$ or lie on the boundary.
		Lemma \ref{lb} shows that they cannot lie
		on the boundary unless $(\vec x;\vec y)=\om_0$. 
		So, substituting $x=y$, and reparametrising with
		$z=2x$, we have  
		\be
		F_0\big(\tfrac{z}{2},\tfrac{z}{2}\big)=
		\tfrac{\b}{4}\big(z^2+\tfrac{(1-z)^2}{r-1}\big)
		-z\log z -(1-z)\log\tfrac{1-z}{r-1}+\log2.
		\ee
		Now, apart from the constant $\log 2$, this is precisely the function
		maximised in \cite[Theorem 1.1]{Bjo16}, with $\beta$ in that paper
		replaced with $\beta/2$ here, and $\vec{x}$ in that paper of the form
		$x_1\ge x_2=\cdots=x_r$.  By the working in that paper and the Appendix
		of \cite{BFU20}, the maximiser is unique for all
		$\beta\neq\beta_\crit=\frac{4(r-1)\log(r-1)}{r-2}$ from
		\eqref{eq:tcrittemp}. This concludes the proof of Proposition
		\ref{prop:maxc>0}. 
	\end{proof}
	
	It would be interesting to determine the structure of the maximisers also for $c<0$, but
	that seems more difficult than the case $c>0$ considered above. 
	It is still true that any maximiser
	has the form \eqref{xys}, where the points $(\xi_i,\eta_i)$ solve a system of the form	\eqref{lmf}. 
	However, it is no longer true that all maximisers satisfy $m=2$ or $k_1=1$. In fact, in
 Proposition~\ref{prop:pdc<0} we will see that more complicated maximisers exist even in  the zero-temperature limit $\beta\rightarrow\infty$. 
	

\section{The ground-state phase diagram}
\label{sec:gspd}

In this section we justify the ground-state phase diagrams 
given in Figures~\ref{fig:PDintro-c>0} and
\ref{fig:PDintro-c<0} 
of the introduction. In the zero temperature limit $\beta\rightarrow\infty$, the logarithmic terms in
the function
$F(\vec x;\vec y)$ 
of \eqref{eq:F} become  negligible, and the maximisation problem
in Theorem \ref{thm:FE-AB} and \ref{thm:FE-WB} reduces to maximising the function
\be
    G(\vec{x};\vec{y})=\sum_{i=1}^r Q(x_i,y_i)
    =\sum_{i=1}^r \frac{1}{2}\left(ax_i^2+by_i^2+2cx_iy_i\right)
\ee
on the domain $\Omega$ defined in \eqref{eq:Om}.
We will determine all maximisers of $G$ for  $c\neq 0$, starting with the 
easier case $c>0$. As has been mentioned, the case $c=0$ can be reduced to results of
\cite{Bjo16}.




\subsection{Diagram for $c>0$}

We first introduce some notation.
For fixed $c$, we split the $ab$-plane into five disjoint regions, defined by
\begin{align*}
D&=\left\{a,\,b<0,\ ab>c^2\right\},& \partial D&=\left\{a,\,b<0,\ ab=c^2\right\},\\
E_1&=\left\{b\leq\frac{-c\rho}{\rho'},\ ab<c^2\right\},&
E_2&=\left\{a\leq\frac{-c\rho'}{\rho},\ ab<c^2\right\},\\
F&=\left\{a>\frac{-c\rho'}{\rho},\ b>\frac{-c\rho}{\rho'}\right\}.
\end{align*}
We refer to $D$ as the disordered and $F$ as the ferromagnetic region. The regions 
$E_1$ and $E_2$ are intermediate between $D$ and $F$. This is illustrated in 
  Figure \ref{fig:PDintro-c>0}.

We also introduce the following points in $\mathbb R^r\times \mathbb R^r$:
 \begin{align*}
 \om_D&=\left({\frac \rho r,\dots,\frac\rho r};{\frac {\rho'}r,\dots,\frac{\rho'} r}\right),\\
  \om_{E_1}&=\left(\rho,{0,\dots,0};\frac{b\rho'-(r-1)c\rho}{br},{\frac{b\rho'+c\rho}{br},\dots,\frac{b\rho'+c\rho}{br}}\right),\\
 \om_{E_2}&=\left(\frac{a\rho-(r-1)c\rho'}{ar},{\frac{a\rho+c\rho'}{ar},\dots,\frac{a\rho+c\rho'}{ar}};\rho',{0,\dots,0}\right),\\
 \om_F&=\left(\rho,{0,\dots,0};\rho',{0,\dots,0}\right).
 \end{align*}
 (Above, we used the notation $\omega_D=\omega_0$.)

The following result completely describes the maximisers of $G\big|_\Omega$.
As before, we may restrict attention to maximisers
$(\vec x^\star;\vec y^\star)$ such that $x_i^\star$ and $y_i^\star$ are decreasing.

\begin{proposition}\label{prop:pdc>0}
Assume that $c>0$ and let
$\om^\star=(\vec x^\star;\vec y^\star)$ be a maximiser of
 $G\big|_{\Om}$ with $x_i^\star$ and $y_i^\star$ decreasing.
 If $(a,b)\in X$, where $X$ is one of $D$, $E_1$, $E_2$ and $F$,
 then $\om^\star$ is unique and equals $\om_X$.
 In the remaining case
 $(a,b)\in \partial D$ there are infinitely many maximisers. Explicitly, they are given by all points  $(x^\star;y^\star)\in \Omega$ such that
    \begin{equation}\label{max5}
\sqrt{-a}\left(x_i^\star-\frac\rho r\right)=\sqrt{-b}\left(y_i^\star-\frac{\rho'}r\right),\qquad 
    1\leq i\leq r.\end{equation}
\end{proposition}

\begin{proof}

We
first consider the case when $Q$ is negative semidefinite, that is,   $(a,b)\in\bar D$. Recall the identity  \eqref{eq:psiid}, which can be written 
\begin{equation}\label{gom}G( \vec x;\vec y)=G(\om_D)+\frac 1 r{\sum_{1\leq i<j\leq r}}Q(x_i-x_j,y_i-y_j).
 \end{equation}
 As we already saw in the proof of Proposition \ref{prop:crittemp}, this immediately implies that $\om_D$ is the unique maximiser in case $D$. If $(a,b)\in\partial D$, then 
 \begin{equation}\label{qsemidef}Q(x,y)=-\frac 12(\sqrt{-a}x-\sqrt{-b}y)^2.\end{equation}
  Then, \eqref{gom} implies that
 $G$ is maximised at all points such that
 $\sqrt{-a}\,x_i-\sqrt{-b}\,y_i $ is independent of $i$. Summing over $i$ gives
 $r(\sqrt{-a}\,x_i-\sqrt{-b}\,y_i)=\sqrt{-a}\,\rho-\sqrt{-b}\,{\rho'}$, which leads
 to \eqref{max5}. Note that if $(x_1^\star,\dots,x_r^\star)$ is any decreasing sequence of non-negative numbers summing to $\rho$ and we solve \eqref{max5}
 for $y_i^\star$, then $(x^\star;y^\star)\in\Omega$ provided that
 \begin{equation}\label{xrcond}x_r^\ast\geq  \frac{\rho}r-\sqrt{\frac
     ba}\frac{\rho'}{r}.\end{equation}
Since the right-hand-side is $<\rho/r$,
this shows that the number of maximisers is indeed infinite in this case.
 
 From now on we assume that $Q$ is not negative semidefinite.
Let $k$ and $l$ denote the number of non-zero entries in $x^\star$ and $y^\star$, respectively. Suppose first that $k\leq l$. Then, $\omega^\star$ is a maximiser of
\[
H(\vec x;\vec y)=\sum_{j=1}^k Q(x_j,y_j)+\sum_{j=k+1}^l Q(0,y_j) 
\]
on the  set
\[
U=\left\{(\vec x;\vec y);\,x_1,\dots,x_k,y_1,\dots,y_l>0,\
\textstyle
  \sum_{j=1}^kx_j=\rho,\ \sum_{j=1}^l y_j=\rho'\right\}. 
\]
There must then exist Lagrange multipliers $\lambda$ and $\mu$ such that
\begin{subequations}\label{glagmult}
\begin{align}
\label{dhdx}\frac{\partial H}{\partial x_j}(\omega^\star)&=ax_j^\star+cy_j^\star=\lambda,\quad &1\leq j\leq k,\\
\label{dhdy}\frac{\partial H}{\partial y_j}(\omega^\star)&=cx_j^\star+by_j^\star=\mu,\quad &1\leq j\leq k,\\
\label{dhdyb}\frac{\partial H}{\partial y_j}(\omega^\star)&=by_j^\star=\mu,\quad &k+1\leq j\leq l.
\end{align}
\end{subequations}

If $ab\neq c^2$,  the system
\eqref{dhdx}--\eqref{dhdy} has a unique solution, so $x_1^\star=\dots=x_k^\star$
and $y_1^\star=\dots=y_k^\star$. This also holds  if $ab=c^2$, where $a,\,b>0$.  In that case,   \eqref{dhdx} 
gives $a(x_1^\star-x_j^\star)+c(y_1^\star-y_j^\star)=0$ for  $j\leq k$. Since 
$a>0$ 
and $c>0$, we can still conclude  that $x_1^\star=x_j^\star$ and $y_1^\star=y_j^\star$.

If  $b\neq 0$, \eqref{dhdyb} gives 
$y_{k+1}^\star=\dots=y_l^\star$. Again, this also holds for $b= 0$. 
Indeed, in that case, if $k<l$, then 
\eqref{dhdy} gives $cx_k^\star=\mu$ and  \eqref{dhdyb} gives $0=\mu$.
This is impossible since $c$ and $x_k^\star$ are both assumed positive.
 Thus,  $k=l$ and the equalities $y_{k+1}^\star=\dots=y_l^\star$ are trivially valid.

The above arguments show that, under the assumption $k\leq l$,
$$\omega^\star=(\underbrace{x_1^\star,\dots,x_1^\star}_k,\underbrace{0,\dots,0}_{r-k};\underbrace{y_1^\star,\dots,y_1^\star}_k,\underbrace{y_{l}^\star,\dots,y_{l}^\star}_{l-k},\underbrace{0,\dots,0}_{r-l}). $$
Next, we prove that either $l=k$ or $l=r$. To see this, assume that
$k<l<r$. On the one hand,  \eqref{dhdy} and \eqref{dhdyb}
give $\mu=cx_1^\star+by_1^\star=by_{l}^\star$. This implies $b(y_{l}^\star-y_1^\star)=cx_1^\star>0$
and hence  $b<0$. On the other hand, 
if $t$ is a small positive number, then 
$(\vec x^\star;\vec y^\star+t(e_{l+1}-e_l))\in U$ and hence
$G(\vec x^\star;\vec y^\star)\geq G(\vec x^\star;\vec y^\star+t(e_{l+1}-e_l))$, where $e_j$ are unit vectors. It follows that
\[
\textstyle
0\geq \frac{\partial H}{\partial y_{l+1}}(\om^\star)-\frac{\partial H}{\partial y_{l}} (\om^\star)
=c(x_{l+1}^\star-x_l^\star)+b(y_{l+1}^\star-y_l^\star)=-b y_l^\star,
\]
 which contradicts $b<0$. After a change of variables, we
 conclude that 
 \begin{equation}\label{omprem}\om^\star=(\underbrace{x_1^\star,\dots,x_1^\star}_k,\underbrace{0,\dots,0}_{r-k};\underbrace{y_1^\star,\dots,y_1^\star}_k,\underbrace{y_{2}^\star,\dots,y_{2}^\star}_{r-k}),\end{equation}
 where the previous cases $l=k$ and $l=r$ correspond to $y_2^\star=0$ and 
 $y_2^\star\neq 0$, respectively.

If $k>1$ in \eqref{omprem} then 
\begin{multline}\label{gvar}G(\vec x^\star+t(e_1-e_k);\vec y^\star+u(e_1-e_k)) 
-G(\om^\star)\\
=Q(x_1+t,y_1+u)+Q(x_1-t,y_1-u)-2Q(x_1,y_1)
=2 Q(t,u).
\end{multline}
Since we assume that $Q$ is not negative semidefinite, it assume positive values in any neighborhood of $(0,0)$. 
This contradicts that 
$\om^\star$ is a maximiser. It follows that $k=1$, that is,
\begin{equation}\label{omspec}\om^\star=(\rho,0,\dots,0;y_1^\star,y_{2}^\star,\dots,y_{2}^\star).\end{equation}

If \eqref{omspec} holds with $y_2^\star=0$ then  $y_1^\star=\rho'$,
that is, $\om^\star=\om_F$.
 If  $y_2^\star\neq 0$, then the variables $y_j^\star$ can be determined from
\[
y_1^\star+(r-1)y_2^\star=\rho',\qquad c\rho+by_1^\star=by_2^\star, 
\]
where the second equation follows from \eqref{dhdy} and \eqref{dhdyb}.
Solving these equations, we find that $\om^\star=\om_{E_1}$. 

So far we have assumed that $k\leq l$. The complementary 
  case  follows by interchanging the roles of the $x$- and $y$-variables. It leads to the additional possibility
$\om^\star=\om_{E_2}$. That is, if $(a,b)\in E_1\cup E_2\cup F$, then the maximum is achieved at one of the points $\om_{E_1}$, $\om_{E_2}$ and $\om_F$.

It is easy to check that, at the point $\om_{E_1}$, the conditions $y_1^\star\geq y_2^\star\geq 0$ are equivalent to $b\leq -c\rho/\rho'$.
Likewise,   $\om_{E_2}$ is only an admissible point if $a\leq -c\rho'/\rho$.
In region $F$, neither of these conditions hold and the only possibility is
$\omega^\star=\om_F$. In region $E_1$, we have ruled out $\om_{E_2}$, so we only need to compare the values at $\om_{E_1}$ and $\om_F$.
By an elementary computation, 
 $$G(\om_F)-G(\om_{E_1})=\frac{(r-1)(c\rho+b\rho')^2}{2br}\leq 0 $$
 since $b<0$ in this case. Equality holds only at the boundary with region $F$, where $\om_{E_1}=\om_F$.
 This proves the result in case $E_1$ and case $E_2$ follows by symmetry.
\end{proof}

To give an example of how the model behaves in the different regions, we 
compute the  magnetisation (see Theorem \ref{thm:mag})
\[
\mathcal M=\frac{\partial \Phi^\ab}{\partial h}\Big|_{h \downarrow 0}=
			\sum_{i=1}^r (x_i^\star+y_i^\star) w_i.
\]
We will
assume that $(a,b)\notin \partial D$ 
and that 	$w_1+\dots+w_r=0$. Since
$x_2^\star=\dots= x_r^\star$ and $y_2^\star=\dots= y_r^\star$ we obtain
\[
\mathcal M
                =(x_1^\star+y_1^\star-x_2^\star-y_2^\star)w_1. 
\]
Inserting the explicit expressions from Proposition \ref{prop:pdc>0} gives 
\[
\mathcal M=\begin{cases}
0, &  (a,b)\in D,\\
\left(1-\frac cb\right)\rho w_1, & (a,b)\in E_1,\\
\left(1-\frac ca\right)\rho'w_1, &   (a,b)\in E_2,\\
w_1, & (a,b)\in F.\end{cases}
\]
We see that $\mathcal M$ has a discontinuity across the curve $\partial D$.
At the half-lines separating region $F$ from $E_1$ and $E_2$, it is  continuous but not differentiable.

\subsection{Diagram for $c<0$}

We now turn to the case $c<0$. As before, we view $c$ as fixed and
describe the phase diagram in the $ab$-plane; see
Figure~\ref{fig:PDproof-c<0}.  
There is then an anti-ferromagnetic phase 
\be\label{aj}
A=\{a,\,b>0\},
\ee
and a disordered phase
\be\label{dj}
D=\{a,\, b< 0,\ ab>c^2\},
\ee
which agrees with the case $c>0$.
There are also a number of intermediate phases. To describe them
geometrically, we introduce the points 
\begin{equation}\label{pj}
P_k=\left(\frac{k\rho'c}{(r-k)\rho},\frac{(r-k-1)\rho c}{(k+1)\rho'}\right),\qquad k=1,2,\dots,r-2, 
\end{equation}
which are all in the region $\{a,b<0,\,ab<c^2\}$, and
\be\label{qj}
Q_k=\left(\frac{k\rho'c}{(r-k)\rho},\frac{(r-k)
\rho c}{k\rho'}\right),\qquad k=1,2,\dots,r-1
\ee
which are on $\partial D=\{a,b<0,\,ab=c^2\}$.  We draw $r-2$ line segments
connecting the origin $a=b=0$ to the points $P_j$. We also draw a zig-zag line, consisting of the horizontal half-line
to the right of $Q_1$, a vertical line segment from $Q_1$ to $P_1$, a horizontal segment from $P_1$ to $Q_2$, a vertical segment from $Q_2$ to $P_2$, continuing in this way and ending with the vertical half-line above $Q_{r-1}$. 
Together with the boundaries of $A$ and $D$, these line segments 
divide the plane into $2r-1$ additional open regions. We will write 
$B_1,\dots,B_{r-1}$ for the regions above and 
$C_1,\dots,C_r$
for those below the zig-zag line, in both cases numbered from
southeast to northwest.  
More explicitly,
\be\label{bj}
\begin{split}
B_1&=\left\{a>\frac{\rho' c}{(r-1)\rho},\, \frac{(r-1)\rho c}{\rho'}< b<0,\, 
\rho^2(r-1)(r-2)a>2(\rho')^2 b
\right\}, \\
B_k&=\bigg\{a>\frac{k\rho'  c}{(r-k)\rho},\, b> \frac{(r-k)\rho c}{k\rho' },\\
&\qquad(r-k)(r-k-1)\rho^2a>k(k+1)(\rho')^2b,\\
&\qquad(r-k+1)(r-k)\rho^2 a<(k-1)k(\rho')^2 b
\bigg\},\qquad 2\leq k\leq r-2, \\
B_{r-1}&=\left\{\frac{(r-1)\rho' c}{\rho}<a<0,\, b> \frac{\rho c}{(r-1)\rho'},\, 
2\rho^2 a< (\rho')^2 (r-1)(r-2)b
\right\},
\end{split}
\ee
and
\be\label{cj}
\begin{split}
C_1&=\left\{b<\frac{(r-1)\rho c}{\rho'},\,ab<c^2\right\},\\
C_k&=\left\{a< \frac{(k-1)\rho' c}{(r-k+1)\rho},\,b<\frac{(r-k)\rho c}{k\rho'},\,ab<c^2\right\},\qquad 2\leq k\leq r-1,\\
C_r&=\left\{a< \frac{(r-1)\rho'c}{\rho},\,ab<c^2\right\}.
\end{split}
\ee
As before, we write
\[
\om_D=\left({\frac \rho r,\dots,\frac\rho r};{\frac {\rho'}r,\dots,\frac{\rho'} r}\right).\]
The maximiser in the anti-ferromagnetic phase is
\[ 
\om_A=\left(\rho,{0,\dots,0};{0,\dots,0},\rho'\right).
\]
We will see that the intermediate regions correspond to the maximisers
\be\label{bkdef}
\om_{B_k}=\Big(\underbrace{\frac\rho k,\dots,\frac\rho
  k}_{k},\underbrace{0,\dots,0}_{r-k};\underbrace{0,\dots,0}_{k},\underbrace{\frac{\rho'}{r-k},\dots,\frac{\rho'}{r-k}}_{r-k}\Big)
\ee
and
\[
\om_{C_k}=\Big(\underbrace{x_1,\dots,x_1}_{k-1},x_2,\underbrace{0,\dots,0}_{r-k};\underbrace{0,\dots,0}_{k-1},y_1,\underbrace{y_2,\dots,y_2}_{r-k}\Big),
\]
where
\begin{subequations}\label{ckdef}
\begin{align}x_1&=\frac{(r+1-k)\rho a b+\rho'bc-(r-k)\rho c^2}{k(r+1-k)a b-(k-1)(r-k) c^2},\\
x_2&=\frac{(r+1-k)\rho ab-(k-1)\rho'bc}{k(r+1-k)a b-(k-1)(r-k) c^2},\\
y_1&=\frac{k\rho' ab-(r-k)\rho ac}{k(r+1-k)a b-(k-1)(r-k) c^2},\\
y_2&=\frac{k\rho' a b+\rho ac-(k-1)\rho' c^2}{k(r+1-k)a b-(k-1)(r-k) c^2}. 
\end{align}
\end{subequations}

    \begin{figure}[h]
    \centering
    \begin{tikzpicture}[scale=3.2]
    
      \path[fill = blueredred, opacity = 0.4] (0,0) -- (0,.5) --
    (-1.875,.5) -- (-1.875,-0.532) -- (-1.36,-0.532) -- cycle;
    \draw[font=\fontsize{8}{10}] (-0.9,0.2) node {$B_4$};
    \path[fill = red, opacity = 0.4](-1.875,.5) -- (-1.875,-0.532) -- (-2,-0.532) -- (-2,.5) -- (-1.875,.5);
    \draw[font=\fontsize{8}{10}] (-1.94,0.2) node {$C_5$};
    \path[fill = bluered, opacity = 0.4](-1.875,-0.532) -- (-1.36,-0.532)
    -- (-1.36,-0.735) -- cycle;
     \draw[font=\fontsize{8}{10}] (-1.5,-0.6) node {$C_4$};
\path[fill = bluebluered, opacity = 0.4](0,0) -- (-1.36,-0.532) --
    (-1.36,-0.735) -- (-0.735,-0.735) -- cycle;
    \draw[font=\fontsize{8}{10}] (-1.05,-0.6) node {$B_3$};
    \path[fill = blue, opacity = 0.4]  (-1.36,-0.735) --
    (-0.735,-0.735) -- (-0.735,-1.36 ) --cycle;
\draw[font=\fontsize{8}{10}] (-0.9,-0.9) node {$C_3$};
\path[fill = bluegreen, opacity = 0.4]  (-0.735,-0.735) -- (0,0) --
(-0.532,-1.36) --  (-0.735,-1.36) -- cycle; 
\draw[font=\fontsize{8}{10}] (-0.6,-1.05) node {$B_2$};
\path[fill = mygreen, opacity = 0.4]  (-0.532,-1.36) --  (-0.735,-1.36) 
-- (-0.532,-1.88) -- cycle; 
\draw[font=\fontsize{8}{10}] (-0.6,-1.5) node {$C_2$};
    \path[fill = yellowgreen, opacity = 0.4](0,0) -- (-0.532,-1.36) --
    (-0.532,-1.88) -- (.5,-1.88) -- (.5,0) -- cycle;
    \draw[font=\fontsize{8}{10}] (0.2,-0.9) node {$B_1$};
     \path[fill = yellow, opacity = 0.4] (-0.532,-1.88) -- (.5,-1.88) --
    (.5,-2) -- (-0.532,-2) -- cycle; 
    \draw[font=\fontsize{8}{10}] (0.2,-1.94) node {$C_1$};
    \fill [lightgray, domain=-2:-0.5, variable=\x](-2, -2) -- plot ({\x}, {1/(1*\x)}) -- (-0.125,-2) -- cycle;
    \draw[font=\fontsize{8}{10}] (-1.4, -1.3) node {$D$};
     \path[fill = myorange, opacity = 0.4] (0.5,0) -- (0,0) --
    (0,0.5) -- (0.5,0.5) -- cycle; 
    \draw[font=\fontsize{8}{10}] (0.2,0.2) node {$A$};

    \draw[->, >=stealth, line width=0.3pt](-2, 0) -- (.6,0);
    \draw[font=\fontsize{8}{10}] (.7, 0) node {$a$};
    \draw[->, >=stealth, line width=0.3pt](0,-2) -- (0,.6);
    \draw[font=\fontsize{8}{10}] (0, .7) node {$b$};

    \draw [smooth,samples=100,domain=-2:-0.5,variable=\x, red]
    plot(\x,{1/(1*\x)});
    \draw[-, >=stealth, line width=0.3pt, red] (-1.875,.5) -- (-1.875,-0.532);
    \draw[font=\fontsize{8}{10}] (-1.6,.6) node
    {$a=\frac{c\rho'(r-1)}{\rho}$};
    (0,0);
    \draw[-, >=stealth, line width=0.3pt, red](-1.875,-0.532) -- (-1.36,-0.532);
    \draw[-, >=stealth, line width=0.3pt, red](-1.36,-0.532) -- (-1.36,-.735);
    \draw[-, >=stealth, line width=0.3pt, red](-1.36,-0.735) -- (-0.735,-0.735);
    \draw[-, >=stealth, line width=0.3pt, red] (-0.735,-0.735) --
    (-0.735,-1.36);
    \draw[-, >=stealth, line width=0.3pt, red]   (-0.735,-1.36) --
    (-0.532,-1.36);
    \draw[-, >=stealth, line width=0.3pt, red]  (-0.532,-1.36) --  (-0.532,-1.88);
    \draw[-, >=stealth, line width=0.3pt, red]  (-0.532,-1.88) -- (.5,-1.88);
    \draw[font=\fontsize{8}{10}] (.8,-1.88) node {$b=\frac{c\rho(r-1)}{\rho'}$};
    \draw[-, >=stealth, line width=0.3pt, red](0,0) -- (-1.36,-0.532);
    \draw[-, >=stealth, line width=0.3pt, red](0,0) -- (-0.735,-0.735);
    \draw[-, >=stealth, line width=0.3pt, red](0,0) -- (-0.532,-1.36);

\filldraw [black] (-1.875,-0.532) circle (.3pt);
\draw[font=\fontsize{8}{10}] (-1.9,-0.6) node {$Q_4$};

\filldraw [black] (-1.36,-0.735)  circle (.3pt);
\draw[font=\fontsize{8}{10}] (-1.4,-0.8) node {$Q_3$};

\filldraw [black]  (-0.735,-1.36) circle (.3pt);
\draw[font=\fontsize{8}{10}] (-.78,-1.44) node {$Q_2$};

\filldraw [black]   (-0.532,-1.88)  circle (.3pt);
\draw[font=\fontsize{8}{10}] (-.58,-1.92) node {$Q_1$};

\filldraw [black] (-1.36,-0.532) circle (.3pt);
\draw[font=\fontsize{8}{10}] (-1.36,-.47) node {$P_3$};

\filldraw [black]  (-0.735,-0.735) circle (.3pt);
\draw[font=\fontsize{8}{10}] (-0.665,-0.745)  node {$P_2$};

\filldraw [black]  (-0.532,-1.36) circle (.3pt);
\draw[font=\fontsize{8}{10}] (-0.47,-1.37)  node {$P_1$};

    \end{tikzpicture}
    \caption{The ground state phase diagram
for $c<0$, in the case $r=5$, with the points $P_k$ \eqref{pj}
and $Q_k$ \eqref{qj} as well as the regions 
$A$ \eqref{aj},
$B_k$ \eqref{bj}, $C_k$ \eqref{cj} and
$D$ \eqref{dj}
indicated.
}
    \label{fig:PDproof-c<0}
\end{figure}
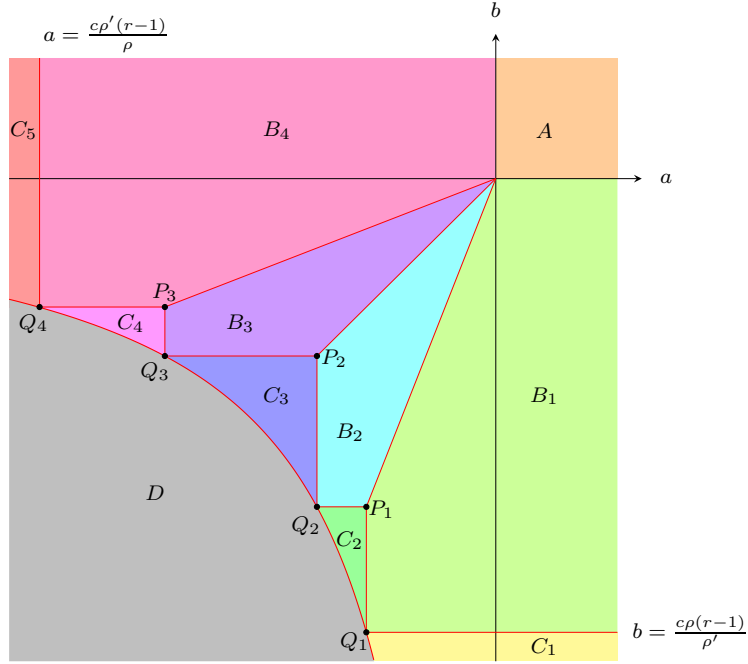

The complete description of the ground state phase diagram for $c<0$ is then as follows.

\begin{proposition}\label{prop:pdc<0}
Assume that $c<0$, $r\geq 3$ and let
$\om^\star=(\vec x^\star;\vec y^\star)$ be a maximiser of
 $G\big|_{\Om}$ with $x_i^\star$ decreasing and $y_i^\star$ increasing.
 If $(a,b)\in X$, where $X$ is one of $A$, $B_k$, $C_k$ or $D$,
 then $\omega^\star$ is unique and equal to $\omega_X$.
 If $(a,b)$ is in the interior of the  line segment separating $B_k$ from
 $C_k$, then $\omega^\star$ is also unique and given by $\omega^\star=\omega_{B_k}=\omega_{C_k}$.
 Likewise, if $(a,b)$ is in the interior of the line segment separating
 $B_{k}$ from  $C_{k+1}$ then 
 $\omega^\star=\omega_{B_k}=\omega_{C_{k+1}}$. If $(a,b)$ is in the interior of
 the line segment separating $B_k$ from $B_{k+1}$, then there are exactly two maximisers, namely, $\omega_{B_{k}}$ and $\omega_{B_{k+1}}$. If $(a,b)=P_k$
 (the corner between $B_k$, $B_{k+1}$ and $C_{k+1}$) then there are infinitely many maximisers, which form the line segment $t \omega_{B_k}+(1-t) \omega_{B_{k+1}}$
 for $0\leq t\leq 1$. In the remaining cases,
 $(a,b)\in\partial A$ or $(a,b)\in\partial D$
 there are also infinitely many maximisers. In the case $\partial D$ they are determined by the conditions
 \begin{equation}\label{semidmax}
\sqrt{-a}\left(x_i^\star-\frac\rho
     r\right)+\sqrt{-b}\left(y_i^\star-\frac{\rho'}r\right)=0,\qquad
   1\leq i\leq r,
\end{equation}
 in the case
 $a>0$, $b=0$ by the conditions
 \begin{subequations}\label{fbound}
 \begin{equation}\label{b0max}x_1^\star=\rho,\qquad x_2^\star=\dots=x_r^\star=y_1^\star=0, 
 \end{equation}
  in the case $a=0$, $b>0$ by the conditions
 \begin{equation}\label{a0max}x_r^\star=y_1^\star=\dots=y_{r-1}^\star=0,\qquad y_r^\star=\rho'\end{equation}
 and, finally, for $a=b=0$ by
 \begin{equation}\label{originmax}x_1^\star y_1^\star=\dots=x_r^\star y_r^\star=0. \end{equation}
 \end{subequations}
\end{proposition}

For convenience,
we formulated Proposition \ref{prop:pdc<0} only for $r\geq 3$. In the case $r=2$ the
 same statement is  correct, except for the fact that 
the equations \eqref{fbound} have the unique solution $\om=\om_A$. In this case $\om_{B_1}=\om_A$, so  $\partial A$ and $B_1$ should be considered as parts of the anti-ferromagnetic phase. Note also that   there are no points $P_k$, and only one region $B_k$. This leads to exactly the same diagram as for $c>0$. We already know this from the discussion after Theorem \ref{thm:FE-WB}.

The proof of Proposition \ref{prop:pdc<0} follows the same strategy as that of
Proposition \ref{prop:pdc>0}. Since the details are more involved, we divide it into a series of lemmas.

\begin{lemma}\label{lem:df}
Proposition \ref{prop:pdc<0} holds if $(a,b)\in\bar A$ or $(a,b)\in\bar D$.
\end{lemma}

\begin{proof}  The case $(a,b)\in D$ follows immediately from
\eqref{gom}. If  $(a,b)\in\partial D$,  \eqref{qsemidef} is replaced by
$$Q(x,y)=-\frac 12(\sqrt{-a}x+\sqrt{-b}y)^2.$$
This leads to the sign change in \eqref{semidmax} compared to \eqref{max5}.
Moreover the condition \eqref{xrcond} is replaced by
 $$x_1^\star\leq  \frac{\rho}r+\sqrt{\frac ba}\frac{\rho'}{r},$$
which shows that the number of maximisers is indeed infinite.

If $(a,b)\in \bar A$, that is, $a,\,b\geq 0$, we can  estimate
\begin{align*}Q(\vec x;\vec y)&=\sum_{j=1}^r\left( \frac a 2\,x_j^2+c x_jy_j +\frac b 2\,y_j^2\right)\\
&\leq \frac a 2(x_1+\dots+x_r)^2+\frac b 2(y_1+\dots+y_r)^2=\frac {a\rho^2+b(\rho')^2}2,\end{align*}
where we deleted the non-positive terms $cx_jy_j$ and added the non-negative terms $ax_ix_j$ and $by_iy_j$ for $i< j$. Equality holds if and only if all those terms vanish. If $a>0$ and $b>0$ this can only happen if $\om=\om_A$.
It is also clear that if $(a,b)\in\partial A$ it happens under the conditions
 \eqref{fbound}.
\end{proof}

\begin{lemma}\label{lem:ce} Assume that $(a,b)\notin\bar A\cup \bar D$.
Then the maximiser $\omega^\star$ in Proposition \ref{prop:pdc<0} is
equal to one of the points $\omega_{B_k}$, $\omega_{C_k}$  or  
 $t \omega_{B_k}+(1-t) \omega_{B_{k+1}}$
 for $0\leq t\leq1$. The last case can only happen if $(a,b)= P_k$.
 \end{lemma}

\begin{proof}
 Let
 $k$ and $l$ be the number of non-zero entries in $x^\star$ and $y^\star$, respectively. Then,  $\omega^\star$ is a maximiser of
$$\sum_{j=1}^{\min(k,r-l)} Q(x_j,0)+\sum_{j=r-l+1}^kQ(x_j,y_j)+\sum_{j=\max(k+1,r-l+1)}^r Q(0,y_j), $$
where the middle sum is empty if $k+l\leq r$. This gives the Lagrange multiplier equations
\begin{subequations}\label{lagr}
\begin{align}
\label{lagr1}ax_j^\star&=\lambda, & 1\leq j\leq\min(k,r-l),\\
\label{lagr2a}ax_j^\star+c y_j^\star&=\lambda, & r-l+1\leq j\leq k,\\
\label{lagr2b}cx_j^\star+b y_j^\star&=\mu, & r-l+1\leq j\leq k,\\
\label{lagr3}by_j^\star &=\mu, &\max(k+1, r-l+1)\leq j\leq r.
\end{align}
\end{subequations}

We will first show that the variables $x_j^\star$ and $y_j^\star$ involved in each group of equations  \eqref{lagr1},  \eqref{lagr2a}--\eqref{lagr2b} and
 \eqref{lagr3} are independent of $j$. 
 This is obvious if, respectively, $a\neq 0$, $ab\neq c^2$ (which holds by assumption) and $b\neq 0$.  By symmetry, it remains to consider the case $a=0$,
 when we must show that  $x_1^\star=\dots=x_{\min(k,r-l)}^\star$.
If $l=r$ there is nothing to prove. If $l<r$  and $k+l>r$ then 
  \eqref{lagr1} and \eqref{lagr2a} give $\lambda=ax_1^\star=0$ and  $\lambda=ax_k^\star+cy_{k}^\star=cy_k^\star$, which is impossible. 
  Finally, suppose $k+l\leq r$. Note that $b<0$ since 
   $(a,b)\notin\partial A$. It then follows from  \eqref{lagr3} that $y_j^\star=\rho'/l$ for $j\geq r-l+1$.
  This gives
 \[
G(\omega^\star)=\sum_{j=1}^k Q(x_j^\star,0)
+lQ\big(0,\tfrac{\rho'} l\big)
  =0+\frac{b(\rho')^2}{l},
\]
   which is maximised when $l=r-1$ and hence 
$k=1$, so the condition we want to prove holds automatically.

So far we have  proved that that, if  $k+l\leq r$,
\begin{equation}\label{ecase}\om^\star=\Big(\underbrace{\frac\rho k,\dots,\frac\rho k}_{k},\underbrace{0,\dots,0}_{r-k};\underbrace{0,\dots,0}_{r-l},\underbrace{\frac{\rho'}{l},\dots,\frac{\rho'}{l}}_{l}\Big), \end{equation}
and if $k+l>r$  (after a change of variables)
\begin{equation}\label{ccase}
\om^\star=\Big(\underbrace{x_1,\dots,x_1}_{r-l},\underbrace{x_2,\dots,x_2}_{k+l-r},\underbrace{0,\dots,0}_{r-k};\underbrace{0,\dots,0}_{r-l},\underbrace{y_1,\dots,y_1}_{k+l-r},\underbrace{y_2,\dots,y_2}_{r-k}\Big).
\end{equation}

In the case \eqref{ecase} we have
$$G(\om^\star)=kQ(\rho/k,0)+lQ(0,\rho'/l)=\frac{a\rho^2}{2k}+\frac{b(\rho')^2}{2l}. $$
Since we assume that at least one of $a$ and $b$ is negative, this can
only be a global maximum if $k+l=r$, that is, $\om^\star=\om_{B_k}$
(see \eqref{bkdef}).  
 
In the case \eqref{ccase}, we claim that $k+l=r+1$. Indeed, if $k+l\geq r+2$ we  find as in \eqref{gvar} that 
$$G(\vec x^\star+t(e_{r-l+1}-e_{r-l+2});\vec y^\star+u(e_{r-l+1}-e_{r-l+2})) =G(\om^\star)+2 Q(t,u),$$
which shows that $\om^\star$ is not a local maximum. We now know that
$$\om^\star=\Big(\underbrace{x_1,\dots,x_1}_{k-1},x_2,\underbrace{0,\dots,0}_{r-k};\underbrace{0,\dots,0}_{k-1},y_1,\underbrace{y_2,\dots,y_2}_{r-k}\Big), $$
where $1\leq k\leq r$.  Suppose first that $2\leq k\leq r-1$. Then,  the Lagrange equations \eqref{lagr} give
$$ax_1=ax_2+c y_1,\qquad cx_2 +b y_1=by_2.$$
Inserting $x_2=\rho-(k-1)x_1$ and $y_1=\rho'-(r-k)y_2$ gives
\begin{subequations}\label{ckeq}
\begin{align}
k ax_1+(r-k)cy_2&=a\rho+c\rho',\label{ckeqa}\\
(k-1) cx_1+(r-k+1)by_2&=c\rho+b\rho'. \label{ckeqb}
\end{align}
\end{subequations}
If the determinant $k(r+1-k)a b- (k-1)(r-k) c^2 \neq 0$, 
we can solve this system and find that
$\om^\star=\om_{C_k}$. If $k=1$, there is no $x_1$ and we must have  $x_2=\rho$.
We can still determine  $y_2$ from \eqref{ckeqb} and obtain
$\om^\star=\om_{C_1}$. Similarly, the case $k=r$ gives
$\om^\star=\om_{C_r}$. 

It remains to consider solutions of \eqref{ckeq} when 
\begin{equation}\label{vandet}k(r+1-k)a b= (k-1)(r-k) c^2\end{equation}
with $2\leq k\leq r-2$.
For solutions to exist we must have (from \eqref{ckeq})
\begin{align*}(k-1)c(a\rho+c\rho')&=ka(c\rho+b\rho'),\\
 (r-k+1)b(a\rho+c\rho')&=(r-k)c(c\rho+b\rho').\end{align*}
It is easy to solve this  for $(a,b)$, and obtain that either 
$(a,b)=(-c\rho'/\rho,-c\rho/\rho')$ or $(a,b)=P_{k-1}$. The first solution  does not satisfy \eqref{vandet} and can be discarded. At the point $P_{k-1}$, \eqref{ckeq} reduces to
\begin{equation}\label{line}k(k-1)\rho'x_1+(r-k)(r-k+1)\rho y_2=r\rho\rho'. \end{equation}
The conditions $x_1\geq x_2\geq 0$ and $0\leq y_1\leq y_2$ mean that
$(x_1,y_2)$ is in the rectangle $[\rho/k,\rho/(k-1)]\times[\rho'/(r-k+1),\rho'/(r-k)]$. The line
\eqref{line} passes through the corners $(\rho/k,\rho'/(r-k))$, $(\rho/(k-1),\rho'/(r-k+1))$ which correspond 
to the points $\om_{B_k}$ and $\om_{B_{k-1}}$. Thus, there are potential maximisers at the line segment between these points.
\end{proof}

It remains to pair up the maximisers with the correct region. 
 
 \begin{lemma}\label{lem:c}
 In the context of Lemma \ref{lem:ce}, if $\om^\star=\om_{C_k}$, then
 either $(a,b)\in \bar C_k$ or $(a,b)$ is on the extensions of the line segments
 separating $C_k$ from $B_k$ and $B_{k-1}$. In the latter case,
  $\om_{C_k}=\om_{B_k}$ and $\om_{C_k}=\om_{B_{k-1}}$, respectively.
 \end{lemma}

 \begin{proof}
 Since  $(a,b)\notin \bar A$, at least one of $a$ and $b$ is negative. Suppose that $a<0$.
 We compute
 \begin{equation}\label{gce}G(\om_{C_k})-G(\om_{B_k})=-\frac{a(k \rho'b-(r-k)\rho c)^2}{2 k(r-k)\Delta_k}, \end{equation}
 where $\Delta_k=k(r+1-k)a b- (k-1)(r-k) c^2$.
 If $\om_{C_k}$ is a global maximiser,
 it follows that either $\Delta_k>0$ or $k\rho'b=(r-k)\rho c$. 
 The second case is the extensions of the line segment
 separating $C_k$ from $B_k$. It is easy to verify that in that case $\om_{C_k}=\om_{B_k}$.
 If $\Delta_k>0$ then
 both $a$ and $b$ are negative. It is then clear from \eqref{ckdef}
 that the conditions $x_2,\,y_1\geq 0$ give $(a,b)\in\bar C_k$.
 
The case when $b<0$ follows in the same way, using instead
 \begin{equation}
 \label{gceb}G(\om_{C_k})-G(\om_{B_{k-1}})=-\frac{b((r+1-k)\rho a-(k-1)\rho' c)^2}{2(k-1)(r+1-k)\Delta_k}.
 \end{equation}
 \end{proof}
 
\begin{lemma}\label{lem:e}
In the context of Lemma \ref{lem:ce}, if $\om^\star=\om_{B_k}$, then
  $(a,b)\in \bar B_k$.
\end{lemma}

\begin{proof}
For $2\leq k\leq r-1$, we compute
$$G(\om_{B_k })-G(\om_{B_{k-1}})=\frac{k(k-1)(\rho')^2 a-(r-k)(r+1-k)\rho^2b}{2k(r-k)(k-1)(r+1-k)}. $$
It follows that, if $\om_{B_k}$ is a global maximiser, then $(a,b)$ is above or on the line separating $B_k$ from $B_{k-1}$. Replacing $k$ by $k+1$ we see that, if
$1\leq k\leq r-2$  then $(a,b)$ is below or on the line separating $B_k$ from $B_{k+1}$. This means that either $(a,b)\in \bar B_k$, $(a,b)\in C_k$ or $(a,b)\in C_{k+1}$. However, if $(a,b)\in C_k$ then the expression \eqref{gce} is strictly positive, so $\om_{B_k}$ is not a maximiser.
Similarly, the case $(a,b)\in C_{k+1}$ is excluded by \eqref{gceb}.
\end{proof}

We can now complete the proof of Proposition \ref{prop:pdc<0}. The case
$(a,b)\in\bar A\cup\bar D$ is handled by Lemma \ref{lem:df}. 
In all other cases except at the points $P_k$ it follows from  Lemma \ref{lem:ce} that $\om^\star=\om_{B_j}$ or $\om^\star=\om_{C_j}$ for some $j$. We can then use Lemma
\ref{lem:c} and Lemma \ref{lem:e} to exclude all possibilities for $\om^\star$ except those mentioned in Proposition \ref{prop:pdc<0}. In most cases this leaves a unique possibility. At the boundary between $B_k$ and $B_{k+1}$ there are two possibilities, but it is easy to verify (and clear from continuity arguments)
that $G(\om_{B_k})=G(\om_{B_{k+1}})$ in this case. At the points $P_k$ 
there are infinitely many possibilities, but it is again easy to verify (and clear from the Lagrange equations) that they are all maximisers.

	\section{Multi-block models}
	\label{sec:MB}
In this section we  generalise the free energy calculation of Theorem
	\ref{thm:FE-AB} to a class of models with 
	$p\geq1$ blocks rather than just the two blocks $A$ and $B$, and with certain
	many-body interactions. 
	
	We first need some
	notation. 
	Let $\g$ be a partition with all parts $>1$, that is
	$\g=(\g_1,\dotsc,\g_\ell)$ is a sequence of integers 
	$\g_1\geq \g_2\geq\dotsb\geq \g_\ell\geq2$.
	We say that a permutation $\s\in S_n$ has cycle-type $\g$ if its
	non-trivial cycles, ordered from longest to shortest, have lengths
	$\g_1,\dotsc,\g_\ell$.  Then $|\g|:=\g_1+\dotsb+\g_\ell\leq n$.
	Let $C_n^\g$ be the set of permutations in
	$S_n$ with cycle-type $\g$;  this is a conjugacy-class of $S_n$.  
	For example, if $\g=(2)$ then $C_n^\g=C_n^{(2)}$ 
	is the set of transpositions in $S_n$, 
	and if $\g=(3)$ then $C_n^\g=C_n^{(3)}$ 
	is the set of three-cycles in $S_n$. 
	Similarly, for $A\se\{1,2,\dotsc,n\}$, let $C_A^\g$ denote the set of
	permutations of the elements of $A$ with cycle-type $\g$.
	
	Let $A_1,\dotsc,A_p$ form a partition of $\{1,\dotsc,n\}$
	with $|A_k|=m_k$.
	Fix a finite set $\G$ of partitions 
	$\g$ with all parts $>1$.  We assume that $n$ and all $m_k$
	are large enough that  $C_n^\g\neq\es$ 
	and $C_{A_k}^\g\neq\es$  for all $\g\in\G$.  
	For $a^\g_1,\dotsc,a_p^\g,c^\g\in\RR$, 
	consider the Hamiltonian 
	\be\label{eq:H-MB}
	H_n^\mb=-n \sum_{\g\in\G}\Big(
	\sum_{k=1}^p \tfrac{a^\g_k}{|C_{A_k}^\g|}
	\sum_{\s\in C_{A_k}^\g} T_\s+
	\tfrac{c^\g}{|C_n^\g|}
	\sum_{\s\in C_n^\g} T_\s
	\Big),
	\ee
	and the partition function $Z_{n}^\mb(\b)=\tr_\VV[e^{-\b H^\mb_n}]$.
	Note that we have the scaling factor $n$ in front of \eqref{eq:H-MB}
	rather than $\tfrac1n$ as in \eqref{eq:H-AB}.  This is because the
	sizes of the conjugacy classes $C_A^\g$ depend on $n$,
	for example for transpositions we have $|C_n^{(2)}|=\binom{n}{2}$.
	
	The form of the Hamiltonian \eqref{eq:H-MB} means that spins at vertices in each block $A_k$ interact with each other through the many-body interaction $T_\s$ (as opposed to the pair-interaction $T_{i,j}=T_{(i,j)}$ before), with strength constants $a_k^\g$ dependent on the cycle type $\g$ of $\s$; as well as this, spins in all blocks together interact with each other similarly, this time with strength constants $c^\g$.
	
	The operators $T_\s$ appearing in \eqref{eq:H-MB}
	may all be written in terms of
	spin-matrices.  Indeed, for transpositions $\s=(i,j)$ this was
	discussed above, and for general $\s$ we may write $T_\s$ as
	a product of $T_{i,j}$'s.  However, we do not pursue an explicit
	formula for $T_\s$ in terms of spin-matrices.
	
	Our result about the free energy of this model is most compactly
	expressed in terms of positive semidefinite Hermitian 
	$r\times r$ matrices $X$.  For such a matrix, having eigenvalues
	$x_1,\dotsc,x_r\geq0$, we use the von Neuman entropy
\eqref{eq:neumann}.
	We have the following:
	
	\begin{theorem}\label{thm:FE-MB}
		Let $p\geq1$ be fixed,  and
		suppose that for all $k=1,\dotsc,p$ we have that 
		$m_k/n\to \rho_k\in(0,1)$ as $n\to\oo$.
		For the Hamiltonian \eqref{eq:H-MB}, we have that
		the free energy is given by
		\be
		\lim_{n\to\oo} \tfrac1n\log Z_{n}^\mb(\b)
		=\max \; \phi_\b( X_1,\dotsc,X_p),
		\ee
		where the maximum is taken over all  positive semidefinite Hermitian 
		$r\times r$ matrices $X_1,\dotsc,X_p$  with 
		$\tr[X_k]=\rho_k$, and where
		\be \label{eq:phi-mb}
		\begin{split}
			\phi_\b( X_1,\dotsc,X_p)
			&= \sum_{k=1}^p S(X_k)\\
			&+\b \sum_{\g\in\G} \Big(
			\sum_{k=1}^p a^\g_k\prod_{j\geq1} \tr[X_k^{\g_j}]+
			c^\g\prod_{j\geq1} \tr[(X_1+\dotsb+X_p)^{\g_j}]
			\Big).
		\end{split}
		\ee
	\end{theorem}

	Before proving Theorem \ref{thm:FE-MB} we discuss a few special cases.
	If we set $p=2$, $\G=\{(2)\}$ and 
	$a^{(2)}_1=(a-c)/2$, $a^{(2)}_2=(b-c)/2$ and 
	$c^{(2)}=c/2$, then 
	\be
	\phi_\b(X_1,X_2)= S(X_1)+S(X_2)+
	\tfrac\b2
	\tr\big[ a X_1^2+b X_2^2+ 2cX_1X_2  \big].
	\ee
	In fact, in this case we recover Theorem \ref{thm:FE-AB},
	i.e.\ we have  $\max\; \phi_\b(X_1,X_2)=\Phi^\ab_\b(a,b,c)$.
	For details, see the discussion around \eqref{eq:trace-ineq}.
	
	If instead we set $p=1$ and all $a^\g_k=0$ then
	\eqref{eq:H-MB} becomes
	\be\label{eq:H-MB-hom}
	H_n^\mb=-n
	\sum_{\g\in\G} \tfrac{c^\g}{|C_n^\g|}
	\sum_{\s\in C_n^\g} T_\s.
	\ee
	We thus obtain a
	homogeneous model of many-body interaction on the complete graph
	$K_n$.  (In fact, \eqref{eq:H-MB-hom} is the image of a general
	central element of $\CC[S_n]$ under the representation $T$.)  
	In this case we get that
	\be
	\tfrac 1n \log Z_{\b,n}^\mb \to
	\max \Big(
	-\sum_{i=1}^r x_i\log x_i +
	\b \sum_{\g\in\G} c^\g p_{\g}(x_1,\dotsc,x_r)
	\Big),
	\ee
	where the maximum is over all
	$x_1,\dotsc,x_r$ satisfying $x_i\geq0$ and 
	$\sum_{i=1}^r x_i=1$, 
	and where 
	$p_\g(x_1,\dotsc,x_r)$ denotes the power-sum symmetric polynomial
	\be\label{eq:powersum}
	p_\g(x_1,\dotsc,x_r)=\prod_{j=1}^\ell (x_1^{\g_j}+\dotsb+x_r^{\g_j}). 
	\ee
	
	It seems likely that Theorems \ref{thm:TS-twoblock}
	and \ref{thm:mag} can be extended to multi-block cases, though we do
	not pursue such extensions here.

We now turn to the proof of  Theorem \ref{thm:FE-MB},
which follows a similar pattern to that of 
Theorem \ref{thm:FE-AB}.  We start by writing
\be
H^\mb_n=-nT\Big(
\sum_{\g\in\G}\Big[
\sum_{k=1}^p a^\g_k\a^\g_{A_k}
+ c^\g \a^\g_n
\Big]
\Big)
=-n\sum_{\g\in\G}\Big[
\sum_{k=1}^p a^\g_k T(\a^\g_{A_k})
+ c^\g T(\a^\g_n)
\Big]
\Big),
\ee
where $T$ is the representation of $\CC[S_n]$ on $\VV$ 
given in \eqref{eq:T}, and 
	\be
	\a^\g_{A_k}=\frac1{|C^\g_{A_k}|}
	\sum_{\s\in C^\g_{A_k}} \s \in\CC[S_{A_k}],
	\qquad
	\a^\g_{n}=\frac1{|C^\g_{n}|}
	\sum_{\s\in C^\g_{n}} \s \in\CC[S_{n}].
	\ee
	As in \eqref{eq:swd} we have a decomposition
	\be
	\VV\cong \bigoplus_{\la\vdash n,\ell(\la)\leq r}
	\dim(U_\la) V_\la.
	\ee
	Here we consider $\VV$ as an $\CC[S_n]$-module only
	(we do not need the $\GL_r(\CC)$-part since we consider only the free
	energy and not correlations).
	As a 
	$\CC[S_{m_1}\times\dotsb\times S_{m_p}]$-module, we have the
	decomposition 
	\be
	V_\la\cong\bigoplus_{\mu(1),\dotsc,\mu(p)} 
	c_{\mu(1),\dotsc,\mu(p)}^\la 
	V_{\mu(1)}\otimes\dotsb\otimes
	V_{\mu(p)},
	\ee
	which generalises \eqref{eq:LR-def}.
	Here $\mu(k)\vdash m_k$ for each $k$ and the multiplicities 
	$c_{\mu(1),\dotsc,\mu(p)}^\la$ are analogs of the Littlewood--Richardson
	coefficients $c_{\mu,\nu}^\la$ and have many similar properties.  In
	particular, a full analog of Horn's inequalities holds:
	$c_{\mu(1),\dotsc,\mu(p)}^\la>0$ if and only if there are
	Hermitian matrices $M(1),\dotsc,M(p)$ with spectra 
	$\mu(1),\dotsc,\mu(p)$ such that 
	$M(1)+\dotsb+M(p)$ has spectrum $\la$
	(see Theorem 17 of \cite{fulton:horn}).
	
	Let us next see how $T(\a^\g_{A_k})$ and $T(\a^\g_n)$ act on these subspaces $V_{\mu(k)}$.  For $m\leq n$
	and $C=C^\g_m$ the conjugacy class of $\g$ in $S_m$, 
	consider $\a=\frac1{|C|}\sum_{\s\in C}\s\in\CC[S_m]$.
	For $\mu\vdash m$, since $\a$ is central in $\CC[S_m]$, it acts on the
	irreducible $V_\mu$ as a scalar, and in fact we have
	\be
	\a|_{V_\mu}=\frac{\chi_\mu(\a)}{d_\mu} \Id_{V_\mu}
	=\frac{\chi_\mu(\g)}{d_\mu} \Id_{V_\mu},
	\ee
	where $\chi_\mu(\g)$ is the character of $V_\mu$ evaluated at any
	permutation of  cycle-type $\g$.   This leads to the following
	expression analogous to \eqref{eq:Z-bi}:
	\be\begin{split}
		Z^\mb_n=\sum_{\la\vdash n,\ell(\la)\leq r} \dim(U_\la)
		&\sum_{\mu(1),\dotsc,\mu(p)} 
		c_{\mu(1),\dotsc,\mu(p)}^\la 
		d_{\mu(1)}\dotsb d_{\mu(p)} \\
		&\cdot\exp\Big(
		n\b \sum_{\g\in\G} \Big[
		\sum_{k=1}^p a^\g_k \tfrac{\chi_{\mu(k)}(\g)}{d_{\mu(k)}} 
		+c^\g \tfrac{\chi_{\la}(\g)}{d_{\la}} 
		\Big]
		\Big).
	\end{split}\ee
	
	As before, the relevant scaling for the limit 
	$\lim_{n\to\oo} \tfrac1n\log Z^\mb_n$ is given by letting 
	$\la/n\to\vec z$ and 
	$\mu(k)/n\to\vec x(k)$ for all $k$.  Also as before, $\dim(U_\la)$ is
	negligible on the relevant scale, and the $d_{\mu(k)}$ obey the
	asymptotics of \eqref{eq:d-asy}.  Below, we prove that 
	$c_{\mu(1),\dotsc,\mu(p)}^\la\leq (n+1)^{pr^2}$ which is also too small
	to contribute to the limit.  
	
	What remains is to identify the limits of the expressions of the
	form $\frac{\chi_\mu(\g)}{d_\mu}$.  The latter limits are well-known
	in the asymptotic representation theory of the symmetric group:
	Thoma's Theorem and the Vershik--Kerov Theorem 
	(see e.g.\ \cite[Corollary 4.2 and Theorem 6.16]{BO}) imply that
	if $\mu/n\to \vec x=(x_1,\dotsc,x_r)$, then
	\be
	\frac{\chi_\mu(\g)}{d_\mu}\to 
	p_{\g}(x_1,\dotsc,x_r), 
	\ee
	where $p_\g$ is the power-sum symmetric polynomial given in
	\eqref{eq:powersum}.  Writing $\vec x(k)=\lim_{n\to\oo} \mu(k)/n$
	and $\vec z=\lim_{n\to\oo} \la/n$, we conclude that the contributing 
	$\vec x(k)$ and $\vec z$ are spectra  of 
	Hermitian
	matrices $X_1,\dotsc,X_p$ and $Z=X_1+\dotsb+X_p$, respectively,
	where $\tr[X_k]=\rho_k$.
	Re-writing the free energy in terms of these matrices, 
	as in \eqref{eq:maxXY}
	and \eqref{eq:FE-WB-matrices}, 
	we obtain the claim \eqref{eq:phi-mb}.
	
	It remains to verify the bound $c_{\mu(1),\dotsc,\mu(p)}^\la\leq (n+1)^{pr^2}$.
	We  use the following 
	combinatorial description of $c_{\mu_1,\dotsc,\mu_p}^\la$ 
	which is mentioned just after 
	Proposition 13 of \cite{fulton:horn}.
	Form a skew
	shape $\nu$ by stacking $\mu(1),\dotsc,\mu(p)$ 
	from bottom left to top right,
	such that the lower left corner of $\mu({k})$ just touches the upper right
	corner of $\mu(k-1)$ as in Figure \ref{ytab-fig}.
	\begin{figure}
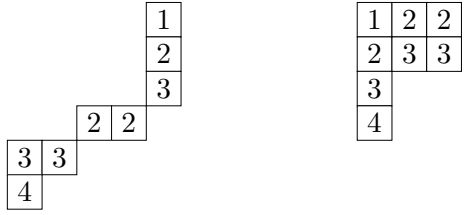

		\centering
		$\gyoung(::::;1,::::;2,::::;3,::;22,33,4)$\qquad\qquad\qquad
		$\young(122,233,3,4)$
		\caption{Left: A skew tableau with shape $\nu$ 
			formed from the three
			partitions $\mu(1)=(2,1)$, $\mu(2)=(2)$ and $\mu(3)=(1,1,1)$.
			Right: its rectification.}
		\label{ytab-fig}
	\end{figure}
	Fix any semistandard tableau 
	$\tau_\la$ of shape
	$\la$, to be concrete let us say that the first row of 
	$\tau_\la$ consists of
	$\la_1$ 1's, the second row of $\la_2$ 2's etc.  Then 
	$c_{\mu(1),\dotsc,\mu(p)}^\la$ is the number of semistandard tableaux 
	$\s_\nu$ of
	skew shape $\nu$ whose \emph{rectification} equals $\tau_\la$. 
	For a full description of the rectification, 
	see \cite[Section 1.2]{fulton:young}, but
	in brief terms the rectification is obtained by `sliding' the numbered
	boxes of $\s_\nu$ until a non-skew shape is obtained.  To see the
	claimed bound, note that in order to obtain the tableau $\tau_\la$,
	the number of boxes labelled 1 in $\nu$ must equal the number of boxes
	labelled 1 in $\la$, and similarly for labels 2, 3, etc.
	Thus, for each row of $\nu$ we have at most
	\[
	(\la_1+1)(\la_2+1)\dotsb(\la_r+1)\leq (n+1)^r
	\]
	choices of entries (from 0 to $\la_1$ 1's, from 0 to $\la_2$ 2's
	etc).  Since $\nu$ has at most $pr$ rows, the total number of choices
	is $\leq [(n+1)^r]^{pr}$, as claimed. \qed

	\appendix

	\section{The trace-inequality \eqref{eq:trace-ineq}}
	\label{sec:trXY}
	
	The inequality \eqref{eq:trace-ineq}
	appears e.g.\ in \cite[Prop.~9.H.1.g-h]{majorization}, but
	we give here  an almost self-contained proof based on Birkhoff's
	theorem, adapted from the discussion in \cite{elegant}.
	The
	problem is to maximise (respectively, minimise) $\tr[XY]$ subject to
	the condition that $X,Y$ are nonnegative definite Hermitian matrices
	with fixed spectra $x_1\geq x_2\geq \dotsb\geq x_r\geq 0$ and
	$y_1\geq y_2\geq \dotsb\geq y_r\geq 0$.  Equivalently, 
	since there are unitary matrices $U$ and $V$ such that 
	$U^*XU=D_x=\mathrm{diag}(x_1,\dotsc,x_r)$ and
	$V^*YV=D_x=\mathrm{diag}(x_1,\dotsc,x_r)$, the goal is to
	to extremise 
	\be
	\tr[UD_xU^*VD_yV^*]=\tr[D_xU^*VD_yV^*U]
	\ee
	over unitaries $U,V$.  Writing $W=U^*V$ we may equivalently extremise
	over the unitary $W$,
	\be
	\tr[D_x W D_yW^*]=\sum_{i,j=1}^r x_iw_{i,j}y_jw^*_{j,i}
	=\sum_{i,j=1}^r x_iy_j|w_{i,j}|^2.
	\ee
	Define the matrix $P=(p_{i,j})_{i,j=1}^r$ where
	$p_{i,j}=|w_{i,j}|^2$.  Since $W$ is unitary, $P$ is doubly
	stochastic (rows and columns sum to 1).  We have by the above
	\be
	\max_W \; \tr[D_x W D_yW^*]\geq 
	\max_P \sum_{i,j=1}^r x_iy_jp_{i,j},
	\ee
	where the second max is over doubly-stochastic matrices $P$ (and
	similarly for the min).  The function to be maximised on the
	right-hand-side is linear in $P$ and the set of doubly-stochastic
	matrices is convex and compact.  Thus the maximum (as well as the
	minimum) is attained at an extreme point of the set of
	doubly-stochastic matrices.  By Birkhoff's theorem  
	\cite[Theorem 2.A.2]{majorization}, the
	extreme points are the permutation matrices $\Pi$.  Since permutation
	matrices are real orthogonal (hence unitary) it follows that 
	\be
	\max_W \;\tr[D_x W D_yW^*]=\max_\Pi \;\tr[D_x \Pi D_y\Pi^*]
	\ee
	and similarly for the minimum.  Thus, we must only find the
	permutation $\pi$ which maximises or minimises the function
	\be
	\sum_{j=1}^r x_j y_{\pi(j)}.
	\ee
	The maximum is obtained for the identity permutation and the minimum
	for the reversal of $12\dotsc r$.  \qedhere

	\section{Equivalence of $Q_{i,j}$ and $P_{i,j}$ in the $\wb$-model}
	In this second appendix we study two representations of the walled
	Brauer algebra $\BB_{n,m}(r)$. We will prove in Lemma \ref{lem:iso_reps_WB} that they are isomorphic
	for all $r\ge2$. This will in particular give the equivalence of our
	$\wb$-model with the same model, but with each $Q_{i,j}$ replaced with
	$P_{i,j}$. More generally Lemma \ref{lem:iso_reps_WB} gives the same statement 
	on general graphs. To be precise, if $G=A\cup B$ 
	is any graph (with $A\cap B=\varnothing$), with $E_A$ the set of edges 
	between two vertices in $A$, $E_B$ similar, and $E_{AB}$ those between 
	a vertex of $A$ and a vertex of $B$, then for all $a,b,c\in\RR$, the following 
	two Hamiltonians are unitarily equivalent:
	\be\label{eq:AppendixB-2models-bipartite}
	    \begin{split}
	        H&=-\sum_{\{i,j\}\in E_A} aT_{i,j} - \sum_{\{i,j\}\in E_B} bT_{i,j} - \sum_{\{i,j\}\in E_{AB}} cP_{i,j} \\
	        H'&=-\sum_{\{i,j\}\in E_A} aT_{i,j} - \sum_{\{i,j\}\in E_B} bT_{i,j} - \sum_{\{i,j\}\in E_{AB}} cQ_{i,j}.
	    \end{split}
	\ee
	This in particular shows that the models with interactions $P_{i,j}$ and 
	$Q_{i,j}$ are equivalent on any bipartite graph; the equivalence of partition 
	functions was proved by Aizenman and Nachtergaele in \cite{an}. The same 
	statement (and in fact slightly stronger) holds on non-bipartite graphs, 
	but only for $r$ odd. Indeed,
    \eqref{eq:AppendixB-2models-bipartite} is very similar to a statement 
    on the model \eqref{eq:ryan}: for any graph $G$ with edge set $E$, for 
    any $L_1,L_2\in\RR$, the following two Hamiltonians are unitarily equivalent 
    for $r$ odd:
	\be\label{eq:AppendixB-2models-complete}
	    \begin{split}
	        H&=-\sum_{\{i,j\}\in E} L_1T_{i,j}+L_2P_{i,j}\\
	        H'&=-\sum_{\{i,j\}\in E} L_1T_{i,j}+L_2Q_{i,j}.
	    \end{split}
	\ee
	This is proved with Lemma B.1 of \cite{Ryan21}, which is the equivalent of our Lemma \ref{lem:iso_reps_WB} below, but for the full Brauer algebra.\\
	
	The representations we consider are defined as follows.
	First, we let $|a\rangle$ denote the standard basis for $\CC^r$,
	indexed using $a\in\{-S,-S+1,\dotsc,S\}$ where $S=(r-1)/2$, and recall
	that $\VV=(\CC^r)^{\otimes n}$.
	Let $T:\BB_{n,m}(r)\to \End(\VV)$ satisfy
	\begin{equation}
	T(\overline{i,j})
	=
	Q_{i,j}, 
	\hspace{1cm}
	T(i,j)
	=
	T_{i,j},
	\end{equation}
	where we  recall that $T_{i,j}$ is the transposition operator, and $\langle
	a_i, a_j| Q_{i,j} |b_i, b_j \rangle = \delta_{a_i, a_j} \delta_{b_i,
		b_j}$.  Similarly, define $\tilde{T}: \BB_{n,m}(r)\to \End(\VV)$ by
	\begin{equation}\label{eqn:defn_of_B_n's_action}
	\tilde{T}(\overline{i,j})
	=
	P_{i,j}, 
	\hspace{1cm}
	\tilde{T}(i,j)
	=
	T_{i,j},
	\end{equation}
	where we recall that $\langle a_i, a_j| P_{i,j} |b_i, b_j \rangle =
	(-1)^{a_i-b_i}\delta_{a_i, -a_j} \delta_{b_i, -b_j}$.  
	\begin{lemma}\label{lem:iso_reps_WB}
		For all $r\ge2$, and all $n$, the representations $T$ and 
		$\tilde{T}$ of $\BB_{n,m}(r)$ are isomorphic via a unitary transformation.
	\end{lemma}

	\begin{proof}
		The proof follows closely that of Lemma B.1 of \cite{Ryan21}.
		For $r$ odd, the lemma actually follows from that result
		by restricting the two representations there to the walled Brauer
		algebra.  So let $r$ be even.
		The elements $(i,j)$ and $(\overline{i,j})$
		generate the algebra $\BB_{n,m}(r)$, so we aim to find an
		invertible linear function $A: \VV \to\VV$ such that  
		\begin{equation}\label{eqn:isomorphic_reps_condition}
		A^{-1}T_{i,j}A = T_{i,j},
		\end{equation}
		for all $1\le i<j\le m$ and $m<i<j\le n$, and
		\be
		A^{-1}Q_{i,j}A = P_{i,j},
		\ee
		for all $1\le i\le m<j\le n$. By the Schur--Weyl duality for the
		general linear and symmetric groups \eqref{eq:swd}, the first
		condition holds if and only if $A = \alpha^{\otimes
			m}\otimes\gamma^{\otimes n-m}$ for some $\alpha,\gamma \in
		\GL_r(\CC)$. Then the second condition also holds if and only if
		$(\alpha\otimes\gamma)^{-1}Q_{i,j}(\alpha\otimes\gamma) = P_{i,j}$
		for all $1\le i\le m<j\le n$, which holds if and only if: 
		\begin{equation}\label{eq:bla1}
		\begin{split}
		(-1)^{a_i-b_i}\delta_{a_i, -a_j} \delta_{b_i, -b_j}
		&=
		\sum_{c_i,c_j,d_i,d_j} (\alpha^{-1})_{a_i,c_i} (\gamma^{-1})_{a_j,c_j} 
		\delta_{c_i, c_j} \delta_{d_i, d_j}
		\alpha_{d_i,b_i} \gamma_{d_j,b_j}\\
		&=
		\sum_{c,d} (\alpha^{-1})_{a_i, c} (\gamma^{-1})_{a_j, c}
		\alpha_{d,b_i} \gamma_{d,b_j}\\
		&=
		(\alpha^{-1} \gamma^{-\intercal})_{a_i, a_j}
		(\alpha^{\intercal } \gamma)_{b_i, b_j}.
		\end{split}
		\end{equation}
		Now recall that we assumed $r$ to be even, meaning that  $S$ and all
		the indices $a_i,a_j,b_i,b_j$ are odd multiples of $\tfrac12$.  Thus
		$(-1)^{a_i}=-(-1)^{-a_i}$ and \eqref{eq:bla1} holds if 
		\begin{equation}\label{eq:antidiag}
		\alpha^\intercal \gamma=
		-(\g^\intercal \a)^{-1}
		=
		\begin{bmatrix}
		& & & & (-1)^{-S}\\
		& & & (-1)^{1-S} & \\
		& & \reflectbox{$\ddots$} & & \\
		& (-1)^{S-1} & & & \\
		(-1)^{S} & & & &
		\end{bmatrix}.
		\end{equation}
		The matrix on the right  in \eqref{eq:antidiag} 
		is an involution whose transpose is its negative, so it suffices to check
		this for $\a^\intercal \g$.  Further, the matrix
		consists 	of the block matrices 
		$(-1)^{r/2}\big[\begin{smallmatrix} 0 & {i} \\ 
		{-i}  & 0 \end{smallmatrix}\big]$ aligned along the antidiagonal, 
		where $i=\sqrt{-1}$. 
		
		Such a pair $\alpha$, $\gamma$ exists: for example let
		\begin{equation*}
		g_1
		=
		\frac{1}{\sqrt{2}}
		\begin{bmatrix}
		i & i\\
		-1 & 1
		\end{bmatrix},
		\qquad 
		g_2
		=
		\frac{1}{\sqrt{2}}
		\begin{bmatrix}
		-1 & 1\\
		-i & -i
		\end{bmatrix},
		\end{equation*}
		take $\a$ to be
		block-antidiagonal with blocks $g_1$, and
		take $\g$ to be block-diagonal with blocks $(-1)^{r/2} g_2$.
		Since
		$g_1^\intercal g_2=\big[\begin{smallmatrix} 0 & {i} \\ 
		{-i}  & 0 \end{smallmatrix}\big]$, $\a^\intercal \g$ is as required.
		Further, since  both $\a$ and $\g$  are unitary,  so is $A$. 
	\end{proof}
	
	We can further prove the following statement, that in the $S=1$ ($r=3$) case, under a certain choice of the isomorphism of representations, the spin matrices are anti-symmetric. This verifies that we can use Theorems \ref{thm:mag} and \ref{thm:TS-twoblock} on the $S=1$ ($r=3$) nematic model with magnetisation term given by a spin matrix $S^{(k)}$, $k=1,2,3$, at each vertex, as noted at the end of Section \ref{sec:heuristics}.
	
	\begin{lemma}\label{lem:AppendixB-mag}
	    For all $k=1,2,3$, there exists a (unitary) isomorphism $\psi_n=\psi^{\otimes n}$ of the representations $T$ and $\tilde{T}$ of $\BB_{n,m}(3)$ (with $\psi_n^{-1}\tilde{T}(b)\psi_n=T(b)$ for all $b\in\BB_{n,m}(3)$), such that $\psi_n^{-1}S^{(k)}\psi_n$ is anti-symmetric (its transpose is its negative).
	\end{lemma}
	\begin{proof}
	    In Lemma \ref{lem:iso_reps_WB}, we showed that representations $T$ and $\tilde{T}$ of $\BB_{n,m}(3)$ are isomorphic. In particular, since $r=3$ odd, we used the Lemma B.1 of \cite{Ryan21}. In that Lemma, one found that a valid isomorphism $\psi_n$ was given by $\psi_n=\psi^{\otimes n}$, where $\psi$ is a $3\times3$ (unitary) matrix
		\be
	    \psi
	    =
		\begin{bmatrix}
		\tfrac{1}{\sqrt{2}} & 0 & \tfrac{i}{\sqrt{2}}\\
		0 & 1 & 0 \\
		\tfrac{-1}{\sqrt{2}} & 0 & \tfrac{i}{\sqrt{2}} \\
		\end{bmatrix},
		\ee
		where $i=\sqrt{-1}$. One then can verify the required identities directly, using the explicit spin matrices
		\be
	    S^{(1)}
	    =
	    \frac{1}{\sqrt{2}}
		\begin{bmatrix}
		0 & 1 & 0\\
		1 & 0 & 1 \\
		0 & 1 & 0 \\
		\end{bmatrix},
		\
		S^{(2)}
	    =
	    \frac{1}{i\sqrt{2}}
		\begin{bmatrix}
		0 & 1 & 0\\
		-1 & 0 & 1 \\
		0 & -1 & 0 \\
		\end{bmatrix},
		\
		S^{(3)}
	    =
		\begin{bmatrix}
		1 & 0 & 0\\
		0 & 0 & 0 \\
		0 & 0 & -1 \\
		\end{bmatrix}.
		\ee
	\end{proof}

\end{document}